\documentclass[journal]{IEEEtran}

\usepackage[utf8]{inputenc}
\usepackage[T1]{fontenc}
\usepackage{url}
\usepackage{ifthen}
\usepackage{cite}
\usepackage[cmex10]{amsmath}

\usepackage{graphicx}
\usepackage{amsfonts}
\usepackage{amsmath}
\usepackage{amssymb}
\usepackage{mathrsfs} 
\usepackage{color}
\usepackage{bm}
\usepackage{latexsym} 

\newtheorem{thm}{Theorem}

\newtheorem{lem}{Lemma}
\newtheorem{proof}{proof}

\newtheorem{rem}{Remark}
\newtheorem{exam}{Example}

\ifCLASSINFOpdf \else \fi \onecolumn \renewcommand{\baselinestretch}{1.3}

\begin{document}
\title{On the Sequence Reconstruction Problem for the Single-Deletion
Two-Substitution Channel}

\author{ Wentu~Song,
         ~Kui~Cai,~\IEEEmembership{Senior Member,~IEEE},
         and Tony Q. S. Quek,~\IEEEmembership{Fellow,~IEEE}
\thanks{Wentu~Song and Kui~Cai are with the Science, Mathematics and Technology
        Cluster, Singapore University of Technology and Design,
        Singapore 487372 (e-mail:
        \{wentu\_song, cai\_kui\}@sutd.edu.sg).
        \emph{Corresponding author: Kui Cai.}}
\thanks{Tony Q. S. Quek is with the Information Systems Technology and Design
        Pillar, Singapore University of Technology and Design,
        Singapore 487372 (e-mail:
        tonyquek@sutd.edu.sg).}
}

\maketitle

\begin{abstract}
The Levenshtein sequence reconstruction problem studies the
reconstruction of a transmitted sequence from multiple erroneous
copies of it. A fundamental question in this field is to determine
the minimum number of erroneous copies required to guarantee
correct reconstruction of the original sequence. This problem is
equivalent to determining the maximum possible intersection size
of two error balls associated with the underlying channel.
Existing research on the sequence reconstruction problem has
largely focused on channels with a single type of error, such as
insertions, deletions, or substitutions alone. However, relatively
little is known for channels that involve a mixture of error
types, for instance, channels allowing both deletions and
substitutions. In this work, we study the sequence reconstruction
problem for the single-deletion two-substitution channel, which
allows one deletion and at most two substitutions applied to the
transmitted sequence. Specifically, we prove that if two $q$-ary
length-$n$ sequences have the Hamming distance $d\geq 2$, where
$q\geq 2$ is any fixed integer, then the intersection size of
their error balls under the single-deletion two-substitution
channel is upper bounded by $(q^2-1)n^2-(3q^2+5q-5)n+O_q(1)$,
where $O_q(1)$ is a constant independent of $n$ but dependent on
$q$. Moreover, we show that this upper bound is tight up to an
additive constant.
\end{abstract}

\begin{IEEEkeywords}
Sequence reconstruction, reconstruction codes, deletion,
substitution.
\end{IEEEkeywords}

\IEEEpeerreviewmaketitle

\section{Introduction}

The sequence reconstruction problem studies the reconstruction of
a sequence from multiple erroneous copies of it (also referred to
as \emph{reads} in data storage applications), which are obtained
by transmitting the original sequence over a number of identical
channels. A fundamental question in this field is to determine the
minimum number of reads (equivalently, the minimum number of
channels) required to guarantee correct reconstruction of the
transmitted sequence. Typically, the transmitted sequence is taken
from a pre-specified code $\mathscr C$ and the minimum number of
reads required for correct reconstruction is equal to one plus the
maximum intersection size between the error balls of any two
distinct codewords of $\mathscr C$. This quantity is referred to
as the \emph{read coverage} of $\mathscr C$ for the underlying
channel. Therefore, determining the read coverage of a given code
$\mathscr C$, or designing codes with prescribed read coverage
(referred to as \emph{reconstruction codes}), is a central
research topic in the study of sequence reconstruction.

This problem was first formulated and systematically studied by
Levenshtein in his seminal works
\cite{Levenshtein01-IT,Levenshtein01-CTA}, where the fundamental
bounds on reconstruction were established for several types of
channels such as the insertion channel, the deletion channel, and
the substitution channel. Following Levenshtein's formulation, a
significant body of research has focused on sequence
reconstruction under channels with a single type of
synchronization error. For deletion channels, Gabrys and Yaakobi
\cite{Gabrys18-SR} studied reconstruction from multiple
deletion-corrupted reads and derived bounds on the read coverage
of single-deletion-correcting codes. This line of work was further
developed by Pham \emph{et al.} \cite{Pham22}, who provided a
complete asymptotic solution for sequence reconstruction over
$t$-deletion channels, assuming that the code $\mathscr C$ is an
$(\ell-1)$-deletion correcting code for some positive integer
$\ell\leq t$. Related problems involving insertions were
investigated in \cite{Sala17-IT}, where exact reconstruction from
insertions in synchronization codes was considered. Beyond
deriving the bounds of read coverage, several works studied the
design of reconstruction codes. Reconstruction codes correcting
two deletions were constructed in \cite{Chrisnata22-IT,
Y-Sun-23-SRCC} and reconstruction codes correcting two insertions
were constructed in \cite{Z-L-G-23}. Constructions of
reconstruction codes for single-edit error were developed in
\cite{KuiCai22, Chee18}, where an edit error means a deletion, an
insertion or a substitution error. Extensions to burst errors were
explored in \cite{Y-Sun-23}, and more recently, reconstruction
under multiple bursts of insertions or deletions was investigated
in \cite{Lan-Sun-Yu-Ge-25}. The relationship between
reconstruction under insertion and deletion channels was further
clarified in \cite{Sun-Ge-25}.

In contrast to channels with a single error type, sequence
reconstruction under mixed-error channels is considerably more
challenging and remains much less understood. Abu-Sini and Yaakobi
\cite{Abu-Sini21} studied Levenshtein's reconstruction problem for
single-insertion single-substitution channel, and computed the
size of the error balls for the single-deletion
multiple-substitution channel. The size and structure of error
balls for channels with at most three edits were further examined
in \cite{Abbasian24}. Sequence reconstruction for single-deletion
single-substitution channel was investigated in
\cite{Wentu25-SeqR}, and a tight bound on the intersection size of
error balls was derived. Reconstruction codes for correcting one
deletion and one substitution were proposed in
\cite{Li-Sun-Ge-25}. A comprehensive overview of sequence
reconstruction problems and their applications, particularly in
DNA-based storage systems, can be found in \cite{Yang25-JSAIT}.

In this work, we study the sequence reconstruction problem for the
$q$-ary $(q\geq 2)$ single-deletion two-substitution channel,
which allows one deletion and at most two substitutions applied to
the transmitted sequence. Codes correcting a mixture of deletions
and substitutions have been studied in several existing works
under the classical error-correction or list-decoding models
\cite{Smagloy24-IT}$-$\cite{Sima20}. However, to the best of our
knowledge, sequence reconstruction for channels allowing one
deletion and multiple substitutions has not been previously
investigated. Our main contribution is an upper bound on the
intersection size of the error balls of two $q$-ary sequences with
Hamming distance $d\geq 2$. Specifically, we prove that for any
two $q$-ary sequences of length $n$ with a Hamming distance $d\geq
2$, the size of the intersection of their error balls under the
single-deletion two-substitution channel is upper bounded by
$(q^2-1)n^2-(3q^2+5q-5)n+O_q(1)$, where $O_q(1)$ is a constant
independent of $n$ but dependent on $q$.\footnote{In this work, as
we assume that $q$ is fixed, so a constant means a number that
does not depend on $n$ but may depend on $q$. For notational
simplicity, we borrow the big $O$ notation and use $O_q(1)$ to
denote an arbitrary constant, which may be positive or negative.}
Moreover, we show that this upper bound is tight up to an additive
constant, thereby characterizing the intersection size of error
balls for the single-deletion two-substitution channel to within a
constant gap.

We remark that the bound derived in this paper can in fact be made
fully explicit, in the sense that the exact value of the $O_q(1)$
term can be determined. However, deriving this constant would
require a more tedious analysis, resulting in substantially longer
and less transparent proofs, without providing additional insight
into the main asymptotic behavior. For this reason, we express
this term using the $O_q(1)$ notation.

The paper is organized as follows. In Section
\uppercase\expandafter{\romannumeral 2}, we state the problem and
our main result, as well as some simple observations and auxiliary
results. In Section \uppercase\expandafter{\romannumeral 3}, we
prove the main for the case that the Hamming distance between the
two sequences is $d=2$. In Section
\uppercase\expandafter{\romannumeral 4}, we prove the main result
for the case that the Hamming distance between the two sequences
is $d\geq 3$. The paper is concluded in Section
\uppercase\expandafter{\romannumeral 5}.

\section{Preliminaries}

We use $|A|$ to denote the size of a set $A$. For any integers
$m\leq n$, let $[m,n]=\{m,m+1,\ldots,n\}$ and called it an
\emph{interval}. If $m>n$, let $[m,n]=\emptyset$. For simplicity,
we denote $[n]=[1,n]$. Let $\Sigma_q=\{0,1,\cdots,q-1\}$, where
$q\geq 2$ is an integer. For any $\bm x\in\Sigma_q^n$, let $x_i$
denote the $i$th component of $\bm x$ and then $\bm x$ is denoted
as $\bm x=x_1x_2\cdots x_n$ or $\bm x=(x_1,x_2,\cdots, x_n)$. If
$D=\{i_1,i_2,\cdots,i_m\}\subseteq[n]$ such that
$i_1<i_2<\cdots<i_m$, then $x_D=x_{i_1}x_{i_2}\cdots x_{i_m}$ is
called a \emph{subsequence} of $\bm x$. If $D$ is an interval,
$x_D$ is called a \emph{substring} of $\bm x$. A \emph{run} of
$\bm x$ is a maximal substring of $\bm x$ consisting of identical
symbols. For any two sequences $\bm x$ and $\bm y$, let $\bm x\bm
y$ denote the concatenation of $\bm x$ and $\bm y$. For any
$a\in\Sigma_q$, $a^L$ is the sequence consisting of $L$
consecutive $a$.

Given two non-negative integers $t$ and $s$ such that $t+s<n$, for
any $\bm{x}\in\Sigma_q^n$, the \emph{$t$-deletion $s$-substitution
ball} of $\bm x$, denoted by $B^{\text{DS}}_{t,s}(\bm{x})$, is the
set of all sequences that can be obtained from $\bm{x}$ by exactly
$t$ deletions and at most $s$ substitutions. The
\emph{$t$-deletion ball} of $\bm x$ is
$B^{\text{D}}_{t}(\bm{x})\triangleq B^{\text{DS}}_{t,0}(\bm{x})$,
and the \emph{$s$-substitution ball} of $\bm x$ is
$B^{\text{S}}_{s}(\bm{x})\triangleq B^{\text{DS}}_{0,s}(\bm{x})$.
Let $B\in\{B^{\text{D}}_{t}, B^{\text{S}}_{s},
B^{\text{DS}}_{t,s}\}$. For any $\bm x,\bm x'\in\Sigma_q^n$, we
denote $B(\bm x,\bm x')=B(\bm x)\cap B(\bm x')$. A code $\mathscr
C\subseteq\Sigma_q^n$ is said to be an $(n,N,B)$-reconstruction
code if $\nu(\mathscr C;B)<N$, where
$$\nu(\mathscr C; B)\triangleq\max\{|B(\bm x,\bm x')|:
\bm x,\bm x'\in\mathscr C,\bm x\neq\bm x'\}$$ is called the read
coverage of $\mathscr C$ with respect to $B$.

For any $\bm x\in\Sigma_q^n$, we denote $\text{supp}(\bm x)=\{i\in
[n]: x_i\neq 0\}$ and call it the support of $\bm x$. If $\bm
x=x_1x_2\cdots x_n$ and $\bm x'=x'_1x'_2\cdots x'_n$, then
$\text{supp}(\bm x-\bm x')$ is the set of $i\in[n]$ such that
$x_i\neq x'_i$. The Hamming distance between $\bm x$ and $\bm x'$,
denoted by $d_{\text{H}}(\bm x,\bm x')$, is defined as the size of
$\text{supp}(\bm x-\bm x')$, i.e., $d_{\text{H}}(\bm x,\bm
x')=|\text{supp}(\bm x-\bm x')|$. By definition, for any $\bm
y\in\Sigma_q^n$, $\bm y\in B^{\text{S}}_{s}(\bm x)$ if and only if
$|\text{supp}(\bm y-\bm x)|\leq s$.

In this work, we assume that $q\geq 2$ is an arbitrarily fixed
positive integer and we are to estimate $\nu\big(\mathscr C;
B=B^{\text{DS}}_{1,2}\big)$, assuming that $\mathscr C\subseteq
\Sigma_q^n$ is a code with minimum Hamming distance $d\geq 2$. Our
main result is the following theorem.
\begin{thm}\label{main-thm}
Suppose $\bm x,\bm x'\in\Sigma_q^n$ such that $d_{\text{H}}(\bm
x,\bm x')\geq 2$. We have $|B^{\text{DS}}_{1,2}(\bm x, \bm
x')|\leq (q^2-1)n^2-(3q^2+5q-5)n+O_q(1)$, where $O_q(1)$ is a
constant independent of $n$ but dependent on $q$. Moreover, there
exist two sequences $\bm x,\bm x'\in\Sigma_q^n$ such that
$d_{\text{H}}(\bm x,\bm x')=2$ and $|B^{\text{DS}}_{1,2}(\bm x,
\bm x')|=(q^2-1)n^2-(3q^2+5q-5)n+O_q(1)$.
\end{thm}

To prove Theorem \ref{main-thm}, we consider two cases according
to the Hamming distance $d_{\text{H}}(\bm x,\bm x')$, namely,
$d_{\text{H}}(\bm x,\bm x')=2$ and $d_{\text{H}}(\bm x,\bm x')\geq
3$. The proof for the case $d_{\text{H}}(\bm x,\bm x')=2$ will be
given in Section \uppercase\expandafter{\romannumeral 3}, while
the proof for the case $d_{\text{H}}(\bm x,\bm x')\geq 3$ will be
presented in Section \uppercase\expandafter{\romannumeral 4}. In
each case, the analysis is further divided into three subcases.
For each subcase, we show that either $|B^{\text{DS}}_{1,2}(\bm x,
\bm x')|\leq (q^2-1)n^2-(3q^2+5q-5)n+O_q(1)$ or
$|B^{\text{DS}}_{1,2}(\bm x, \bm x')|\leq a_qn^2+b_qn+O_q(1)$ for
some constant $a_q$ and $b_q$ such that $a_q<q^2-1$. Consequently,
for sufficiently large $n$, the latter bound is also upper bounded
by $(q^2-1)n^2-(3q^2+5q-5)n+O_q(1)$. Therefore, in all cases, we
obtain $|B^{\text{DS}}_{1,2}(\bm x, \bm x')|\leq
(q^2-1)n^2-(3q^2+5q-5)n+O_q(1)$.

In the rest of this section, we present some observations and
auxiliary results that will be used in the proof of Theorem
\ref{main-thm}.

\subsection{Intersection of error balls for the substitution channel}

It is easy to show that the size of $s$-substitution balls is
$($e.g., see \cite[Chapter 1]{Huffman}$)$
\begin{align}\label{eq-lmd-0s}|B^{\text{S}}_{s}(\bm u)|
=\xi^{q,n}_{0,s}\triangleq \sum_{k=0}^{s}\binom{n}{k}(q-1)^k,
~\forall\!~\bm u\in\Sigma_q^{n}.\end{align}

Consider the intersection size of substitution balls. Let
\begin{equation}\label{eq-lmd-ds2}
\xi^{2,n}_{d,s}=\!\left\{\!\begin{aligned}
&\sum_{i=0}^{d}\binom{d}{i}\sum_{k=0}^{s-d+\min\{i,d-i\}}
\binom{n-d}{k}, ~~~\!~~1\leq d\leq s,\\
&\sum_{i=d-s}^{s}\binom{d}{i}
\sum_{k=0}^{s-d+\min\{i,d-i\}}\binom{n-d}{k}, ~~s<d\leq 2s.
\end{aligned}\right.
\end{equation}
For $q>2$, let
\begin{equation}\label{eq-lmd-ds3}
\xi^{q,n}_{d,s}=\!\left\{\!\begin{aligned}
&\sum_{i=0}^{d}\sum_{j=0}^{d-i}\binom{d}{i}\binom{d-i}{j}
\eta_{i,j}, ~~~~~\!~~~1\leq d\leq s,\\
&\sum_{i=d-s}^{s}\sum_{j=d-s}^{d-i}\binom{d}{i}\binom{d-i}{j}
\eta_{i,j}, ~~s<d\leq 2s.
\end{aligned}\right.
\end{equation}
where
$\eta_{i,j}=(q-2)^{d-i-j}\sum_{k=0}^{s-d+\min\{i,j\}}\binom{n-d}{k}(q-1)^{k}$.
Then we have the following lemma.
\begin{lem}\label{lem-sub-int-size}
Suppose $\bm u,\bm u'\in\Sigma_q^{n}$ such that the Hamming
distance $d_{\text{H}}(\bm u,\bm u')=d\geq 1$. If $d\leq 2s$, then
$|B^{\text{S}}_{s}(\bm u,\bm u')|=\xi^{q,n}_{d,s}$; if $d\geq
2s+1$, then $|B^{\text{S}}_{s}(\bm u,\bm u')|=0$.
\end{lem}
\begin{proof}
For $q=2$, the results can be obtained from \cite[Lemma
2]{Eitan19}, so we only consider $q>2$.

Let $S=\text{supp}(\bm u-\bm u')$. For any $\bm v\in\Sigma_q^{n}$,
let $D_1=\{\ell\in S: v_\ell=u_\ell\}$, $D_2=\{\ell\in S:
v_\ell=u'_\ell\}$ and $D_3=\{\ell\in [n]\backslash S: v_\ell\neq
u_\ell\}$. Then $D_1,D_2,D_3$ are mutually disjoint. Clearly,
$d_{\text{H}}(\bm v,\bm u)=d-|D_1|+|D_3|$ and $d_{\text{H}}(\bm
v,\bm u')=d-|D_2|+|D_3|$, so $\bm v\in B^{\text{S}}_{s}(\bm u,\bm
u')$ if and only if $d-|D_1|+|D_3|\leq s$ and $d-|D_2|+|D_3|\leq
s$. Therefore, if $\bm v\in B^{\text{S}}_{s}(\bm u,\bm u')$, then
we can obtain: 1) $d-|D_1|\leq s-|D_3|\leq s$; 2) $d-|D_2|\leq
s-|D_3|\leq s$; and 3) $|D_3|\leq
s-\max\{d-|D_1|,d-|D_2|\}=s-d+\min\{|D_1|,|D_2|\}$. We need to
consider the following two cases.

Case 1: $1\leq d\leq s$. Then conditions 1) and 2) are always
satisfied. Note that $D_1\cup D_2\subseteq S$, so we have $0\leq
|D_1|\leq d$ and $0\leq |D_2|\leq d-|D_1|$. Thus, in this case,
from conditions 1)$-$3), we can obtain: 4) $0\leq |D_1|\leq d$; 5)
$0\leq |D_2|\leq d-|D_1|$; and 6) $|D_3|\leq
s-d+\min\{|D_1|,|D_2|\}$. Conversely, it is easy to verify that if
$D_1,D_2,D_3$ satisfy conditions 4)$-$6), then $\bm v\in
B^{\text{S}}_{s}(\bm u,\bm u')$. Let $i=|D_1|$, $j=|D_2|$ and
$k=|D_3|$. Then we can obtain $|B^{\text{S}}_{s}(\bm u,\bm
u')|=\xi^{q,n}_{d,s}$.

Case 2: $s+1\leq d\leq 2s$. From conditions 1) and 2), we have
$d-s\leq|D_i|$, $i=1,2$. Note that $D_1\cup D_2\subseteq S$, so
$|D_1|+|D_2|\leq d$, which implies $|D_1|\leq d-|D_2|\leq s$ and
$|D_2|\leq d-|D_1|$. Thus, in this case, from conditions 1)$-$3),
we can obtain: 7) $d-s\leq |D_1|\leq s$; 85) $d-s\leq |D_2|\leq
d-|D_1|$; and 9) $|D_3|\leq s-d+\min\{|D_1|,|D_2|\}$. Conversely,
it is easy to verify that if $D_1,D_2,D_3$ satisfy conditions
7)$-$9), then $\bm v\in B^{\text{S}}_{s}(\bm u,\bm u')$. Let
$i=|D_1|$, $j=|D_2|$ and $k=|D_3|$. Then we can also obtain
$|B^{\text{S}}_{s}(\bm u,\bm u')|=\xi^{q,n}_{d,s}$.
\end{proof}


In this paper, we consider $s=2$. By \eqref{eq-lmd-0s}, we can
obtain
\begin{align}\label{xxx-SBN}
\xi^{q,n-1}_{0,2}
=\frac{(q-1)^2}{2}n^2-\frac{(3q-5)(q-1)}{2}n+(q^2-3q+3).\end{align}
Moreover, for $q\geq 2$, by \eqref{eq-lmd-ds2} and
\eqref{eq-lmd-ds3}, we can obtain\footnote{When $q=2$, by
\eqref{eq-lmd-ds2}, we can obtain $\xi^{2,n-1}_{d,2}=2(n-1)$ for
$d\in\{1,2\}$ and $\xi^{2,n-1}_{d,2}=6$ for $d\in\{3,4\}$. It is
easy to verify that for $q=2$, $\xi^{2,n-1}_{d,2}$ can also be
expressed as \eqref{xxp-case2}.}
\begin{equation}\label{xxp-case2}
\xi^{q,n-1}_{d,2}=\!\left\{\!\begin{aligned}
&q(q-1)n-2q^2+3q, ~~~\!~~~d=1,\\
&2(q-1)n+q^2-6q+6, \!~~d=2,\\
&6q-6, ~~~~~~~~~~~~~~~~~\!\!~~~~~~~d=3,\\
&6, ~~~~~~~~~~~~~~~~~~~~~\!\!~~~~~~~~\!~~d=4.
\end{aligned}\right.
\end{equation} It is easy to verify that
$\xi^{q,n-1}_{d,2}<\xi^{q,n-1}_{d-1,2}$ for $n>3$ and
$d\in\{1,2,3,4\}$. The following lemma gives a specific
description of the sequences in $B^{\text{S}}_{2}(\bm u,\bm u')$
for any $\bm u,\bm u'\in\Sigma_q^{n-1}$.

\begin{lem}\label{lem-Cap-SE}
Suppose $\bm u,\bm u'\in\Sigma_q^{n-1}$.
\begin{enumerate}
\item [1)] If $\text{supp}(\bm u-\bm u')=\{i_1\}$, then $\bm u=\bm
wu_{i_1}\bm w'$ and $\bm u'=\bm wu'_{i_1}\bm w'$ for some $\bm
w\in\Sigma_q^{i_1-1}$ and $\bm w'\in\Sigma_q^{n-1-i_1}$ and
$u_{i_1}\neq u'_{i_1}$. Then for any $\bm v\in\Sigma_q^{n-1}$,
$\bm v\in B^{\text{S}}_{2}(\bm u,\bm u')$ if and only if it is of
the form $\bm v=[\bm w]\!~v\!~[\bm w']$, where $[\bm
w]\in\Sigma_q^{i_1-1}$ and $[\bm w']\in\Sigma_q^{n-1-i_1}$, such
that $v\in\Sigma_q$ and $[\bm w]\!~[\bm w']$ can be obtained from
$\bm w\bm w'$ by at most one substitution, where $[\bm w]\!~[\bm
w']$ is the concatenation of the two sequences $[\bm w]$ and $[\bm
w']$.

\item [2)] If $\text{supp}(\bm u-\bm u')=\{i_1,i_2\}$ and
$i_1<i_2$, then we can write $\bm u=\bm w_1u_{i_1}\bm
w_2u_{i_2}\bm w_3$ and $\bm u'=\bm w_1u'_{i_1}\bm w_2u'_{i_2}\bm
w_3$, where $u_{i_1}\neq u'_{i_1}$ and $u_{i_2}\neq u'_{i_2}$,
such that $\bm w_1\in\Sigma_q^{i_1-1}$, $\bm
w_2\in\Sigma_q^{i_2-i_1-1}$ and $\bm w_3\in\Sigma_q^{n-1-i_2}$.
For any $\bm v\in\Sigma_q^{n-1}$, $\bm v\in B^{\text{S}}_{2}(\bm
u,\bm u')$ if and only if it is of the form $\bm v=[\bm w_1]v[\bm
w_2]v'\!~[\bm w_3]$, where $[\bm w_1]\in\Sigma_q^{i_1-1}$, $[\bm
w_2]\in\Sigma_q^{i_2-i_1-1}$ and $[\bm w_3]\in\Sigma_q^{n-1-i_2}$,
such that one of the following two conditions hold: (i) $[\bm
w_1]\!~[\bm w_2]\!~[\bm w_3]=\bm w_1\bm w_2\bm w_3$ and
$v\!~v'\in\Sigma_q^2$; (ii) $[\bm w_1]\!~[\bm w_2]\!~[\bm w_3]$ is
obtained from $\bm w_1\bm w_2\bm w_3$ by exactly one substitution
and $vv'\in\{u_{i_1}u'_{i_2}, u'_{i_1}u_{i_2}\}$, where $[\bm
w_1]\!~[\bm w_2]\!~[\bm w_3]$ is the concatenation of $[\bm w_1]$,
$[\bm w_2]$ and $[\bm w_3]$. Clearly, conditions (i) and (ii) are
exclusive.

\item [3)] If  $\text{supp}(\bm u-\bm u')=\{i_1,i_2,i_3\}$ such
that $i_1<i_2<i_3$, then we can write $\bm u=\bm w_1u_{i_1}\bm
w_2u_{i_2}\bm w_3u_{i_3}\bm w_4$ and $\bm u'=\bm w_1u'_{i_1}\bm
w_2u'_{i_2}\bm w_3u'_{i_3}\bm w_4$, where $u_{i_\ell}\neq
u'_{i_\ell}$, $\ell=1,2,3$, such that $\bm
w_1\in\Sigma_q^{i_1-1}$, $\bm w_2\in\Sigma_q^{i_2-i_1-1}$, $\bm
w_3\in\Sigma_q^{i_3-i_2-1}$ and $\bm w_4\in\Sigma_q^{n-1-i_3}$.
For any $\bm v\in\Sigma_q^{n-1}$, $\bm v\in B^{\text{S}}_{2}(\bm
u,\bm u')$ if and only if it is of the form $\bm v=\bm
w_1v_{i_1}\bm w_2v_{i_2}\bm w_3v_{i_3}\bm w_4$ such that
$(v_{i_1},v_{i_2},v_{i_3})\in \{(\Sigma_q,u_{i_2},u'_{i_3})$,
$(\Sigma_q,u'_{i_2},u_{i_3})$, $(u_{i_1},\Sigma_q,u'_{i_3})$,
$(u'_{i_1},\Sigma_q,u_{i_3})$, $(u_{i_1},u'_{i_2},\Sigma_q)$,
$(u'_{i_1},u_{i_2},\Sigma_q)\}$, where
$(\Sigma_q,u_{i_2},u'_{i_3})=\{(v,u_{i_2},u'_{i_3}):v\in\Sigma_q\}$
and the other notations are defined similarly.

\item [4)] If  $\text{supp}(\bm u-\bm u')=\{i_1,i_2,i_3,i_4\}$
such that $i_1<i_2<i_3<i_4$, then we can write $\bm u=\bm
w_1u_{i_1}\bm w_2u_{i_2}\bm w_3u_{i_3}\bm w_4u_{i_4}\bm w_5$ and
$\bm u'=\bm w_1u'_{i_1}\bm w_2u'_{i_2}\bm w_3u'_{i_3}\bm
w_4u'_{i_4}\bm w_5$, where $u_{i_\ell}\neq u'_{i_\ell}$,
$\ell=1,2,3,4$, such that $\bm w_1\in\Sigma_q^{i_1-1}$, $\bm
w_2\in\Sigma_q^{i_2-i_1-1}$, $\bm w_3\in\Sigma_q^{i_3-i_2-1}$,
$\bm w_4\in\Sigma_q^{i_4-i_3-1}$ and $\bm
w_5\in\Sigma_q^{n-1-i_4}$. For any $\bm v\in\Sigma_q^{n-1}$, $\bm
v\in B^{\text{S}}_{2}(\bm u,\bm u')$ if and only if it is of the
form $\bm v=\bm w_1v_{i_1}\bm w_2v_{i_2}\bm w_3v_{i_3}\bm
w_4v_{i_4}\bm w_5$ such that $v_{i_1}v_{i_2}v_{i_3}v_{i_4}\in
\{u_{i_1}u_{i_2}u'_{i_3}u'_{i_4}$,
$u_{i_1}u'_{i_2}u_{i_3}u'_{i_4}$,
$u_{i_1}u'_{i_2}u'_{i_3}u_{i_4}$,
$u'_{i_1}u_{i_2}u_{i_3}u'_{i_4}$,
$u'_{i_1}u_{i_2}u'_{i_3}u_{i_4}$,
$u'_{i_1}u'_{i_2}u_{i_3}u_{i_4}\}$.
\end{enumerate}
\end{lem}
\begin{proof}
Note that $\bm v\in B^{\text{S}}_{2}(\bm u,\bm u')$ if and only if
$d_{\text{H}}(\bm v,\bm u)\leq 2$ and $d_{\text{H}}(\bm v,\bm
u')\leq 2$. The four claims are easily verified.
\end{proof}

\subsection{Some observations and more auxiliary results}

\emph{Observation 1}: Given any $\bm x,\bm x'\in\Sigma_q^n$, to
simplify the notations, for any $j,j'\in[n]$, we denote
$E^{j}_{j'}=B^{\text{S}}_{2}\big(x_{[n]\backslash\{j\}},
x'_{[n]\backslash\{j'\}}\big)$. Moreover, let
$E^{j}=\bigcup_{j'\in[n]}E^{j}_{j'}$. Then we have
\begin{align}\label{eq1-Obsv-1}
B^{\text{DS}}_{1,s}(\bm x,\bm x')
=\bigcup_{j,j'\in[n]}B^{\text{S}}_{s}\big(x_{[n]\backslash\{j\}},
x'_{[n]\backslash\{j'\}}\big)=\bigcup_{j\in[n]}E^{j}.\end{align}
We can simplify \eqref{eq1-Obsv-1} as follows. Let $I_i,
i=1,2,\cdots,m,$ be mutually disjoint intervals such that
$\bigcup_{i=1}^mI_i=[n]$ and each $x_{I_i}$ is contained in a run
of $\bm x$. Let $I'_{i'}, i'=1,2,\cdots,m',$ be mutually disjoint
intervals such that $\bigcup_{i'=1}^{m'}I'_{i'}=[n]$ and each
$x_{I'_{i'}}$ is contained in a run of $\bm x'$. For each $i\in
[m]$ and $i'\in[m']$, arbitrarily fix a $j_i\in I_i$ and a
$j'_{i'}\in I'_{i'}$. Then we have $E^{j_i}
=\bigcup_{i'=1}^{m'}E^{j_i}_{j'_{i'}}=\bigcup_{i'=1}^{m'}
B^{\text{S}}_{s}\big(x_{[n]\backslash\{j_i\}},
x'_{[n]\backslash\{j'_{i'}\}}\big)$ and
\begin{align}\label{eq2-Obsv-1}
B^{\text{DS}}_{1,s}(\bm x,\bm x')=\bigcup_{i=1}^m
E^{j_i}=\bigcup_{i=1}^m\bigcup_{i'=1}^{m'}E^{j_i}_{j'_{i'}}.
\end{align}

\emph{Observation 2}: Given any $\bm x,\bm x'\in\Sigma_q^n$, for
any $j,j'\in[n]$, if $j\leq j'$, we have
\begin{align*}
d_{\text{H}}(x_{[n]\backslash\{j\}},
x'_{[n]\backslash\{j'\}}\big)=\big|\text{supp}(\bm x-\bm
x')\backslash[j,j']\big|+\Delta_{[j,j']}\end{align*} where
$\Delta_{[j,j']}=|\{i\in[j,j'-1]: x'_{i}\neq x_{i+1}\}|$ is the
number of $i\in[j,j'-1]$ such that $x'_{i}\neq x_{i+1}$; if
$j>j'$, we have
\begin{align*}d_{\text{H}}(x_{[n]\backslash\{j\}},
x'_{[n]\backslash\{j'\}}\big)=\big|\text{supp}(\bm x-\bm
x')\backslash[j',j]\big|+\Delta'_{[j',j]}\end{align*} where
$\Delta'_{[j',j]}=|\{i\in[j',j-1]: x_{i}\neq x'_{i+1}\}|$.

\begin{rem}\label{rem-spc-Obv2}
As a special case of Observation 2, for $j,j'\in[n]$ such that
$x_{j}$ is in the $\lambda$th run of $\bm x$ and $x_{j'}$ is in
the $\lambda'$th run of $\bm x$, then
$d_{\text{H}}(x_{[n]\backslash\{j\}},
x_{[n]\backslash\{j'\}}\big)=|\lambda'-\lambda|$. This simple fact
will also be used in our discussions.
\end{rem}


\begin{lem}\label{lem-dtn-jpj-run}
Suppose $\bm x,\bm x'\in\Sigma_q^n$ and $\text{supp}(\bm x-\bm
x')=\{i_1,i_2,\cdots,i_d\}$ such that $d\geq 2$ and
$i_1<i_2<\cdots<i_d$. Let $\{j_1,j_2,\cdots,j_m\}\subseteq [n]$ be
such that $j_{i-1}<j_i$ and $x_{[j_{i-1}+1,j_i]}$, $i=1,2,\cdots,
m$, are all runs of $\bm x$, and
$\{j'_1,j'_2,\cdots,j'_{m'}\}\subseteq [n]$ be such that
$j'_{i'-1}<j'_{i'}$ and $x'_{[j'_{i'-1}+1,j'_{i'}]}$,
$i'=1,2,\cdots, m'$, are all runs of $\bm x'$, where we set
$j_0=j'_0=0$. Moreover, we set $i_0=0$ and $i_{d+1}=n+1$. The
following statements hold:
\begin{enumerate}
 \item[1)] For any $d'<d$, there is no pair
 $(j_{i},j'_{i'})$ such that $\text{supp}(\bm x-\bm x')\cap
 \big([j_{i},j'_{i'}]\cup[j'_{i'},j_{i}]\big)=\emptyset$
 and $d_{\text{H}}\big(x_{[n]\backslash\{j_{i}\}},
 x'_{[n]\backslash\{j'_{i'}\}}\big)=d'$.
 \item[2)] For any $j_{i},j'_{i'}$ such that $\text{supp}(\bm x-\bm x')\cap
 \big([j_{i},j'_{i'}]\cup[j'_{i'},j_{i}]\big)=\emptyset$,
 $d_{\text{H}}\big(x_{[n]\backslash\{j_{i}\}},
 x'_{[n]\backslash\{j'_{i'}\}}\big)=d$ if and only if
 $j_{i}=j'_{i'}$. Hence, there are at most $n-d$ pairs
 $(j_{i},j'_{i'})$ such that $\text{supp}(\bm x-\bm x')\cap
 \big([j_{i},j'_{i'}]\cup[j'_{i'},j_{i}]\big)=\emptyset$
 and $d_{\text{H}}\big(x_{[n]\backslash\{j_{i}\}},
 x'_{[n]\backslash\{j'_{i'}\}}\big)=d$.
 \item[3)] For any $d'>d$, there are at most $n-d$ pairs $(j_{i},j'_{i'})$
 such that $j_{i}<j'_{i'}$, $\text{supp}(\bm x-\bm x')\cap
 \big[j_{i},j'_{i'}]=\emptyset$ and
 $d_{\text{H}}\big(x_{[n]\backslash\{j_{i}\}},
 x'_{[n]\backslash\{j'_{i'}\}}\big)=d'$;
 there are at most $n-d$ pairs $(j_{i},j'_{i'})$
 such that $j_{i}>j'_{i'}$, $\text{supp}(\bm x-\bm x')\cap
 [j'_{i'},j_{i}]=\emptyset$ and
 $d_{\text{H}}\big(x_{[n]\backslash\{j_{i}\}},
 x'_{[n]\backslash\{j'_{i'}\}}\big)=d'$.
 \item[4)] For any $\lambda,\lambda'\in[d]$ such that
 $\lambda\leq\lambda'$, let $S_{[i_\lambda,i_{\lambda'}]}=
 \left|\text{supp}(\bm x-\bm x')\backslash[i_\lambda,i_{\lambda'}]\right|$. Then for any
 $d'<S_{[i_\lambda,i_{\lambda'}]}+\Delta_{[i_\lambda,i_{\lambda'}]}$,
 there is no pair $(j_{i},j'_{i'})$ such that $i_{\lambda-1}<j_{i}\leq i_{\lambda}
 \leq i_{\lambda'}\leq j'_{i'}<i_{\lambda'+1}$ and $d_{\text{H}}(x_{[n]\backslash\{j_i\}},
 x'_{[n]\backslash\{j'_{i'}\}}\big)=d'$; for any
 $d'\geq S_{[i_\lambda,i_{\lambda'}]}+\Delta_{[i_\lambda,i_{\lambda'}]}$,
 there are at most $d'-S_{[i_\lambda,i_{\lambda'}]}-\Delta_{[i_\lambda,i_{\lambda'}]}+1$ pairs
 $(j_{i},j'_{i'})$ such that $i_{\lambda-1}<j_{i}\leq
 i_{\lambda}\leq i_{\lambda'}\leq j'_{i'}<i_{\lambda'+1}$
 and $d_{\text{H}}(x_{[n]\backslash\{j_i\}},
 x'_{[n]\backslash\{j'_{i'}\}}\big)=d'$, where $\Delta_{[i_\lambda,i_{\lambda'}]}$ is defined as in Observation 2.
 \item[5)] Let $\lambda,\lambda'$ and $S_{[i_\lambda,i_{\lambda'}]}$ be as in 4).
 Then for any $d'<S_{[i_\lambda,i_{\lambda'}]}+\Delta'_{[i_\lambda,i_{\lambda'}]}$,
 there is no pair $(j_{i},j'_{i'})$ such that $i_{\lambda-1}<j'_{i'}\leq i_{\lambda}
 \leq i_{\lambda'}\leq j_{i}<i_{\lambda'+1}$ and $d_{\text{H}}(x_{[n]\backslash\{j_i\}},
 x'_{[n]\backslash\{j'_{i'}\}}\big)=d'$; for any
 $d'\geq S_{[i_\lambda,i_{\lambda'}]}+\Delta'_{[i_\lambda,i_{\lambda'}]}$, there are at most
 $d'-S_{[i_\lambda,i_{\lambda'}]}-\Delta'_{[i_\lambda,i_{\lambda'}]}+1$ pairs
 $(j_{i},j'_{i'})$ such that $i_{\lambda-1}<j'_{i'}\leq i_{\lambda}
 \leq i_{\lambda'}\leq j_{i}<i_{\lambda'+1}$
 and $d_{\text{H}}(x_{[n]\backslash\{j_i\}},
 x'_{[n]\backslash\{j'_{i'}\}}\big)=d'$, where
 $\Delta'_{[i_\lambda,i_{\lambda'}]}$ is defined as in Observation 2.
\end{enumerate}
\end{lem}
\begin{proof}
Note that for $j_{i}\leq j'_{i'}$, $\text{supp}(\bm x-\bm x')\cap
\big([j_{i},j'_{i'}]\cup[j'_{i'},j_{i}]\big)=\emptyset$ if and
only if $\text{supp}(\bm x-\bm x')\cap [j_{i},j'_{i'}]=\emptyset$;
for $j_{i}>j'_{i'}$, $\text{supp}(\bm x-\bm x')\cap
\big([j_{i},j'_{i'}]\cup[j'_{i'},j_{i}]\big)=\emptyset$ if and
only if $\text{supp}(\bm x-\bm x')\cap [j'_{i'},j_{i}]=\emptyset$.
Since $\text{supp}(\bm x-\bm x')=\{i_1,i_2,\cdots,i_d\}$, then for
both cases, $\text{supp}(\bm x-\bm x')\cap
\big([j_{i},j'_{i'}]\cup[j'_{i'},j_{i}]\big)=\emptyset$ if and
only if $j_{i},j'_{i'}\in[i_{\lambda-1}+1,i_\lambda-1]$ for some
$\lambda\in[d+1]$.

1) If $j_{i}\leq j'_{i'}$ and $\text{supp}(\bm x-\bm x')\cap
\big([j_{i},j'_{i'}]\cup[j'_{i'},j_{i}]\big)=\emptyset$, we have
$|\text{supp}(\bm x-\bm x')\backslash[j_{i},j'_{i'}]|=d$, by
Observation 2, $d_{\text{H}}\big(x_{[n]\backslash\{j_i\}},
x'_{[n]\backslash\{j'_{i'}\}}\big)\geq |\text{supp}(\bm x-\bm
x')\backslash[j_{i},j'_{i'}]|=d$. Hence, for any $d'<d$, there is
no pair $(j_{i},j'_{i'})$ such that
$d_{\text{H}}\big(x_{[n]\backslash\{j_{i}\}},
x'_{[n]\backslash\{j'_{i'}\}}\big)=d'$. For $j_{i}>j'_{i'}$, the
proof is similar.

2) Since $\text{supp}(\bm x-\bm x')\cap
\big([j_{i},j'_{i'}]\cup[j'_{i'},j_{i}]\big)=\emptyset$, we have
$j_{i},j'_{i'}\in[i_{\lambda-1}+1,i_\lambda-1]$ for some
$\lambda\in[d+1]$. If $j_{i}=j'_{i'}$, by Observation 2, it is
easy to see that $d_{\text{H}}\big(x_{[n]\backslash\{j_{i}\}},
x'_{[n]\backslash\{j'_{i'}\}}\big)=d$. Conversely, suppose
$d_{\text{H}}\big(x_{[n]\backslash\{j_{i}\}},
x'_{[n]\backslash\{j'_{i'}\}}\big)=d$. If $j_{i}<j'_{i'}$, then
$|\text{supp}(\bm x-\bm x')\backslash[j_{i},j'_{i'}]|=d~($noticing
that $j_{i},j'_{i'}\in[i_{\lambda-1}+1,i_\lambda-1])$, so by
Observation 2, we must have $\Delta_{[j_{i},j'_{i'}]}=0$, i.e.,
$x'_{i}=x_{i+1}$ for all $i\in[j_{i},j'_{i'}-1]$. Moreover, by the
definition of $\text{supp}(\bm x-\bm x')$, we have $x_{i}=x'_{i}$
for all $i\in[i_{\lambda-1}+1,i_\lambda-1]$. Therefore, we can
obtain $x_{i}=x'_{i}=x_{i+1}=x'_{i+1}$ for all
$i\in[j_{i},j'_{i'}-1]$, which contradicts to the definition of
$j_i~($i.e., $x_{[j_{i-1}+1,j_{i}]}$ is a run of $\bm x)$. If
$j_{i}>j'_{i'}$, then similarly we can find
$x_{i}=x'_{i}=x_{i+1}=x'_{i+1}$ for all $i\in[j'_{i'},j_{i}-1]$,
which contradicts to the definition of $j_{i'}~($i.e.,
$x_{[j'_{i'-1}+1,j'_{i'}]}$ is a run of $\bm x')$. Thus, we must
have $j_{i}=j'_{i'}$. Note that $j_{i}\notin \text{supp}(\bm x-\bm
x')$, so the number of such pairs is at most $n-d$.

3) It suffices to prove that for each $\lambda\in[d+1]$ and each
$j_i\in[i_{\lambda-1}+1,i_{\lambda}-1]$, there is at most one
$j'_{i'}$ such that $j_{i}<j'_{i'}<i_{\lambda}$ and
$d_{\text{H}}\big(x_{[n]\backslash\{j_{i}\}},
x'_{[n]\backslash\{j'_{i'}\}}\big)=d'$. We can prove this by
contradiction. Suppose $i',i''\in[m']$ such that
$j_i<j'_{i'}<j'_{i''}<i_{\lambda}$ and
$d_{\text{H}}\big(x_{[n]\backslash\{j_{i}\}},
x'_{[n]\backslash\{j'_{i'}\}}\big)=d_{\text{H}}\big(x_{[n]\backslash\{j_{i}\}},
x'_{[n]\backslash\{j'_{i''}\}}\big)=d'$. Clearly, we have
$\big|\text{supp}(\bm x-\bm
x')\backslash[j_i,j'_{i'}]\big|=d=\big|\text{supp}(\bm x-\bm
x')\backslash[j_i,j'_{i''}]\big|$. By Observation 2, we can obtain
$\Delta_{[j_{i}, j'_{i'}]}=\Delta_{[j_{i}, j'_{i''}]}=d'-d$. Then
by the definition of $\Delta_{[j_{i}, j'_{i'}]}$ and
$\Delta_{[j_{i}, j'_{i''}]}$, we must have $x'_{i}=x_{i+1}$ for
all $i\in[j'_{i'},j'_{i''}-1]$. On the other hand, as
$i_{\lambda-1}<j_i<j'_{i'}<j'_{i''}<i_{\lambda}$, we have
$x_{i}=x'_{i}$ for all $i\in[j'_{i'},j'_{i''}]$. Therefore, we can
obtain
$x_{j'_{i'}}=x'_{j'_{i'}}=x_{j'_{i'}+1}=x'_{j'_{i'}+1}=\cdots=x_{j'_{i''}}=x'_{j'_{i''}}$,
which contradicts to the assumption that $x_{[j'_{i''-1},
j'_{i''}]}$ is a run of $\bm x'$. Thus, for each $\lambda\in[d+1]$
and each $j_i\in[i_{\lambda-1}+1,i_{\lambda}-1]$, there is at most
one $j'_{i'}$ such that $j_{i}<j'_{i'}<i_{\lambda}$ and
$d_{\text{H}}\big(x_{[n]\backslash\{j_{i}\}},
x'_{[n]\backslash\{j'_{i'}\}}\big)=d'$. In total, there are at
most $n-d$ pairs $(j_{i},j'_{i'})$ such that $j_{i}<j'_{i'}$,
$\text{supp}(\bm x-\bm x')\cap [j'_{i'},j_{i}]=\emptyset$ and
$d_{\text{H}}\big(x_{[n]\backslash\{j_{i}\}},
x'_{[n]\backslash\{j'_{i'}\}}\big)=d'$. For $j_{i}>j'_{i'}$, the
proof is similar.

4) By the definition of $\text{supp}(\bm x-\bm x')$, we have
$x_{[i_{\lambda-1}+1,i_{\lambda}-1]}=x'_{[i_{\lambda-1}+1,i_{\lambda}-1]}$
and
$x_{[i_{\lambda'}+1,i_{\lambda'+1}-1]}=x'_{[i_{\lambda'}+1,i_{\lambda'+1}-1]}$.
Clearly, for any $j_i\in[i_{\lambda-1}+1,i_{\lambda}]$ and
$j'_{i'}\in[i_{\lambda'},i_{\lambda'+1}-1]$, we can obtain
$|\text{supp}(\bm x-\bm x')\backslash[j_{i},j'_{i'}]|
=|\text{supp}(\bm x-\bm x')\backslash[i_{\lambda},j_{\lambda'}]|
=S_{[i_\lambda,i_{\lambda'}]}$ and
$\Delta_{[j_i,j'_{i'}]}=(\tau-i)+\Delta_{[i_\lambda,i_{\lambda'}]}+(i'-\tau')$,
where $\tau\in[m]$ be such that
$i_{\lambda}\in[j_{\tau-1}+1,j_{\tau}]$ and $\tau'\in[m']$ be such
that $i_{\lambda'}\in[j'_{\tau'-1}+1,j'_{\tau'}]$. By the
assumption, we have $\tau-i\geq 0$ and $i'-\tau'\geq 0$, so by
Observation 2, we have $d_{\text{H}}\big(x_{[n]\backslash\{j_i\}},
x'_{[n]\backslash\{j'_{i'}\}}\big)=|\text{supp}(\bm x-\bm
x')\backslash[j_{i},j'_{i'}]|+\Delta_{[j_i,j'_{i'}]}
=S_{[i_\lambda,i_{\lambda'}]}+(\tau-i)+\Delta_{[i_\lambda,i_{\lambda'}]}+(i'-\tau')
\geq
S_{[i_\lambda,i_{\lambda'}]}+\Delta_{[i_\lambda,i_{\lambda'}]}$.
Therefore, there is no such pair $(j_{i},j'_{i'})$ with
$d_{\text{H}}(x_{[n]\backslash\{j_i\}},
x'_{[n]\backslash\{j'_{i'}\}}\big)<S_{[i_\lambda,i_{\lambda'}]}+\Delta_{[i_\lambda,i_{\lambda'}]}$.
Moreover, for each $d'\geq
S_{[i_\lambda,i_{\lambda'}]}+\Delta_{[i_\lambda,i_{\lambda'}]}$,
if $d_{\text{H}}(x_{[n]\backslash\{j_i\}},
x'_{[n]\backslash\{j'_{i'}\}}\big)=d'$, then we have
$(\tau-i)+(i'-\tau')=d'-(S_{[i_\lambda,i_{\lambda'}]}+\Delta_{[i_\lambda,i_{\lambda'}]})$.
Note that there are at most
$d'-(S_{[i_\lambda,i_{\lambda'}]}+\Delta_{[i_\lambda,i_{\lambda'}]})+1$
pairs $(\tau-i,i'-\tau')$ such that $\tau-i\geq 0$, $i'-\tau'\geq
0$ and
$(\tau-i)+(i'-\tau')=d'-(S_{[i_\lambda,i_{\lambda'}]}+\Delta_{[i_\lambda,i_{\lambda'}]})$.
So, there are at most
$d'-(S_{[i_\lambda,i_{\lambda'}]}+\Delta_{[i_\lambda,i_{\lambda'}]})+1$
pairs $(j_i,j'_{i'})$ such that
$d_{\text{H}}\big(x_{[n]\backslash\{j_i\}},
x'_{[n]\backslash\{j'_{i'}\}}\big)=d'$.

5) The proof is similar to 4).
\end{proof}


\section{Intersection of error balls of sequences with
Hamming distance $d=2$}

In this section, we always assume that $\bm x,\bm x'\in\Sigma_q^n$
such that $\text{supp}(\bm x-\bm x')=\{i_1,i_2\}$ and $i_1<i_2$.
We will prove that $|B^{\text{DS}}_{1,2}(\bm x, \bm x')|\leq
(q^2-1)n^2-(3q^2+5q-5)n+O_q(1)$ and show that this bound is
achievable. Recall that according to the definition in Observation
1, for any $j,j'\in[n]$, we have
$E^{j}_{j'}=B^{\text{S}}_{2}\big(x_{[n]\backslash\{j\}},
x'_{[n]\backslash\{j'\}}\big)$. The following lemma will be used
to simplify the expression of $B^{\text{DS}}_{1,2}(\bm x,\bm x')$.

\begin{lem}\label{lem-d2-smpl}
Let $\{j_1,j_2,\cdots,j_m\}\subseteq [n]$ be such that
$j_{i-1}<j_i$ and $x_{[j_{i-1}+1,j_i]}$, $i=1,2,\cdots, m$, are
all runs of $\bm x$. The following statements hold.
\begin{enumerate}
 \item[1)] If $\ell\in[m]$ such that
 $[j_{\ell-2}+1,j_{\ell}]\cap\{i_1,i_2\}=\emptyset$, then
 $E^{j_{\ell}}_{j_{\ell-1}}\subseteq
 \big(E^{j_{\ell}}_{j_{\ell}}\cup E^{j_{\ell-1}}_{j_{\ell-1}}\big)$; if
 $\ell\in[m]$ such that
 $[j_{\ell-3}+1,j_{\ell}]\cap\{i_1,i_2\}=\emptyset$, then we have
 $E^{j_\ell}_{j_{\ell-2}}\subseteq
 \big(E^{j_{\ell}}_{j_{\ell}}\cup E^{j_{\ell-2}}_{j_{\ell-2}}\big)$.
 \item[2)] If $\ell\in[m]$ such that
 $[j_{\ell-1}+1,j_{\ell+1}]\cap\{i_1,i_2\}=\emptyset$, then
 $E^{j_{\ell}}_{j_{\ell+1}}\subseteq
 \big(E^{j_{\ell}}_{j_{\ell}}\cup E^{j_{\ell+1}}_{j_{\ell+1}}\big)$; if
 $\ell\in[m]$ such that
 $[j_{\ell-1}+1,j_{\ell+2}]\cap\{i_1,i_2\}=\emptyset$, then we have
 $E^{j_\ell}_{j_{\ell+2}}\subseteq
 \big(E^{j_{\ell}}_{j_{\ell}}\cup E^{j_{\ell+2}}_{j_{\ell+2}}\big)$.
\end{enumerate}
\end{lem}
\begin{proof}
We only prove 1), because the proof of 2) is similar.

By the assumption of $\{j_1,j_2,\cdots,j_m\}$, we can assume
$x_{[j_{\ell-2}+1,j_{\ell-1}]}=x_{\ell-1}^{k_{\ell-1}}$ and
$x_{[j_{\ell-1}+1,j_{\ell}]}=x_{\ell}^{k_{\ell}}$. As
$[j_{\ell-2}+1,j_{\ell}]\cap\{i_1,i_2\}=\emptyset$, without loss
of generality, assume $j_\ell<i_1$. Noticing that $\text{supp}(\bm
x-\bm x')=\{i_1,i_2\}$, we can obtain the following
\eqref{eq-lem-d2-cmp2}, \eqref{eq-lem-d2-cmp3},
\eqref{eq-lem-d2-cmp1} and \eqref{eq-lem-d2-cmp4}, where the
corresponding omitted parts across different equations are
identical.\footnote{In this paper, when we write several equations
together for comparison, we will replace some of their identical
parts with ellipses for the sake of brevity.} For any $\bm v\in
E^{j_\ell}_{j_{\ell-1}}
=B^{\text{S}}_{2}\big(x_{[n]\backslash\{j_{\ell}\}},
x'_{[n]\backslash\{j_{\ell-1}\}}\big)$, by \eqref{eq-lem-d2-cmp2},
\eqref{eq-lem-d2-cmp3} and by 3) of Lemma \ref{lem-Cap-SE}, we can
obtain \eqref{eq-lem-d2-cmp5}.
\begin{align}
x_{[n]\backslash\{j_\ell\}}\!~&
=\cdots~\bar{x}_{\ell-1}^{k_{\ell-1}-1}~~\bar{x}_{\ell-1}~~
\bar{x}_{\ell}^{k_{\ell}-1}~\cdots\!~x_{i_1}\!~\cdots
\!~x_{i_2}\!~\cdots
\label{eq-lem-d2-cmp2}\\
x'_{[n]\backslash\{j_{\ell-1}\}}&
=\cdots~\bar{x}_{\ell-1}^{k_{\ell-1}-1}
~~~\bar{x}_{\ell}~~~~\bar{x}_{\ell}^{k_{\ell}-1}~
\cdots\!~x'_{i_1}\!~\cdots\!~x'_{i_2}\!~\cdots
\label{eq-lem-d2-cmp3}\\
\bm v~~& =\cdots~\bar{x}_{\ell-1}^{k_{\ell-1}-1}
~~\!~~v~~\!~~~\bar{x}_{\ell}^{k_{\ell}-1} ~\cdots\!~~u~~\cdots~
u'~\cdots
\label{eq-lem-d2-cmp5}\\
x'_{[n]\backslash\{j_\ell\}}\!~&
=\cdots~\bar{x}_{\ell-1}^{k_{\ell-1}-1}~~
\bar{x}_{\ell-1}~~\bar{x}_{\ell}^{k_{\ell}-1}
~\cdots\!~x'_{i_1}\!~\cdots\!~x'_{i_2}\!~\cdots
\label{eq-lem-d2-cmp1}\\
x_{[n]\backslash\{j_{\ell-1}\}}&
=\cdots~\bar{x}_{\ell-1}^{k_{\ell-1}-1}
~~~\bar{x}_{\ell}~~~~\bar{x}_{\ell}^{k_{\ell}-1}
~\cdots\!~x_{i_1}\!~\cdots\!~x_{i_2}\!~\cdots
\label{eq-lem-d2-cmp4}
\end{align}
where $(v,u,u')\in\{(\Sigma_q,x_{i_1},x'_{i_2})$,
$(\Sigma_q,x'_{i_1},x_{i_2})$,
$(\bar{x}_{\ell-1},\Sigma_q,x'_{i_2})$,
$(\bar{x}_{\ell},\Sigma_q,x_{i_2})$,
$(\bar{x}_{\ell-1},x'_{i_1},\Sigma_q)$,
$(\bar{x}_{\ell},x_{i_1},\Sigma_q)\}$. For
$(v,u,u')\in\{(\Sigma_q,x_{i_1},x'_{i_2})$,
$(\Sigma_q,x'_{i_1},x_{i_2})$,
$(\bar{x}_{\ell-1},\Sigma_q,x'_{i_2})$,
$(\bar{x}_{\ell-1},x'_{i_1},\Sigma_q)\}$, by comparing
\eqref{eq-lem-d2-cmp5} and \eqref{eq-lem-d2-cmp1}, we can find
that $d_{\text{H}}\big(\bm v,
x'_{[n]\backslash\{j_{\ell}\}}\big)\leq 2$, and so $\bm v\in
B^{\text{S}}_{2}\big(x_{[n]\backslash\{j_\ell\}},
x'_{[n]\backslash\{j_{\ell-1}\}}\big)\cap
B^{\text{S}}_{2}\big(x'_{[n]\backslash\{j_\ell\}}\big)\subseteq
B^{\text{S}}_{2}\big(x_{[n]\backslash\{j_\ell\}},
x'_{[n]\backslash\{j_{\ell}\}}\big)=E^{j_\ell}_{j_{\ell}}$; for
$(v,u,u')\in\{(\bar{x}_{\ell},\Sigma_q,x_{i_2})$,
$(\bar{x}_{\ell},x_{i_1},\Sigma_q)\}$, by comparing
\eqref{eq-lem-d2-cmp5} and \eqref{eq-lem-d2-cmp4}, we can find
that $d_{\text{H}}\big(\bm v,
x_{[n]\backslash\{j_{\ell-1}\}}\big)\leq 2$, and so $\bm v\in
B^{\text{S}}_{2}\big(x_{[n]\backslash\{j_\ell\}},
x'_{[n]\backslash\{j_{\ell-1}\}}\big)\cap
B^{\text{S}}_{2}\big(x_{[n]\backslash\{j_{\ell-1}\}}\big)\subseteq
B^{\text{S}}_{2}\big(x_{[n]\backslash\{j_{\ell-1}\}},
x'_{[n]\backslash\{j_{\ell-1}\}}\big)=E^{j_{\ell-1}}_{j_{\ell-1}}$.
Thus, we have $E^{j_\ell}_{j_{\ell-1}}\subseteq
\big(E^{j_\ell}_{j_\ell}\cup E^{j_{\ell-1}}_{j_{\ell-1}}\big)$.

By the assumption of $\{j_1,j_2,\cdots,j_m\}$, we can assume
$x_{[j_{\ell-3}+1,j_{\ell-2}]}=x_{\ell-2}^{k_{\ell-2}}$,
$x_{[j_{\ell-2}+1,j_{\ell-1}]}=x_{\ell-1}^{k_{\ell-1}}$ and
$x_{[j_{\ell-1}+1,j_{\ell}]}=x_{\ell}^{k_{\ell}}$. As
$[j_{\ell-3}+1,j_{\ell}]\cap\{i_1,i_2\}=\emptyset$, without loss
of generality, we can assume $j_\ell<i_1$. Then we can obtain the
following \eqref{eq1-lem-d2-cmp}, \eqref{eq2-lem-d2-cmp},
\eqref{eq4-lem-d2-cmp} and \eqref{eq5-lem-d2-cmp}. For any $\bm
v\in E^{j_\ell}_{j_{\ell-2}}
=B^{\text{S}}_{2}\big(x_{[n]\backslash\{j_\ell\}},
x'_{[n]\backslash\{j_{\ell-2}\}}\big)$, by \eqref{eq1-lem-d2-cmp},
\eqref{eq2-lem-d2-cmp} and by 4) of Lemma \ref{lem-Cap-SE}, we can
obtain \eqref{eq3-lem-d2-cmp}.
\begin{align}
x_{[n]\backslash\{j_\ell\}}\!~&
=\cdots~\bar{x}_{\ell-2}^{k_{\ell-2}-1}~~\bar{x}_{\ell-2}~
~\bar{x}_{\ell-1}^{k_{\ell-1}-1}~~\bar{x}_{\ell-1}~~
\bar{x}_{\ell}^{k_{\ell}-1}
~\cdots\!~x_{i_1}\!~\cdots\!~x_{i_2}\!~\cdots
\label{eq1-lem-d2-cmp}\\
x'_{[n]\backslash\{j_{\ell-2}\}}
&=\cdots~\bar{x}_{\ell-2}^{k_{\ell-2}-1}
~~\bar{x}_{\ell-1}~~\bar{x}_{\ell-1}^{k_{\ell-1}-1}
~~\!~~\bar{x}_{\ell}~\!~~~
\bar{x}_{\ell}^{k_{\ell}-1}~\cdots\!~x'_{i_1}\!~\cdots\!~x'_{i_2}\!~
\cdots\label{eq2-lem-d2-cmp}\\
\bm v~~& =\cdots~\bar{x}_{\ell-2}^{k_{\ell-2}-1}
~~~~v~~~~\bar{x}_{\ell-1}^{k_{\ell-1}-1}~\!~~~v'~~~~
\bar{x}_{\ell}^{k_{\ell}-1}~\cdots\!~~u~~\cdots~ u'~\cdots
\label{eq3-lem-d2-cmp}\\
x'_{[n]\backslash\{j_\ell\}}\!~
&=\cdots~\bar{x}_{\ell-2}^{k_{\ell-2}-1} ~~\bar{x}_{\ell-2}
~~\bar{x}_{\ell-1}^{k_{\ell-1}-1}~~\bar{x}_{\ell-1}~~
\bar{x}_{\ell}^{k_{\ell}-1}~\cdots\!~x'_{i_1}\!~\cdots\!~x'_{i_2}\!~
\cdots\label{eq4-lem-d2-cmp}\\
x_{[n]\backslash\{j_{\ell-2}\}}&
=\cdots~\bar{x}_{\ell-2}^{k_{\ell-2}-1}
~~\bar{x}_{\ell-1}~~\bar{x}_{\ell-1}^{k_{\ell-1}-1}
~\!~~~\bar{x}_{\ell}~\!~~~
\bar{x}_{\ell}^{k_{\ell}-1}~\cdots\!~x_{i_1}\!~\cdots\!~x_{i_2}\!~
\cdots\label{eq5-lem-d2-cmp}
\end{align}
where $(v,v',u,u')\in
\{(\bar{x}_{\ell-2},\bar{x}_{\ell-1},x'_{i_1},x'_{i_2})$,
$(\bar{x}_{\ell-2},\bar{x}_{\ell},x_{i_1},x'_{i_2})$,
$(\bar{x}_{\ell-2},\bar{x}_{\ell},x'_{i_1},x_{i_2})$,
$(\bar{x}_{\ell-1},\bar{x}_{\ell-1},x_{i_1},x'_{i_2})$,
$(\bar{x}_{\ell-1},\bar{x}_{\ell-1},x'_{i_1},x_{i_2})$,
$(\bar{x}_{\ell-1},\bar{x}_{\ell},x_{i_1},x_{i_2})\}$. For
$(v,v',u,u')\in
\{\bar{x}_{\ell-2},\bar{x}_{\ell-1},x'_{i_1},x'_{i_2})$,
$(\bar{x}_{\ell-2},\bar{x}_{\ell},x_{i_1},x'_{i_2})$,
$(\bar{x}_{\ell-2},\bar{x}_{\ell},x'_{i_1},x_{i_2})$,
$(\bar{x}_{\ell-1},\bar{x}_{\ell-1},x_{i_1},x'_{i_2})$,
$(\bar{x}_{\ell-1},\bar{x}_{\ell-1},x'_{i_1},x_{i_2})\}$, by
comparing \eqref{eq3-lem-d2-cmp} and \eqref{eq4-lem-d2-cmp}, we
can find that $d_{\text{H}}\big(\bm v,
x'_{[n]\backslash\{j_{\ell}\}}\big)\leq 2$, and so $\bm v\in
B^{\text{S}}_{2}\big(x_{[n]\backslash\{j_\ell\}},
x'_{[n]\backslash\{j_{\ell-2}\}}\big)\cap
B^{\text{S}}_{2}\big(x'_{[n]\backslash\{j_\ell\}}\big)\subseteq
B^{\text{S}}_{2}\big(x_{[n]\backslash\{j_\ell\}},
x'_{[n]\backslash\{j_{\ell}\}}\big)=E^{j_\ell}_{j_{\ell}}$; for
$(v,v',u,u')=(\bar{x}_{\ell-1},\bar{x}_{\ell},x_{i_1},x_{i_2})$,
by comparing \eqref{eq3-lem-d2-cmp} and \eqref{eq5-lem-d2-cmp}, we
can find that $d_{\text{H}}\big(\bm v,
x_{[n]\backslash\{j_{\ell-2}\}}\big)=0$, and so $\bm v\in
B^{\text{S}}_{2}\big(x_{[n]\backslash\{j_\ell\}},
x'_{[n]\backslash\{j_{\ell-2}\}}\big)\cap
B^{\text{S}}_{2}\big(x_{[n]\backslash\{j_{\ell-2}\}}\big)\subseteq
B^{\text{S}}_{2}\big(x_{[n]\backslash\{j_{\ell-2}\}},
x'_{[n]\backslash\{j_{\ell-2}\}}\big)=E^{j_{\ell-2}}_{j_{\ell-2}}$.
Thus, we have $E^{j_\ell}_{j_{\ell-2}}\subseteq
\big(E^{j_\ell}_{j_\ell}\cup E^{j_{\ell-2}}_{j_{\ell-2}}\big)$.
\end{proof}

\begin{rem}
In 1) of Lemma \ref{lem-d2-smpl}, as
$x_{[j_{\ell-3}+1,j_{\ell-2}]}$, $x_{[j_{\ell-2}+1,j_{\ell-1}]}$
and $x_{[j_{\ell-1}+1,j_{\ell}]}$ are three runs of $\bm x$, so
we can obtain $E^{j}_{j'}=E^{j_\ell}_{j_{\ell-1}}\subseteq
E^{j_\ell}_{j_{\ell}}\cup
E^{j_{\ell-1}}_{j_{\ell-1}}=E^{j}_{j}\cup E^{j'}_{j'}$ for any
$j\in[j_{\ell-1}+1,j_{\ell}]$ and
$j'\in[j_{\ell-2}+1,j_{\ell-1}]$, and
$E^{j}_{j'}=E^{j_\ell}_{j_{\ell-2}}\subseteq
E^{j_\ell}_{j_{\ell}}\cup
E^{j_{\ell-2}}_{j_{\ell-2}}=E^{j}_{j}\cup E^{j'}_{j'}$ for any
$j\in[j_{\ell-1}+1,j_{\ell}]$ and
$j'\in[j_{\ell-3}+1,j_{\ell-2}]$. Similarly, in 2) of Lemma
\ref{lem-d2-smpl}, we can obtain $E^{j}_{j'}\subseteq
E^{j}_{j}\cup E^{j'}_{j'}$ for any $j\in[j_{\ell-1}+1,j_{\ell}]$
and $j'\in[j_{\ell}+1,j_{\ell+1}]~($Resp.
$j'\in[j_{\ell+1}+1,j_{\ell+2}])$.
\end{rem}

\subsection{For the case that both $x_{[n]\backslash\{i_1\}}
=x'_{[n]\backslash\{i_2\}}$ and $x_{[n]\backslash\{i_2\}}
=x'_{[n]\backslash\{i_1\}}$}

In this subsection, we consider the case that both
$x_{[n]\backslash\{i_1\}}=x'_{[n]\backslash\{i_2\}}$ and
$x_{[n]\backslash\{i_2\}}=x'_{[n]\backslash\{i_1\}}$. For this
case, we must have $i_2=i_1+1$, $x_{i_1}=x'_{i_2}$ and
$x_{i_2}=x'_{i_1}$. In fact, from the assumption
$x_{[n]\backslash\{i_1\}}=x'_{[n]\backslash\{i_2\}}$, we can find
$x'_{i_1}=x_{i_1+1}$, and from
$x_{[n]\backslash\{i_2\}}=x'_{[n]\backslash\{i_1\}}$, we can find
$x_{i_1}=x'_{i_1+1}$. Hence, we have $x_{i_1+1}\neq x'_{i_1+1}$.
This is because otherwise, we can obtain
$x'_{i_1}=x_{i_1+1}=x'_{i_1+1}=x_{i_1}$, which contradicts to the
assumption that $i_1\in\text{supp}(\bm x-\bm x')$. As
$\text{supp}(\bm x-\bm x')=\{i_1,i_2\}$, we have $i_2=i_1+1$,
$x_{i_1}=x'_{i_2}$ and $x_{i_2}=x'_{i_1}$. We state this fact as
the following remark for emphasis.

\begin{rem}\label{rem-both-i1-i2}
If both $x_{[n]\backslash\{i_1\}}=x'_{[n]\backslash\{i_2\}}$ and
$x_{[n]\backslash\{i_2\}}=x'_{[n]\backslash\{i_1\}}$, then
$i_2=i_1+1$, $x_{i_1}=x'_{i_2}\neq x_{i_2}=x'_{i_1}$. Hence, $\bm
x =\bm wab\bm w'$ and $\bm x'=\bm wba\bm w'$, where $\bm
w=x_{[1,i_1-1]}=x'_{[1,i_1-1]}, a=x_{i_1}=x'_{i_2},
b=x'_{i_1}=x_{i_2}$ and $\bm w'=x_{[i_2+1,n]}=x'_{[i_2+1,n]}$.
\end{rem}

When using Observation 1 to analyze $B^{\text{DS}}_{1,2}(\bm x,\bm
x')$, one needs to split $\bm x$ and $\bm x'$ into runs. The
following remark gives two forms of $\bm x$ and $\bm x'$ when
split into runs.

\begin{rem}\label{rem-form-run}
Let $\{j_1,j_2,\cdots,j_m\}\subseteq [n]$ be such that
$j_{i-1}<j_i$ and $x_{[j_{i-1}+1,j_i]}$, $i=1,2,\cdots, m$, are
all runs of $\bm x$, and $\{j'_1,j'_2,\cdots,j'_{m'}\}\subseteq
[n]$ be such that $j'_{i'-1}<j'_{i'}$ and
$x'_{[j'_{i'-1}+1,j'_{i'}]}$, $i'=1,2,\cdots, m'$, are all runs of
$\bm x'$, where we set $j_0=j'_0=0$. Then we can denote
$x_{[j_{i-1}+1,j_i]}=\bar{x}_i^{k_i}$ for each $i\in[m]$ and
$x'_{[j'_{i'-1}+1,j'_{i'}]}=\tilde{x}_{i'}^{k'_{i'}}$ for each
$i'\in[m']$, where $k_i=j_i-j_{i-1}$ and
$k'_{i'}=j'_{i'}-j'_{i'-1}$. Let $\tau\in[m]$ and $\tau'\in[m']$
be such that $i_1\in [j_{\tau-1}+1,j_{\tau}]\cap
[j'_{\tau'-1}+1,j'_{\tau'}]$. In other words, $x_{i_1}$ is in the
$\tau$th run of $\bm x$ and $x'_{i_1}$ is in the $\tau'$th run of
$\bm x'$. Then by Remark \ref{rem-both-i1-i2}, we have $i_2\in
[j_{\tau}+1,j_{\tau+1}]\cap [j'_{\tau'}+1,j'_{\tau'+1}]$ and
$$B^{\text{DS}}_{1,2}(\bm x,\bm
x')=\big(\bigcup_{i=1}^{\tau-1}E^{j_i}\big) \cup E^{i_1}\cup
E^{i_2}\cup\big(\bigcup_{i=\tau+2}^{m}E^{j_i}\big).$$ Also by
Remark \ref{rem-both-i1-i2}, we have
\begin{align}\label{eq-seqx-form}\bm
x=\bar{x}_1^{k_1}~\cdots~\bar{x}_{\tau-1}^{k_{\tau-1}}~a^{\sigma_1}
~a~b~b^{\sigma_2}
~\bar{x}_{\tau+2}^{k_{\tau+2}}~\cdots~\bar{x}_{m}^{k_{m}}\end{align}
and \begin{align}\label{eq-seqxp-form}\bm
x'=\tilde{x}_1^{k'_1}~\cdots~\tilde{x}_{\tau'-1}^{k'_{\tau'-1}}~b^{\sigma'_1}
~b~a~a^{\sigma'_2}
~\tilde{x}_{\tau'+2}^{k'_{\tau'+2}}~\cdots~\tilde{x}_{m'}^{k'_{m'}}\end{align}
where $\sigma_1,\sigma_2,\sigma_1',\sigma_2'\geq 0$ such that
$a^{\sigma_1}a=\bar{x}_{\tau}^{k_{\tau}}$,
$b\!~b^{\sigma_2}=\bar{x}_{\tau+1}^{k_{\tau+1}}$,
$b^{\sigma'_1}b=\tilde{x}_{\tau'}^{k'_{\tau'}}$ and
$a\!~a^{\sigma'_2}=\tilde{x}_{\tau'+1}^{k'_{\tau'+1}}$. By
comparison, we have
$\bar{x}_1^{k_1}~\cdots~\bar{x}_{\tau-1}^{k_{\tau-1}}~a^{\sigma_1}
=\bm
w=\tilde{x}_1^{k'_1}~\cdots~\tilde{x}_{\tau'-1}^{k'_{\tau'-1}}~b^{\sigma'_1}$
and $b^{\sigma_2}
~\bar{x}_{\tau+2}^{k_{\tau+2}}~\cdots~\bar{x}_{m}^{k_{m}} =\bm
w'=a^{\sigma'_2}
~\tilde{x}_{\tau'+2}^{k'_{\tau'+2}}~\cdots~\tilde{x}_{m'}^{k'_{m'}}$.
Letting $\bm
w=\bar{x}_1^{k_1}~\cdots~\bar{x}_{\tau-1}^{k_{\tau-1}}~a^{\sigma_1}$
and $\bm w'=a^{\sigma'_2}
~\tilde{x}_{\tau'+2}^{k'_{\tau'+2}}~\cdots~\tilde{x}_{m'}^{k'_{m'}}$,
we can write $\bm x$ and $\bm x'$ in the form
\begin{align}\label{form-x-xp-1}
\bm x
&=\bar{x}_1^{k_1}~\cdots~\bar{x}_{\tau-1}^{k_{\tau-1}}~a^{\sigma_1}
~a~b~a^{\sigma'_2}~\tilde{x}_{\tau'+2}^{k'_{\tau'+2}}~\cdots
~\tilde{x}_{m'}^{k'_{m'}}\nonumber\\
\bm
x'&=\bar{x}_1^{k_1}~\cdots~\bar{x}_{\tau-1}^{k_{\tau-1}}~a^{\sigma_1}
~b~a~a^{\sigma'_2}~\tilde{x}_{\tau'+2}^{k'_{\tau'+2}}~\cdots
~\tilde{x}_{m'}^{k'_{m'}}
\end{align}
Similarly, by letting $\bm w
=\tilde{x}_1^{k'_1}~\cdots~\tilde{x}_{\tau'-1}^{k'_{\tau'-1}}~b^{\sigma'_1}$
and $\bm w'=b^{\sigma_2}
~\bar{x}_{\tau+2}^{k_{\tau+2}}~\cdots~\bar{x}_{m}^{k_{m}}$, we can
write $\bm x$ and $\bm x'$ in the form
\begin{align}\label{form-x-xp-2}
\bm x
&=\tilde{x}_1^{k'_1}~\cdots~\tilde{x}_{\tau'-1}^{k'_{\tau'-1}}
~b^{\sigma'_1}~a~b~b^{\sigma_2}
~\bar{x}_{\tau+2}^{k_{\tau+2}}~\cdots~\bar{x}_{m}^{k_{m}}\nonumber\\
\bm
x'&=\tilde{x}_1^{k'_1}~\cdots~\tilde{x}_{\tau'-1}^{k'_{\tau'-1}}
~b^{\sigma'_1}~b~a~b^{\sigma_2}
~\bar{x}_{\tau+2}^{k_{\tau+2}}~\cdots~\bar{x}_{m}^{k_{m}}
\end{align}
For each $i\in[1,\tau-1]$, from \eqref{form-x-xp-1} and by
Observation 1, we can further obtain
$E^{j_i}=\big(\bigcup_{\ell=1}^{\tau-1}E^{j_i}_{j_\ell}\big)\cup
E^{j_i}_{j_{\tau-1}+1}\cup E^{j_i}_{i_1}\cup
E^{j_i}_{i_2}\cup\big(\bigcup_{\ell'=\tau'+2}^{m'}E^{j_i}_{j'_{\ell'}}\big)$;
for each $i\in[\tau+2,m]$, from \eqref{form-x-xp-2} and by
Observation 1, we can obtain
$E^{j_i}=\big(\bigcup_{\ell'=1}^{\tau'-1}E^{j_i}_{j_{\ell'}}\big)\cup
E^{j_i}_{i_1}\cup E^{j_i}_{i_2}\cup E^{j_i}_{i_2+1}\cup
\big(\bigcup_{\ell=\tau+2}^{m}E^{j_i}_{j_{\ell}}\big)$. By
considering the Hamming distance between the corresponding
subsequences, $E^{j_i}$ can be simplified $($see Claims 1, 2 and 3
later$)$.
\end{rem}

\begin{exam}\label{exm-xxp-form}
Let $\bm x=00010111011100111$ and $\bm x'=00010110111100111$. Then
we have $\bm x=\bm w\!~10\!~\bm w'$ and $\bm x=\bm w\!~01\!~\bm
w'$, where $\bm w=0001011$ and $\bm w'=11100111$. Note that both
$\bm x$ and $\bm x'$ have $8$ runs and $i_1=8\in I_4\cap I'_5$, so
$\tau=4$ and $\tau'=5$. According to Remark \ref{rem-form-run}, we
can write $\bm
w=0001011=\bar{x}_1^3\bar{x}_2^1\bar{x}_3^11^{\sigma_1}
=\tilde{x}_1^3\tilde{x}_2^1\tilde{x}_3^1\tilde{x}_4^20^{\sigma'_1}$,
where $\bar{x}_1\bar{x}_2\bar{x}_3=010$,
$\tilde{x}_1\tilde{x}_2\tilde{x}_3\tilde{x}_4=0101$, $\sigma_1=2$
and $\sigma_1'=0$. Similarly, we can write $\bm
w'=11100111=0^{\sigma_2}\bar{x}_6^3\bar{x}_7^2\bar{x}_8^3
=1^{\sigma'_2}\tilde{x}_7^2\tilde{x}_8^3$, where
$\bar{x}_6\bar{x}_7\bar{x}_8=101$, $\tilde{x}_7\tilde{x}_8=01$,
$\sigma_2=0$ and $\sigma_2'=3$. Note that we alway have
$\sigma_1\sigma_1'=0$ and $\sigma_2\sigma_2'=0$.
\end{exam}

\begin{rem}\label{rem-dst-subseq}
Let $\{j_i: i\in[m]\}$, $\{j'_{i'}: i'\in[m]\}$, $\tau$ and
$\tau'$ be defined as in Remark \ref{rem-form-run}. Note that in
\eqref{form-x-xp-1}, we have $a\neq \bar{x}_{\tau-1}$ and $a\neq
\tilde{x}_{\tau'+2}$ because they belong to different runs of $\bm
x~($Resp. $\bm x')$. Then from \eqref{form-x-xp-1} and by
Observation 2, we can easily find the Hamming distance
$d_{\text{H}}\big(x_{[n]\backslash\{j\}},
x'_{[n]\backslash\{j'}\big)$ for any
$(j,j')\in\{(j_\ell,j_{\ell'}):
\ell,\ell'\in[1,\tau-1]\}\cup\{(j_\ell,j'_{\ell'}):
\ell\in[1,\tau-1], \ell'\in[\tau'+2,m']\}$. The following are some
examples to show this idea and will be used in our subsequent
discussions.
\begin{enumerate}
 \item[1)] For any $\ell<\ell'\leq\tau-1$, we have
 $d_{\text{H}}\big(x_{[n]\backslash\{j_\ell\}},
 x'_{[n]\backslash\{j_{\ell'}}\big)=\ell'-\ell+2$. This is because
 from \eqref{form-x-xp-1}, we can obtain $\text{supp}(\bm x-\bm
 x')\backslash[j_{\ell},j_{\ell'}]=\{i_1,i_2\}$ and
 $\text{supp}\big(x_{[j_{\ell}+1,j_{\ell'}]}
 -x'_{[j_{\ell},j_{\ell'}-1]}\big)=\{j_{\ell},\cdots,j_{\ell'-1}\}$,
 where for convenience, for any $j\leq j''<j'$,
 we say that $j''\in\text{supp}\big(x_{[j+1,j']}
 -x'_{[j,j'-1]}\big)$ if $x'_{j''}\neq x_{j''+1}$.
 Therefore, by Observation 2, we have $d_{\text{H}}\big(x_{[n]\backslash\{j_\ell\}},
 x'_{[n]\backslash\{j_{\ell'}}\big)=\big|\text{supp}(\bm x-\bm
 x')\backslash[j_{\ell},j_{\ell'}]\big|
 +\big|\text{supp}\big(x_{[j_{\ell}+1,j_{\ell'}]}
 -x'_{[j_{\ell},j_{\ell'}-1]}\big)\big|=\ell'-\ell+2$.
 \item[2)] For any $\ell\leq\tau-1$ and $\ell'\geq\tau'+2$, we have
 $d_{\text{H}}\big(x_{[n]\backslash\{j_\ell\}},
 x'_{[n]\backslash\{j'_{\ell'}}\big)=\big((\tau-1)-\ell\big)+2+\big(\ell'-(\tau'+2)\big)$. This is because
 from \eqref{form-x-xp-1}, we can obtain $\text{supp}(\bm x-\bm
 x')\backslash[j_{\ell},j_{\ell'}]=\emptyset$ and
 $\text{supp}\big(x_{[j_{\ell}+1,j_{\ell'}]}
 -x'_{[j_{\ell},j_{\ell'}-1]}\big)=\{j_{\ell},\cdots,j_{\tau-1}\}
 \cup\{j'_{\tau'+1}\}\cup\{j'_{\tau'+2},\cdots,j'_{\ell'-1}\}$.
 Therefore, by Observation 2, we have $d_{\text{H}}\big(x_{[n]\backslash\{j_\ell\}},
 x'_{[n]\backslash\{j_{\ell'}}\big)=\big|\text{supp}(\bm x-\bm
 x')\backslash[j_{\ell},j_{\ell'}]\big|
 +\big|\text{supp}\big(x_{[j_{\ell}+1,j_{\ell'}]}
 -x'_{[j_{\ell},j_{\ell'}-1]}\big)\big|=\big((\tau-1)-\ell+1\big)+1+\big(\ell'-(\tau'+2)\big)
 =\big((\tau-1)-\ell\big)+2+\big(\ell'-(\tau'+2)\big)$.
 \item[3)] In Table 1, we list the Hamming distances
 $d_{\text{H}}\big(x_{[n]\backslash\{j\}},
 x'_{[n]\backslash\{j'}\big)$ for all $j\in\big\{j_\ell: \ell\leq
 \tau-1\big\}$ and $j'\in\big\{j_\ell: \tau-4\leq \ell\leq
 \tau-1\big\}\cup\{j_{\tau-1}+1,i_1,i_2\}\cup\big\{j'_{\ell}:
 \ell\geq\tau'+2\big\}$.
\end{enumerate}
\end{rem}

Similarly, in \eqref{form-x-xp-2}, we have $b\neq
\tilde{x}_{\tau'-1}$ and $b\neq \bar{x}_{\tau+2}$, and so from
\eqref{form-x-xp-2} and by Observation 2, we can easily find the
Hamming distance $d_{\text{H}}\big(x_{[n]\backslash\{j\}},
x'_{[n]\backslash\{j'}\big)$ for any
$(j,j')\in\{(j_\ell,j_{\ell'}):
\ell,\ell'\in[\tau+2,m]\}\cup\{(j_\ell,j'_{\ell'}):
\ell\in[\tau+2,m], \ell'\in[1,\tau'-1]\}$.

\begin{table}[htbp]
\begin{center}
\small
\renewcommand\arraystretch{1.1}
\begin{tabular}{|p{1.1cm}|p{0.8cm}|p{0.8cm}|p{0.8cm}|p{0.8cm}|p{1.1cm}|p{0.8cm}|p{0.8cm}|p{0.8cm}|p{0.8cm}|p{0.8cm}|p{1.2cm}|p{1.6cm}|}
\hline    ~      & $j_{\tau-4}$ & $j_{\tau-3}$ & $j_{\tau-2}$ & $j_{\tau-1}$ & $j_{\tau-1}+1$ & $i_{1}$ & $i_{2}$ & $j'_{\tau'+2}$ & $j'_{\tau'+3}$ & $j'_{\tau'+4}$ & $\geq j'_{\tau'+5}$ \\
\hline $j_{\tau-1}$   & $5$       & $4$       & $3$       & $2$     & $3$       & $2$     & $1$     & $2$        & $3$        & $4$        & $\geq 5$        \\
\hline $j_{\tau-2}$ & $4$       & $3$       & $2$       & $3$     & $4$       & $3$     & $2$     & $3$        & $4$        & $5$        & $\geq 6$        \\
\hline $j_{\tau-3}$ & $3$       & $2$       & $3$       & $4$     & $5$       & $4$     & $3$     & $4$        & $5$        & $6$        & $\geq 7$        \\
\hline $j_{\tau-4}$ & $2$       & $3$       & $4$       & $5$     & $6$       & $5$     & $4$     & $5$        & $6$        & $7$        & $\geq 8$        \\
\hline $\leq j_{\tau-5}$ & $\geq 3$       & $\geq 4$       & $\geq 5$       & $\geq 6$     & $\geq 7$       & $\geq 6$     & $\geq 5$     & $\geq 6$        & $\geq 7$        & $\geq 8$        & $\geq 9$        \\
\hline
\end{tabular}\\
\end{center}
\vspace{0.0mm} Table 1. The Hamming distance
$d_{\text{H}}\big(x_{[n]\backslash\{j\}},
x'_{[n]\backslash\{j'}\big)$ for $j\in\big\{j_\ell: \ell\leq
\tau-1\big\}$ and $j'\in\big\{j_\ell: \tau-4\leq \ell\leq
\tau-1\big\}\cup\{j_{\tau-1}+1,i_1,i_2\}\cup\big\{j'_{\ell}:
\ell\geq\tau'+2\big\}$, where $\{j_i: i\in[m]\}$, $\{j'_{i'}:
i'\in[m]\}$, $\tau$ and $\tau'$ are defined as in Remark
\ref{rem-form-run}. In this table, each row corresponds to a $j$
and each column corresponds to a $j'$. It is assumed that
$\sigma_1>0$ and the Hamming distance
$d_{\text{H}}\big(x_{[n]\backslash\{j\}},
x'_{[n]\backslash\{j'}\big)$ can be obtained from
\eqref{form-x-xp-1} and by Observation 2. The table for the case
that $\sigma_1=0$ can be obtained from this table by removing the
column with $j'=j_{\tau-1}+1$.
\end{table}

The following lemma is the main result of this subsection.

\begin{lem}\label{lem-both-i1-i2}
Suppose $\bm x,\bm x'\in\Sigma_q^n$ such that $\text{supp}(\bm
x-\bm x')=\{i_1,i_2\}$ and $i_1<i_2$. If both
$x_{[n]\backslash\{i_1\}}=x'_{[n]\backslash\{i_2\}}$ and
$x_{[n]\backslash\{i_2\}}=x'_{[n]\backslash\{i_1\}}$, then we have
$|B^{\text{DS}}_{1,2}(\bm x, \bm x')|\leq
(q^2-1)n^2-(3q^2+5q-5)n+O_q(1)$.
\end{lem}

Before proving Lemma~\ref{lem-both-i1-i2}, we introduce several
auxiliary claims, whose proofs are given in Appendix~A.

\textbf{Claim 1}: We have
$E^{i_1}=B^{\text{S}}_{2}\big(x_{[n]\backslash\{i_1\}}\big)
=B^{\text{S}}_{2}\big(x'_{[n]\backslash\{i_2\}}\big)$ and
$E^{i_2}=B^{\text{S}}_{2}\big(x_{[n]\backslash\{i_2\}}\big)
=B^{\text{S}}_{2}\big(x'_{[n]\backslash\{i_1\}}\big)$. Moreover,
$|E^{i_1}\cup E^{i_2}|=2\xi^{q,n-1}_{0,2}-\xi^{q,n-1}_{1,2}
=(q-1)^2n^2
-(4q^2-9q+5)n+4q^2-9q+6=(q-1)^2n^2-(4q^2-9q+5)n+O_q(1)$.

\textbf{Claim 2}: Let $\{j_{i}: i\in[m]\}$, $\{j'_{i'}:
i'\in[m']\}$, $\tau$ and $\tau'$ be defined as in Remark
\ref{rem-form-run}. Then $$B^{\text{DS}}_{1,2}(\bm x,\bm
x')=\Big(\bigcup\nolimits_{\ell=1}^{\tau-1}
E^{j_\ell}_{j_\ell}\Big)\cup\Big(\bigcup\nolimits_{\ell=\tau+2}^{m}
E^{j_\ell}_{j_\ell}\Big)\cup(E^{i_{1}}\cup E^{i_{2}})\cup
E^{j_{\tau-1}}_{j'_{\tau'+2}}\cup
E^{j_{\tau+2}}_{j'_{\tau'-1}}\cup A,$$ where
$A\subseteq\Sigma_q^{n-1}$ and $|A|=O_q(1)$.

\textbf{Claim 3}: Let $\{j_i: i\in[m]\}$, $\{j'_{i'}:
i'\in[m']\}$, $\tau$ and $\tau'$ be defined as in Remark
\ref{rem-form-run}. We have
$$\Big|\Big(\bigcup\nolimits_{\ell=1}^{\tau-1}
E^{j_\ell}_{j_\ell}\Big)\cap\Big(\bigcup\nolimits_{\ell=\tau+2}^{m}
E^{j_\ell}_{j_\ell}\Big)\Big|=O_q(1).$$

\textbf{Claim 4}: Let $\{j_i: i\in[m]\}$, $\{j'_{i'}:
i'\in[m']\}$, $\tau$ and $\tau'$ be defined as in Remark
\ref{rem-form-run}.
\begin{itemize}
 \item[1)] If $\tau>1$, then
 $\big|\big(\bigcup_{\ell=1}^{\tau-1}E^{j_\ell}_{j_\ell}\big)\backslash
 (E^{i_1}\cup E^{i_2})\big|=(q-1)(2\tau-4)n+(q^2-8q+4)\tau+1_{\sigma_1>0}(q-1)n+O_q(1)$,
 where $1_{\sigma_1>0}=1$ if $\sigma_1>0$ and $1_{\sigma_1>0}=0$ if $\sigma_1=0$.
 \item[2)] If $\tau<m-1$, then
 $\Big|\left(\bigcup_{\ell=\tau+2}^{m}E^{j_\ell}_{j_\ell}\right)\backslash
 (E^{i_1}\cup E^{i_2})\Big|=(q-1)(2m-2\tau-4)n+(q^2-8q+4)(m-\tau)+1_{\sigma_2>0}(q-1)n+O_q(1)$,
 where $1_{\sigma_2>0}=1$ if $\sigma_2>0$ and $1_{\sigma_2>0}=0$ if $\sigma_2=0$.
\end{itemize}


\textbf{Claim 5}: Let $\{j_i: i\in[m]\}$, $\{j'_{i'}: i'\in[m]\}$,
$\tau$ and $\tau'$ be defined as in Remark \ref{rem-form-run}.
Denote
$\Omega=\Big(\bigcup_{\ell=1}^{\tau-1}E^{j_\ell}_{j_\ell}\Big)\cup
\Big(\bigcup_{\ell=\tau+2}^{m}E^{j_\ell}_{j_\ell}\Big)\cup(E^{i_1}\cup
E^{i_2})$. If $\tau>1$ and $\tau'<m'-1$, we have
$$\Big|E^{j_{\tau-1}}_{j'_{\tau'+2}}\cup\Omega
\Big|=|\Omega|+(q-1)n+O_q(1).$$ If $\tau<m-1$ and $\tau'>1$, we
have
$$\Big|E^{j_{\tau+2}}_{j'_{\tau'-1}}\cup\Omega
\Big|=|\Omega|+(q-1)n+O_q(1).$$ Moreover, if $1<\tau<m-1$ and
$1<\tau'<m'-1$, then
$$\Big|E^{j_{\tau-1}}_{j'_{\tau'+2}}\cup
E^{j_{\tau+2}}_{j'_{\tau'-1}}\cup\Omega\Big|
=|\Omega|+2(q-1)n+O_q(1).$$

Now we can prove Lemma \ref{lem-both-i1-i2}.

\begin{proof}[Proof of Lemma \ref{lem-both-i1-i2}]
Let $\{j_i: i\in[m]\}$, $\{j'_{i'}: i'\in[m]\}$, $\tau$ and
$\tau'$ be defined as in Remark \ref{rem-form-run}.

First, suppose $\tau=m-1$. By 1) of Claim 4, we can obtain
\begin{align}\label{Sum-L-EE}
\Big|\Big(\bigcup_{\ell=1}^{\tau-1}E^{j_\ell}_{j_\ell}\Big)\big\backslash
(E^{i_1}\cup
E^{i_2})\Big|&=(q-1)[2(m-1)-4]n+(q^2-8q+4)(m-1)+1_{\sigma_1>0}(q-1)n+O_q(1)\nonumber\\
&=2(q-1)mn-6(q-1)n+(q^2-8q+4)m+1_{\sigma_1>0}(q-1)n+O_q(1)\nonumber\\
&\leq
2(q-1)(n-\sigma_1-\sigma_2)n-6(q-1)n+(q^2-8q+4)(n-\sigma_1-\sigma_2)+1_{\sigma_1>0}(q-1)n+O_q(1)\nonumber\\
&=2(q-1)n^2+(q^2-14q+10)n-(q-1)(2\sigma_1-1_{\sigma_1>0})n-2(q-1)\sigma_2n+O_q(1)\nonumber\\
&\leq 2(q-1)n^2+(q^2-14q+10)n-2(q-1)\sigma_2n+O_q(1)
\end{align} where the first inequality holds because by
\eqref{eq-seqx-form}, we have $m\leq n-\sigma_1-\sigma_2$. To
compute $\big|B^{\text{DS}}_{1,2}(\bm x,\bm x')\big|$, we need to
consider the following two cases.

Case 1: $\sigma_2=0$. Noticing that $\tau=m-1$, by
\eqref{form-x-xp-1} and \eqref{form-x-xp-2}, we have
$\bigcup_{\ell=\tau+2}^{m}E^{j_\ell}_{j_\ell}=\emptyset$,
$E^{j_{\tau-1}}_{j'_{\tau'+2}}=\emptyset$ and
$E^{j_{\tau+2}}_{j'_{\tau'-1}}=\emptyset$, so by Claim 2, we can
obtain $B^{\text{DS}}_{1,2}(\bm x,\bm
x')=\Big(\bigcup_{\ell=1}^{\tau-1}E^{j_\ell}_{j_\ell}\Big)\cup(E^{i_{1}}\cup
E^{i_{2}})\cup A,$ where $A\subseteq\Sigma_q^{n-1}$ and
$|A|=O_q(1)$. Further by \eqref{Sum-L-EE} and Claim 1, we can
obtain
\begin{align*}
\big|B^{\text{DS}}_{1,2}(\bm x,\bm
x')\big|&=\Big|\Big(\bigcup_{\ell=1}^{\tau-1}E^{j_\ell}_{j_\ell}\Big)\cup
(E^{i_1}\cup E^{i_2})\cup
A\Big|\\&=\Big|\Big(\bigcup_{\ell=1}^{\tau-1}E^{j_\ell}_{j_\ell}\Big)\big\backslash
(E^{i_1}\cup E^{i_2})\Big|+|E^{i_1}\cup E^{i_2}|+O_q(1)\\&\leq
2(q-1)n^2+(q^2-14q+10)n+(q-1)^2n^2-(4q^2-9q+5)n+O_q(1)\\
&=(q^2-1)n^2-(3q^2+5q-5)n+O_q(1)
\end{align*}

Case 2: $\sigma_2>0$. Noticing that $\tau=m-1$, by
\eqref{form-x-xp-1} and \eqref{form-x-xp-2}, we have
$\bigcup_{\ell=\tau+2}^{m}E^{j_\ell}_{j_\ell}=\emptyset$ and
$E^{j_{\tau+2}}_{j'_{\tau'-1}}=\emptyset$, so by Claim 2, we can
obtain $B^{\text{DS}}_{1,2}(\bm x,\bm
x')=\Big(\bigcup_{\ell=1}^{\tau-1}E^{j_\ell}_{j_\ell}\Big)\cup(E^{i_{1}}\cup
E^{i_{2}})\cup E^{j_{\tau-1}}_{j'_{\tau'+2}}\cup A,$ where
$A\subseteq\Sigma_q^{n-1}$ and $|A|=O_q(1)$. Further by Claim 5,
\eqref{Sum-L-EE} and Claim 1, we can obtain
\begin{align*}
\big|B^{\text{DS}}_{1,2}(\bm x,\bm
x')\big|&=\Big|\Big(\bigcup_{\ell=1}^{\tau-1}E^{j_\ell}_{j_\ell}\Big)\cup
(E^{i_1}\cup E^{i_2})\cup E^{j_{\tau-1}}_{j'_{\tau'+2}}\cup
A\Big|\\&=\Big|\Big(\bigcup_{\ell=1}^{\tau-1}E^{j_\ell}_{j_\ell}\Big)\big\backslash
(E^{i_1}\cup E^{i_2})\Big|+|E^{i_1}\cup
E^{i_2}|+(q-1)n+O_q(1)\\&\leq
2(q-1)n^2+(q^2-14q+10)n-2(q-1)\sigma_2n+(q-1)^2n^2-(4q^2-9q+5)n+(q-1)n+O_q(1)\\
&=(q^2-1)n^2-(3q^2+4q-4)n-2(q-1)\sigma_2n+O_q(1)\\
&\leq(q^2-1)n^2-(3q^2+6q-6)n+O_q(1)\\
&\leq(q^2-1)n^2-(3q^2+5q-5)n+O_q(1)
\end{align*}
where the last inequality holds because
$\leq(q^2-1)n^2-(3q^2+6q-6)n<(q^2-1)n^2-(3q^2+5q-5)n$. Thus, for
$\tau=m-1$, we have $\big|B^{\text{DS}}_{1,2}(\bm x,\bm
x')\big|\leq(q^2-1)n^2-(3q^2+5q-5)n+O_q(1)$. Similarly, if
$\tau=1$, we can prove $\big|B^{\text{DS}}_{1,2}(\bm x,\bm
x')\big|\leq(q^2-1)n^2-(3q^2+5q-5)n+O_q(1)$.

Finally, we suppose $1<\tau<m-1$. By Claim 3 and Claim 4, we can
obtain
\begin{align*}
|\Omega|&=\Big|\Big(\bigcup_{\ell=1}^{\tau-1}E^{j_\ell}_{j_\ell}\Big)\big\backslash
(E^{i_1}\cup
E^{i_2})\Big|+\Big|\Big(\bigcup_{\ell=1}^{\tau-1}E^{j_\ell}_{j_\ell}\Big)\big\backslash
(E^{i_1}\cup
E^{i_2})\Big|+|E^{i_1}\cup E^{i_2}|+O_q(1)\nonumber\\
&=(q-1)(2m-8)n+(q^2-8q+4)m+(1_{\sigma_1>0}+1_{\sigma_2>0})(q-1)n\nonumber\\
&~~~~+(q-1)^2n^2-(4q^2-9q+5)n+O_q(1)\nonumber\\
&=(q-1)^2n^2+2(q-1)mn+(q^2-8q+4)m-(4q^2-q-3)n+(1_{\sigma_1>0}+1_{\sigma_2>0})(q-1)n+O_q(1)\nonumber\\
&\leq
(q-1)^2n^2+2(q-1)(n-\sigma_1-\sigma_2)n+(q^2-8q+4)(n-\sigma_1-\sigma_2)-(4q^2-q-3)n\nonumber\\
&~~~~+(1_{\sigma_1>0}+1_{\sigma_2>0})(q-1)n+O_q(1)\nonumber\\
&=(q^2-1)n^2-(3q^2+7q-7)n-(q-1)\big(2\sigma_1+2\sigma_2-1_{\sigma_1>0}-1_{\sigma_2>0}\big)n+O_q(1)\nonumber\\
&\leq (q^2-1)n^2-(3q^2+7q-7)n+O_q(1)
\end{align*}
where
$\Omega=\Big(\bigcup_{\ell=1}^{\tau-1}E^{j_\ell}_{j_\ell}\Big)\cup
\Big(\bigcup_{\ell=\tau+2}^{m}E^{j_\ell}_{j_\ell}\Big)\cup(E^{i_1}\cup
E^{i_2})$. Further, by Claim 2 and Claim 5, we can obtain
\begin{align*}
\big|B^{\text{DS}}_{1,2}(\bm x,\bm
x')\big|&=\Big|E^{j_{\tau-1}}_{j'_{\tau'+2}}\cup
E^{j_{\tau+2}}_{j'_{\tau'-1}}\cup\Omega\cup
A\Big|\\&=|\Omega|+2(q-1)n+O_q(1)\\
&\leq(q^2-1)n^2-(3q^2+7q-7)n+2(q-1)n+O_q(1)\\
&=(q^2-1)n^2-(3q^2+5q-5)n+O_q(1)
\end{align*}

Thus, we always have $\big|B^{\text{DS}}_{1,2}(\bm x,\bm
x')\big|\leq(q^2-1)n^2-(3q^2+5q-5)n+O_q(1)$, which completes the
proof.
\end{proof}

According to the proof of Lemma \ref{lem-both-i1-i2}, when
$1<\tau<m-1$, $\sigma_1=\sigma_2=0$ and $m=n$, we can obtain
$\big|B^{\text{DS}}_{1,2}(\bm x,\bm
x')\big|=(q^2-1)n^2-(3q^2+5q-5)n+O_q(1)$. This implies that the
bound in Lemma \ref{lem-both-i1-i2} is tight up to an additive
constant. For example, let $\bm x=101010A_{n-6}(10)$ and $\bm
x'=100110A_{n-6}(10)$, where $n\geq 6$ and $A_{n-6}(10)$ is the
length-$n-6$ sequence that starts with $1$ and consists of
alternating symbols $1$ and $0$. Then
$\big|B^{\text{DS}}_{1,2}(\bm x,\bm
x')\big|=(q^2-1)n^2-(3q^2+5q-5)n+O_q(1)$.

\subsection{For the case that either $x_{[n]\backslash\{i_1\}}
\neq x'_{[n]\backslash\{i_2\}}$ or $x_{[n]\backslash\{i_2\}} \neq
x'_{[n]\backslash\{i_1\}}$ but not both}

Without loss of generality, we assume that
$x_{[n]\backslash\{i_1\}}=x'_{[n]\backslash\{i_2\}}$ and
$x_{[n]\backslash\{i_2\}} \neq x'_{[n]\backslash\{i_1\}}$. From
the assumption
$x_{[n]\backslash\{i_1\}}=x'_{[n]\backslash\{i_2\}}$, we can
obtain $x'_{j'}=x_{j'+1}$ for all $j'\in[i_1,i_2-1]$. Moreover,
since $\text{supp}(\bm x-\bm x')=\{i_1,i_2\}$, we have
$x_{j'}=x'_{j'}$ for each $j'\in[i_1+1,i_2-1]$. Hence, $\bm x$ and
$\bm x'$ can be written as the form
\begin{align*}
\bm x &=\!~\bm w~a~b^{\sigma_0}~b~\bm w'\\
\bm x'&=\!~\bm w~b~b^{\sigma_0}~c~\bm w'
\end{align*}
where $\bm w=x_{[1,i_1-1]}=x'_{[1,i_1-1]}, a=x_{i_1},
b=x'_{i_1}=x_{i_2}, c=x'_{i_2}$, $\bm
w'=x_{[i_2+1,n]}=x'_{[i_2+1,n]}$ and $\sigma_0=i_2-i_1-1\geq 0$.
Since $x_{[n]\backslash\{i_2\}} \neq x'_{[n]\backslash\{i_1\}}$,
we have $a\neq c$ if $\sigma_0=0$.

Let $\{j_1,j_2,\cdots,j_m\}$, $\{j'_1,j'_2,\cdots,j'_{m'}\}$,
$\tau$ and $\tau'$ be defined as in Remark \ref{rem-form-run}.
Then we can write
\begin{align}\label{eq-seqx-form-nb}\bm
x=\bar{x}_1^{k_1}~\cdots~\bar{x}_{\tau-1}^{k_{\tau-1}}~a^{\sigma_1}
~a~b^{\sigma_0}~b~b^{\sigma_2}
~\bar{x}_{\tau+2}^{k_{\tau+2}}~\cdots~\bar{x}_{m}^{k_{m}}\end{align}
and \begin{align}\label{eq-seqxp-form-nb}\bm
x'=\tilde{x}_1^{k'_1}~\cdots~\tilde{x}_{\tau'-1}^{k'_{\tau'-1}}~b^{\sigma'_1}
~b~b^{\sigma_0}~c~c^{\sigma'_2}
~\tilde{x}_{\tau'+2}^{k'_{\tau'+2}}~\cdots~\tilde{x}_{m'}^{k'_{m'}}\end{align}
where $\sigma_1,\sigma_2,\sigma_1',\sigma_2'\geq 0$ such that
$a^{\sigma_1}a=\bar{x}_{\tau}^{k_{\tau}}$,
$b^{\sigma_0}\!~b\!~b^{\sigma_2}=\bar{x}_{\tau+1}^{k_{\tau+1}}$,
$b^{\sigma'_1}b\!~b^{\sigma_0}~=\tilde{x}_{\tau'}^{k'_{\tau'}}$
and $a\!~a^{\sigma'_2}=\tilde{x}_{\tau'+1}^{k'_{\tau'+1}}$. Then
by Observation 1, we have $B^{\text{DS}}_{1,2}(\bm x,\bm
x')=\big(\bigcup_{i=1}^{\tau-1}E^{j_i}\big)\cup E^{i_1}\cup
E^{i_2}\cup\big(\bigcup_{i=\tau+2}^{m}E^{j_i}\big)$. Moreover, for
each $i\in[1,\tau-1]$, we have
$E^{j_i}=\big(\bigcup_{\ell=1}^{\tau-1}E^{j_i}_{j_\ell}\big)\cup
E^{j_i}_{j_{\tau-1}+1}\cup E^{j_i}_{i_1}\cup
E^{j_i}_{i_2}\cup\big(\bigcup_{\ell'=\tau'+2}^{m'}E^{j_i}_{j'_{\ell'}}\big)$;
for each $i\in[\tau+2,m]$, we have
$E^{j_i}=\big(\bigcup_{\ell'=1}^{\tau'-1}E^{j_i}_{j_{\ell'}}\big)\cup
E^{j_i}_{i_1}\cup E^{j_i}_{i_2}\cup E^{j_i}_{i_2+1}\cup
\big(\bigcup_{\ell=\tau+2}^{m}E^{j_i}_{j_{\ell}}\big)$. As in
Claim 1, we have
$E^{i_1}=B^{\text{S}}_{2}\big(x_{[n]\backslash\{i_1\}}\big)
=B^{\text{S}}_{2}\big(x'_{[n]\backslash\{i_2\}}\big)$ and
$$|E^{i_1}|=\xi^{q,n-1}_{0,2}=\frac{(q-1)^2}{2}n^2
-\frac{(3q-5)(q-1)}{2}n+(q^2-3q+3)$$ Similar to Claim 6 and Claim
6$'~($see Appendix A$)$, we can obtain
\begin{itemize}
 \item[1)] $E^{i_2}\subseteq\bigcup_{j'\in J}E^{i_2}_{j'}$ and
 $|E^{i_2}|\leq 2\xi^{q,n-1}_{1,2}+2\xi^{q,n-1}_{2,2}+O_q(1)$, where
 $J=\{j'_{\tau'-3}, j'_{\tau'-2}, j'_{\tau'-1},
 i_1, i_2, j_{\tau+2}, j_{\tau+3}, j_{\tau+4}\}$.
 \item[2)] $E^{j_i}\subseteq
 \bigcup\nolimits_{\ell=\max\{1,i-2\}}^{\min\{m,i+2\}}
 E^{j_\ell}_{j_\ell}, ~~\forall\!~ i\in[1, \tau-5]\cup[\tau+5,m].$
 \item[3)] $\bigcup_{i=\tau-4}^{\tau-1}E^{j_{i}}\subseteq
 \left(\bigcup_{\ell=\tau-6}^{\tau-1}E^{j_\ell}_{j_\ell}\right)
 \cup E^{j_{\tau-1}}_{i_1}\cup E^{j_{\tau-1}}_{j_{\tau'+2}}\cup A$,
 where
 $\big|E^{j_{\tau-1}}_{i_1}\big|=\big|E^{j_{\tau-1}}_{\tau'+2}\big|
 =\xi^{q,n-1}_{2,2}$ and $|A|=O_q(1)$.
 \item[4)] $\bigcup_{i=\tau+2}^{\tau+4}E^{j_{i}}\subseteq
 \left(\bigcup_{\ell=\tau+2}^{\tau+6}E^{j_\ell}_{j_\ell}\right)
 \cup E^{j_{\tau+2}}_{i_1}\cup A'$, where
 $\big|E^{j_{\tau+2}}_{i_1}\big|
 \leq\xi^{q,n-1}_{2,2}$ and $|A'|=O_q(1)$.
\end{itemize}
Thus, we can obtain
\begin{align}\label{eq1-pf-main-thm}
\left|B^{\text{DS}}_{1,2}(\bm x,\bm
x')\right|&=\left|\left(\bigcup_{i=1}^{\tau-1}E^{j_i}\right)\cup
E^{i_1}\cup E^{i_2}\cup\left(\bigcup_{i=\tau+2}^{m}E^{j_i}\right)\right|\nonumber\\
&=\left|\left(\bigcup_{\ell=1}^{\tau-1}E^{j_\ell}_{j_\ell}\right)\cup
E^{i_1}\cup E^{i_2}\cup E^{j_{\tau-1}}_{i_1}\cup
E^{j_{\tau-1}}_{j_{\tau'+2}}\cup
E^{j_{\tau+2}}_{i_1}\cup\left(\bigcup_{\ell=\tau+2}^{m}E^{j_\ell}_{j_\ell}\right)\cup A\cup A'\right|\nonumber\\
&\leq\sum_{\ell=1}^{\tau-1}\left|E^{j_\ell}_{j_\ell}\right|
+\left|E^{i_1}\right|+\left|E^{i_2}\right|+\left|E^{j_{\tau-1}}_{i_1}\right|
+\left|E^{j_{\tau-1}}_{j_{\tau'+2}}\right|
+\left|E^{j_{\tau+2}}_{i_1}\right|
+\sum_{\ell=\tau+2}^{m}\left|E^{j_\ell}_{j_\ell}\right|+|A|+|A'|\nonumber\\
&\leq\xi^{q,n-1}_{0,2}+2\xi^{q,n-1}_{1,2}+(m+3)\xi^{q,n-1}_{2,2}+O_q(1)\nonumber\\
&=\frac{(q-1)^2}{2}n^2-\frac{(3q-5)(q-1)}{2}n+(q^2-3q+3)+2[q(q-1)n-2q^2+3q]\nonumber\\
&~~~~+(m+3)[2(q-1)n+q^2-6q+6]+O_q(1)\nonumber\\
&\leq \frac{(q-1)^2}{2}n^2-\frac{(3q-5)(q-1)}{2}n+(q^2-3q+3)+2[q(q-1)n-2q^2+3q]\nonumber\\
&~~~~+(n+3)[2(q-1)n+q^2-6q+6]+O_q(1)\nonumber\\
&=\frac{(q+3)(q-1)}{2}n^2+\frac{(q+17)(q-1)}{2}n+O_q(1)\nonumber\\
&\leq(q^2-1)n^2-(3q^2+5q-5)n+O_q(1).
\end{align}
For the last inequality, noticing that
$\frac{(q+3)(q-1)}{2}<q^2-1$, so for $n$ sufficiently large, we
can obtain $\frac{(q+3)(q-1)}{2}n^2+\frac{(q+17)(q-1)}{2}n+O_q(1)
\leq(q^2-1)n^2-(3q^2+5q-5)n+O_q(1)$.


\subsection{For the case that both $x_{[n]\backslash\{i_1\}}
\neq x'_{[n]\backslash\{i_2\}}$ and $x_{[n]\backslash\{i_2\}}\neq
x'_{[n]\backslash\{i_1\}}$}

Similar to Remark \ref{rem-form-run}, let $1\leq
j_1<j_2<\cdots<j_m\leq n$ be such that $x_{[j_{i-1}+1,j_i]}$,
$i=1,2,\cdots, m$, are all runs of $\bm x$, and $1\leq
j'_1<j'_2<\cdots<j'_{m'}\leq n$ be such that
$x'_{[j'_{i'-1}+1,j'_{i'}]}$, $i'=1,2,\cdots, m'$, are all runs of
$\bm x'$, where we set $j_0=j'_0=0$. By Observation 1, we have
$B^{\text{DS}}_{1,2}(\bm x,\bm x')=\bigcup_{i=1}^m
E^{j_i}=\bigcup_{i=1}^m\bigcup_{i'=1}^{m'}E^{j_i}_{j'_{i'}}$.
Therefore, \begin{align*}
\left|B^{\text{DS}}_{1,2}(\bm x,\bm
x')\right|\leq\sum_{i=1}^m\sum_{i'=1}^{m'}\left|E^{j_i}_{j'_{i'}}\right|.
\end{align*}
For each $d'\in\{0,1,2,3,4\}$, we will use Lemma
\ref{lem-dtn-jpj-run} to estimate the number of pairs $(j_{i},
j'_{i'})$ such that $d_{\text{H}}\big(x_{[n]\backslash\{j_{i}\}},
x'_{[n]\backslash\{j'_{i'}\}}\big)=d'$. As
$x_{[n]\backslash\{i_1\}}\neq x'_{[n]\backslash\{i_2\}}$ and
$x_{[n]\backslash\{i_2\}}\neq x'_{[n]\backslash\{i_1\}}$, we can
obtain $\Delta_{[i_1,i_{2}]}>0$ and $\Delta'_{[i_1,i_{2}]}>0$,
where $\Delta_{[i_1,i_{2}]}$ and $\Delta'_{[i_1,i_{2}]}$ are
defined in Lemma \ref{lem-dtn-jpj-run}. We first consider
$j_{\ell}\leq j'_{\ell'}$, for which we have the following three
cases.
\begin{enumerate}
\item[1.] $\big|\{i_1,i_2\}\backslash[j_{i}, j'_{i'}]\big|=0$.
Then it must be $j_{i}\leq i_1<i_2\leq j'_{i'}$ and
$S_{[i_1,i_{2}]}=0$. As $x_{[n]\backslash\{i_1\}}\neq
x'_{[n]\backslash\{i_2\}})$, we have $\Delta_{[i_1,i_{2}]}>0$.
Hence, according to 4) of Lemma \ref{lem-dtn-jpj-run} $($with
$\lambda=1$ and $\lambda'=2)$, there is no pair
$(j_{\ell},j'_{\ell'})$ such that
$d_{\text{H}}(x_{[n]\backslash\{j_\ell\}},
x'_{[n]\backslash\{j'_{\ell'}\}}\big)=0$. Moreover, for each
$d'\in\{1,2,3,4\}$, there are at most $d'$ pairs $(j_{i},
j'_{i'})$ such that $d_{\text{H}}\big(x_{[n]\backslash\{j_{i}\}},
x'_{[n]\backslash\{j'_{i'}\}}\big)=d'$.

\item[2.] $\big|\{i_1,i_2\}\backslash[j_{i}, j'_{i'}]\big|=1$.
Then it must be $j_{i}\leq i_1\leq j'_{i'}<i_2~($then
$\lambda=\lambda'=1)$ or $i_1<j_{i}\leq i_2\leq j'_{i'}~($then
$\lambda=\lambda'=2)$. For each case, by definition, we have
$S_{[i_\lambda,i_{\lambda'}]}=1$ and
$\Delta_{[i_\lambda,i_{\lambda'}]}=0$. According to 4) of Lemma
\ref{lem-dtn-jpj-run}, there is no pair $(j_{\ell},j'_{\ell'})$
such that $d_{\text{H}}(x_{[n]\backslash\{j_\ell\}},
x'_{[n]\backslash\{j'_{\ell'}\}}\big)=0$. Moreover, for each case
and each $d'\in\{1,2,3,4\}$, there are at most $d'$ pairs $(j_{i},
j'_{i'})$ such that $d_{\text{H}}\big(x_{[n]\backslash\{j_{i}\}},
x'_{[n]\backslash\{j'_{i'}\}}\big)=d'$.

\item[3.] $\big|\{i_1,i_2\}\backslash[j_{i}, j'_{i'}]\big|=2$.
Then we have $\{i_1,i_2\}\cap[j_{i}, j'_{i'}]=\emptyset$. For
$d'\in\{0,1\}$, by 1) of Lemma \ref{lem-dtn-jpj-run}, there is no
pair $(j_{\ell},j'_{\ell'})$ such that
$d_{\text{H}}(x_{[n]\backslash\{j_\ell\}},
x'_{[n]\backslash\{j'_{\ell'}\}}\big)=d'$. By 2) of Lemma
\ref{lem-dtn-jpj-run}, there are at most $n-2$ pairs $(j_{i},
j'_{i'})$ such that $d_{\text{H}}\big(x_{[n]\backslash\{j_{i}\}},
x'_{[n]\backslash\{j'_{i'}\}}\big)=2$. Moreover, for each
$d'\in\{3,4\}$, by 3) of Lemma \ref{lem-dtn-jpj-run}, there are at
most $n-2$ pairs $(j_{i}, j'_{i'})$ such that
$d_{\text{H}}\big(x_{[n]\backslash\{j_{i}\}},
x'_{[n]\backslash\{j'_{i'}\}}\big)=d'$.
\end{enumerate}
Similarly, for $j_{\ell}>j'_{\ell'}$, we have the following three
cases.
\begin{enumerate}
\item[1.] $\big|\{i_1,i_2\}\backslash[j'_{i'}, j_{i}]\big|=0$.
Then it must be $j'_{i'}\leq i_1<i_2\leq j_{i}$ and
$S_{[i_1,i_{2}]}=0$. As $x_{[n]\backslash\{i_2\}}\neq
x'_{[n]\backslash\{i_1\}})$, we have $\Delta'_{[i_1,i_{2}]}>0$.
Hence, according to 5) of Lemma \ref{lem-dtn-jpj-run} $($with
$\lambda=1$ and $\lambda'=2)$, there is no pair
$(j_{\ell},j'_{\ell'})$ such that
$d_{\text{H}}(x_{[n]\backslash\{j_\ell\}},
x'_{[n]\backslash\{j'_{\ell'}\}}\big)=0$. Moreover, for each
$d'\in\{1,2,3,4\}$, there are at most $d'$ pairs $(j_{i},
j'_{i'})$ such that $d_{\text{H}}\big(x_{[n]\backslash\{j_{i}\}},
x'_{[n]\backslash\{j'_{i'}\}}\big)=d'$.

\item[2.] $\big|\{i_1,i_2\}\backslash[j'_{i'}, j_{i}]\big|=1$.
Then it must be $j'_{i'}\leq i_1\leq j_{i}<i_2$ or
$i_1<j'_{i'}\leq i_2\leq j_{i}$. For each case, by 5) of Lemma
\ref{lem-dtn-jpj-run}, there is no pair $(j_{\ell},j'_{\ell'})$
such that $d_{\text{H}}(x_{[n]\backslash\{j_\ell\}},
x'_{[n]\backslash\{j'_{\ell'}\}}\big)=0$. Moreover, for each case
and each $d'\in\{1,2,3,4\}$, there are at most $d'$ pairs $(j_{i},
j'_{i'})$ such that $d_{\text{H}}\big(x_{[n]\backslash\{j_{i}\}},
x'_{[n]\backslash\{j'_{i'}\}}\big)=d'$.

\item[3.] $\big|\{i_1,i_2\}\backslash[j'_{i'}, j_{i}]\big|=2$.
Then we have $\{i_1,i_2\}\cap[j'_{i'}, j_{i}]=\emptyset$. For
$d'\in\{0,1\}$, by 1) of Lemma \ref{lem-dtn-jpj-run}, there is no
pair $(j_{\ell},j'_{\ell'})$ such that
$d_{\text{H}}(x_{[n]\backslash\{j_\ell\}},
x'_{[n]\backslash\{j'_{\ell'}\}}\big)=d'$. By 2) of Lemma
\ref{lem-dtn-jpj-run}, there is no pair $(j_{i}, j'_{i'})$ such
that $d_{\text{H}}\big(x_{[n]\backslash\{j_{i}\}},
x'_{[n]\backslash\{j'_{i'}\}}\big)=2~($noticing that
$j_{\ell}>j'_{\ell'})$. Moreover, for each $d'\in\{3,4\}$, by 3)
of Lemma \ref{lem-dtn-jpj-run}, there are at most $n-2$ pairs
$(j_{i}, j'_{i'})$ such that
$d_{\text{H}}\big(x_{[n]\backslash\{j_{i}\}},
x'_{[n]\backslash\{j'_{i'}\}}\big)=d'$.
\end{enumerate}
For reference, we list in Table 2 the number of pairs
$(j_{\ell},j'_{\ell'})$ such that
$d_{\text{H}}(x_{[n]\backslash\{j_\ell\}},
x'_{[n]\backslash\{j'_{\ell'}\}}\big)=d'$ for each
$d'\in\{0,1,2,3,4\}$ and each case. By the above discussions, and
by Lemma \ref{lem-sub-int-size} and \eqref{xxp-case2}, we can
obtain
\begin{align}\label{eq2-pf-main-thm}
\left|B^{\text{DS}}_{1,2}(\bm x,\bm
x')\right|&\leq\sum_{i=1}^m\sum_{i'=1}^{m'}\left|E^{j_i}_{j'_{i'}}\right|\nonumber\\
&\leq 6\xi^{q,n-1}_{1,2}+(n+10)\xi^{q,n-1}_{2,2}+(2n+14)\xi^{q,n-1}_{3,2}+(2n+20)\xi^{q,n-1}_{4,2}\nonumber\\
&=6[q(q-1)n-2q^2+3q]+(n+10)[2(q-1)n+q^2-6q+6]+(2n+14)(6q-6)+(2n+20)(6)\nonumber\\
&=2(q-1)n^2+(7q^2+20q-14)n+O_q(1)\nonumber\\
&\leq (q^2-1)n^2-(3q^2+5q-5)n+O_q(1)
\end{align}
where the last inequality holds for sufficiently large $n$.

\begin{table}[htbp]
\begin{center}
\small
\renewcommand\arraystretch{1.4}
\begin{tabular}{|p{4.6cm}|p{1.3cm}|p{1.3cm}|p{1.3cm}|p{1.3cm}|p{1.3cm}|p{1.3cm}|}
\hline  ~ & $d'=0$ & $d'=1$ & $d'=2$ & $d'=3$ & $d'=4$   \\
\hline $j_{i}\leq j'_{i'}$ and $\big|\{i_1,i_2\}\backslash[j'_{i'}, j_{i}]\big|=0$ & $~~~~0$ & $1$ & $2$ & $3$ & $4$   \\
\hline $j_{i}\leq j'_{i'}$ and $\big|\{i_1,i_2\}\backslash[j'_{i'}, j_{i}]\big|=1$ & $~~~~0$ & $2$ & $4$ & $6$ & $8$   \\
\hline $j_{i}\leq j'_{i'}$ and $\big|\{i_1,i_2\}\backslash[j'_{i'}, j_{i}]\big|=2$ & $~~~~0$ & $0$ & $n-2$ & $n-2$ & $n-2$   \\
\hline $j_{i}>j'_{i'}$ and $\big|\{i_1,i_2\}\backslash[j'_{i'}, j_{i}]\big|=0$     & $~~~~0$ & $1$ & $2$ & $3$ & $4$   \\
\hline $j_{i}>j'_{i'}$ and $\big|\{i_1,i_2\}\backslash[j'_{i'}, j_{i}]\big|=1$     & $~~~~0$ & $2$ & $4$ & $6$ & $8$   \\
\hline $j_{i}>j'_{i'}$ and $\big|\{i_1,i_2\}\backslash[j'_{i'}, j_{i}]\big|=2$     & $~~~~0$ & $0$ & $0$ & $n-2$ & $n-2$   \\
\hline
\end{tabular}\\
\end{center}
\vspace{0.0mm} Table 2. The upper bound of the number of pairs
$(j_{\ell},j'_{\ell'})$ such that
$d_{\text{H}}(x_{[n]\backslash\{j_\ell\}},
x'_{[n]\backslash\{j'_{\ell'}\}}\big)=d'$ for each
$d'\in\{0,1,2,3,4\}$ and each case when $d=2$,
$x_{[n]\backslash\{i_1\}} \neq x'_{[n]\backslash\{i_2\}}$ and
$x_{[n]\backslash\{i_2\}}\neq x'_{[n]\backslash\{i_1\}}$. The
total number of pairs $(j_{\ell},j'_{\ell'})$ such that
$d_{\text{H}}(x_{[n]\backslash\{j_\ell\}},
x'_{[n]\backslash\{j'_{\ell'}\}}\big)=d'$ is the sum of the
corresponding column. For example, the total number of pairs
$(j_{\ell},j'_{\ell'})$ such that
$d_{\text{H}}(x_{[n]\backslash\{j_\ell\}},
x'_{[n]\backslash\{j'_{\ell'}\}}\big)=d'=2$ is
$2+4+(n-2)+2+4=n+10$.
\end{table}

\section{Intersection of error balls of sequences with
Hamming distance $d\geq 3$}

Throughout this section, we suppose $\bm x,\bm x'\in\Sigma_q^n$
with $d=d_{\text{H}}(\bm x, \bm x')\geq 3$. Let $1\leq
j_1<j_2<\cdots<j_m\leq n$ be such that $x_{[j_{i-1}+1,j_i]}$,
$i=1,2,\cdots, m$, are all runs of $\bm x$, and $1\leq
j'_1<j'_2<\cdots<j'_{m'}\leq n$ be such that
$x'_{[j'_{i'-1}+1,j'_{i'}]}$, $i'=1,2,\cdots, m'$, are all runs of
$\bm x'$, where we set $j_0=j'_0=0$. We will prove
$\left|B^{\text{DS}}_{1,2}(\bm x,\bm x')\right|\leq
(q^2-1)n^2-(3q^2+5q-5)n+O_q(1)$. The basic idea of the proof is
the same as that in Section III. C. First, for each
$d'\in\{0,1,2,3,4\}$, we can use Lemma \ref{lem-dtn-jpj-run} to
estimate the number of pairs $(j_{i}, j'_{i'})$ such that
$d_{\text{H}}\big(x_{[n]\backslash\{j_{i}\}},
x'_{[n]\backslash\{j'_{i'}\}}\big)=d'$. Then by Observation 1, and
by Lemma \ref{lem-sub-int-size} and \eqref{xxp-case2}, we can find
$\left|B^{\text{DS}}_{1,2}(\bm x,\bm
x')\right|=\left|\bigcup_{i=1}^m\bigcup_{i'=1}^{m'}E^{j_i}_{j'_{i'}}\right|
\leq\sum_{i=1}^m\sum_{i'=1}^{m'}\left|E^{j_i}_{j'_{i'}}\right|\leq
(q^2-1)n^2-(3q^2+5q-5)n+O_q(1)$.

We divide our proof into three cases: $d=3$, $d=4$ and $d\geq 5$.

\subsection{Intersection of error balls of sequences with Hamming
distance $d=3$}

Suppose $\text{supp}(\bm x-\bm x')=\{i_1,i_2,i_3\}$ and
$i_1<i_2<i_3$. For each $d'\in\{0,1,2,3,4\}$, we will use Lemma
\ref{lem-dtn-jpj-run} to estimate the number of pairs $(j_{i},
j'_{i'})$ such that $d_{\text{H}}\big(x_{[n]\backslash\{j_{i}\}},
x'_{[n]\backslash\{j'_{i'}\}}\big)=d'$. First, consider $j_{i}\leq
j'_{i'}$. We have the following four cases.
\begin{enumerate}
\item[1.] $\big|\{i_1,i_2,i_3\}\backslash[j_{i}, j'_{i'}]\big|=0$.
Then it must be $j_{i}\leq i_1<i_3\leq j'_{i'}$ and
$S_{[i_1,i_{3}]}=0$. Clearly, by definition, we have
$\Delta_{[i_1,i_{3}]}\geq 0$. By 4) of Lemma \ref{lem-dtn-jpj-run}
$($with $\lambda=1$ and $\lambda'=3)$, for each
$d'\in\{0,1,2,3,4\}$, there are at most $d'+1$ pairs $(j_{i},
j'_{i'})$ such that $d_{\text{H}}\big(x_{[n]\backslash\{j_{i}\}},
x'_{[n]\backslash\{j'_{i'}\}}\big)=d'$.

\item[2.] $\big|\{i_1,i_2,i_3\}\backslash[j_{i}, j'_{i'}]\big|=1$.
Then it must be $j_{i}\leq i_1<i_2\leq j'_{i'}<i_3~($then
$\lambda=1$ and $\lambda'=2)$ or $i_1<j_{i}\leq i_2<i_3\leq
j'_{i'}~($then $\lambda=2$ and $\lambda'=3)$. For each case, we
have $S_{[i_\lambda,i_{\lambda'}]}=1$ and
$\Delta_{[i_\lambda,i_{\lambda'}]}\geq 0$. By 4) of Lemma
\ref{lem-dtn-jpj-run}, there is no pair $(j_{\ell},j'_{\ell'})$
such that $d_{\text{H}}(x_{[n]\backslash\{j_\ell\}},
x'_{[n]\backslash\{j'_{\ell'}\}}\big)=0$. Moreover, for each case
and each $d'\in\{1,2,3,4\}$, there are at most $d'$ pairs $(j_{i},
j'_{i'})$ such that $d_{\text{H}}\big(x_{[n]\backslash\{j_{i}\}},
x'_{[n]\backslash\{j'_{i'}\}}\big)=d'$.

\item[3.] $\big|\{i_1,i_2,i_3\}\backslash[j_{i}, j'_{i'}]\big|=2$.
Then it must be one of the following cases: $j_{i}\leq i_1\leq
j'_{i'}<i_2~($then $\lambda=\lambda'=1)$; $i_1<j_{i}\leq i_2\leq
j'_{i'}<i_3~($then $\lambda=\lambda'=2)$; or $i_2<j_{i}\leq
i_3\leq j'_{i'}~($then $\lambda=\lambda'=3)$. For each case, we
have $S_{[i_\lambda,i_{\lambda'}]}=2$ and
$\Delta_{[i_\lambda,i_{\lambda'}]}\geq 0$. By 4) of Lemma
\ref{lem-dtn-jpj-run}, for each case and each $d'\in\{0,1\}$,
there is no pair $(j_{\ell},j'_{\ell'})$ such that
$d_{\text{H}}(x_{[n]\backslash\{j_\ell\}},
x'_{[n]\backslash\{j'_{\ell'}\}}\big)=d'$. Moreover, for each case
and each $d'\in\{2,3,4\}$, there are at most $d'-1$ pairs $(j_{i},
j'_{i'})$ such that $d_{\text{H}}\big(x_{[n]\backslash\{j_{i}\}},
x'_{[n]\backslash\{j'_{i'}\}}\big)=d'$.

\item[4.] $\big|\{i_1,i_2,i_3\}\backslash[j_{i}, j'_{i'}]\big|=3$.
Then we have $\{i_1,i_2,i_3\}\cap[j_{i}, j'_{i'}]=\emptyset$. For
$d'\in\{0,1,2\}$, by 1) of Lemma \ref{lem-dtn-jpj-run}, there is
no pair $(j_{\ell},j'_{\ell'})$ such that
$d_{\text{H}}(x_{[n]\backslash\{j_\ell\}},
x'_{[n]\backslash\{j'_{\ell'}\}}\big)=d'$. By 2) of Lemma
\ref{lem-dtn-jpj-run}, there are at most $n-3$ pairs $(j_{i},
j'_{i'})$ such that $d_{\text{H}}\big(x_{[n]\backslash\{j_{i}\}},
x'_{[n]\backslash\{j'_{i'}\}}\big)=3$. Moreover, by 3) of Lemma
\ref{lem-dtn-jpj-run}, there are at most $n-3$ pairs $(j_{i},
j'_{i'})$ such that $d_{\text{H}}\big(x_{[n]\backslash\{j_{i}\}},
x'_{[n]\backslash\{j'_{i'}\}}\big)=4$.
\end{enumerate}
Similarly, for $j_{i}>j'_{i'}$, we have the following four cases.
\begin{enumerate}
\item[1.] $\big|\{i_1,i_2,i_3\}\backslash[j'_{i'}, j_{i}]\big|=0$.
Then $j'_{i'}\leq i_1<i_3\leq j_{i}$ and we have
$S_{[i_1,i_{3}]}=0$ and $\Delta'_{[i_1,i_{3}]}\geq 0$. By 5) of
Lemma \ref{lem-dtn-jpj-run} $($with $\lambda=1$ and $\lambda'=3)$,
for each $d'\in\{0,1,2,3,4\}$, there are at most $d'+1$ pairs
$(j_{i}, j'_{i'})$ such that
$d_{\text{H}}\big(x_{[n]\backslash\{j_{i}\}},
x'_{[n]\backslash\{j'_{i'}\}}\big)=d'$.

\item[2.] $\big|\{i_1,i_2,i_3\}\backslash[j'_{i'}, j_{i}]\big|=1$.
Then $j'_{i'}\leq i_1<i_2\leq j_{i}<i_3~($then $\lambda=1$ and
$\lambda'=2)$ or $i_1<j'_{i'}\leq i_2<i_3\leq j_{i}~($then
$\lambda=2$ and $\lambda'=3)$. For each case, we have
$S_{[i_\lambda,i_{\lambda'}]}=1$ and
$\Delta'_{[i_\lambda,i_{\lambda'}]}\geq 0$. By 5) of Lemma
\ref{lem-dtn-jpj-run}, there is no pair $(j_{\ell},j'_{\ell'})$
such that $d_{\text{H}}(x_{[n]\backslash\{j_\ell\}},
x'_{[n]\backslash\{j'_{\ell'}\}}\big)=0$. Moreover, for each case
and each $d'\in\{1,2,3,4\}$, there are at most $d'$ pairs $(j_{i},
j'_{i'})$ such that $d_{\text{H}}\big(x_{[n]\backslash\{j_{i}\}},
x'_{[n]\backslash\{j'_{i'}\}}\big)=d'$.

\item[3.] $\big|\{i_1,i_2,i_3\}\backslash[j'_{i'}, j_{i}]\big|=2$.
Then it must be one of the following cases: $j'_{i'}\leq i_1\leq
j_{i}<i_2~($then $\lambda=\lambda'=1)$; $i_1<j'_{i'}\leq i_2\leq
j_{i}<i_3~($then $\lambda=\lambda'=2)$; or $i_2<j'_{i'}\leq
i_3\leq j_{i}~($then $\lambda=\lambda'=3)$. For each case, we have
$S_{[i_\lambda,i_{\lambda'}]}=2$ and
$\Delta_{[i_\lambda,i_{\lambda'}]}\geq 0$. By 5) of Lemma
\ref{lem-dtn-jpj-run}, for each case and each $d'\in\{0,1\}$,
there is no pair $(j_{\ell},j'_{\ell'})$ such that
$d_{\text{H}}(x_{[n]\backslash\{j_\ell\}},
x'_{[n]\backslash\{j'_{\ell'}\}}\big)=d'$. Moreover, for each case
and each $d'\in\{2,3,4\}$, there are at most $d'-1$ pairs $(j_{i},
j'_{i'})$ such that $d_{\text{H}}\big(x_{[n]\backslash\{j_{i}\}},
x'_{[n]\backslash\{j'_{i'}\}}\big)=d'$.

\item[4.] $\big|\{i_1,i_2,i_3\}\backslash[j'_{i'}, j_{i}]\big|=3$.
Then we have $\{i_1,i_2,i_3\}\cap[j'_{i'}, j_{i}]=\emptyset$. For
$d'\in\{0,1,2\}$, by 1) of Lemma \ref{lem-dtn-jpj-run}, there is
no pair $(j_{\ell},j'_{\ell'})$ such that
$d_{\text{H}}(x_{[n]\backslash\{j_\ell\}},
x'_{[n]\backslash\{j'_{\ell'}\}}\big)=d'$. By 2) of Lemma
\ref{lem-dtn-jpj-run}, there is no pair $(j_{i}, j'_{i'})$ such
that $d_{\text{H}}\big(x_{[n]\backslash\{j_{i}\}},
x'_{[n]\backslash\{j'_{i'}\}}\big)=3~($noticing that
$j_i>j'_{i'})$. Moreover, by 3) of Lemma \ref{lem-dtn-jpj-run},
there are at most $n-3$ pairs $(j_{i}, j'_{i'})$ such that
$d_{\text{H}}\big(x_{[n]\backslash\{j_{i}\}},
x'_{[n]\backslash\{j'_{i'}\}}\big)=4$.
\end{enumerate}
For reference, we list in Table 3 the number of pairs
$(j_{\ell},j'_{\ell'})$ such that
$d_{\text{H}}(x_{[n]\backslash\{j_\ell\}},
x'_{[n]\backslash\{j'_{\ell'}\}}\big)=d'$ for each
$d'\in\{0,1,2,3,4\}$ and each case. By the above discussions, and
by Lemma \ref{lem-sub-int-size} and \eqref{xxp-case2}, we can
obtain
\begin{align}\label{eq3-pf-main-thm}
\left|B^{\text{DS}}_{1,2}(\bm x,\bm
x')\right|&\leq\sum_{i=1}^m\sum_{i'=1}^{m'}\left|E^{j_i}_{j'_{i'}}\right|\nonumber\\
&\leq
2\xi^{q,n-1}_{0,2}+8\xi^{q,n-1}_{1,2}+20\xi^{q,n-1}_{2,2}+(n+29)\xi^{q,n-1}_{3,2}+(2n+38)\xi^{q,n-1}_{4,2}\nonumber\\
&=2\left[\frac{(q-1)^2}{2}n^2-\frac{(3q-5)(q-1)}{2}n+(q^2-3q+3)\right]+8[q(q-1)n-2q^2+3q]\nonumber\\
&~~~+20[2(q-1)n+q^2-6q+6]+(n+29)(6q-6)+(2n+38)(6)\nonumber\\
&=(q-1)^2n^2+(5q^2+46q-39)n+O_q(1)\nonumber\\
&\leq (q^2-1)n^2-(3q^2+5q-5)n+O_q(1).
\end{align}

\begin{table}[htbp]
\begin{center}
\small
\renewcommand\arraystretch{1.4}
\begin{tabular}{|p{5.2cm}|p{1.3cm}|p{1.3cm}|p{1.3cm}|p{1.3cm}|p{1.3cm}|p{1.3cm}|}
\hline  ~ & $d'=0$ & $d'=1$ & $d'=2$ & $d'=3$ & $d'=4$   \\
\hline $j_{i}\leq j'_{i'}$ and $\big|\{i_1,i_2,i_3\}\backslash[j_{i}, j'_{i'}]\big|=0$ & $~~~~1$ & $2$ & $3$ & $4$ & $5$   \\
\hline $j_{i}\leq j'_{i'}$ and $\big|\{i_1,i_2,i_3\}\backslash[j_{i}, j'_{i'}]\big|=1$ & $~~~~0$ & $2$ & $4$ & $6$ & $8$   \\
\hline $j_{i}\leq j'_{i'}$ and $\big|\{i_1,i_2,i_3\}\backslash[j_{i}, j'_{i'}]\big|=2$ & $~~~~0$ & $0$ & $3$ & $6$ & $9$   \\
\hline $j_{i}\leq j'_{i'}$ and $\big|\{i_1,i_2,i_3\}\backslash[j_{i}, j'_{i'}]\big|=3$ & $~~~~0$ & $0$ & $0$ & $n-3$ & $n-3$   \\
\hline $j_{i}>j'_{i'}$ and $\big|\{i_1,i_2,i_3\}\backslash[j'_{i'}, j_{i}]\big|=0$     & $~~~~1$ & $2$ & $3$ & $4$ & $5$   \\
\hline $j_{i}>j'_{i'}$ and $\big|\{i_1,i_2,i_3\}\backslash[j'_{i'}, j_{i}]\big|=1$     & $~~~~0$ & $2$ & $4$ & $6$ & $8$   \\
\hline $j_{i}>j'_{i'}$ and $\big|\{i_1,i_2,i_3\}\backslash[j'_{i'}, j_{i}]\big|=2$     & $~~~~0$ & $0$ & $3$ & $6$ & $9$   \\
\hline $j_{i}>j'_{i'}$ and $\big|\{i_1,i_2,i_3\}\backslash[j'_{i'}, j_{i}]\big|=3$     & $~~~~0$ & $0$ & $0$ & $0$ & $n-3$   \\
\hline
\end{tabular}\\
\end{center}
\vspace{0.0mm} Table 3. The upper bound of the number of pairs
$(j_{\ell},j'_{\ell'})$ such that
$d_{\text{H}}(x_{[n]\backslash\{j_\ell\}},
x'_{[n]\backslash\{j'_{\ell'}\}}\big)=d'$ for each
$d'\in\{0,1,2,3,4\}$ and each case when $d=3$.
\end{table}

\subsection{Intersection of error balls of sequences with Hamming
distance $d=4$}

Suppose $\text{supp}(\bm x-\bm x')=\{i_1,i_2,i_3,i_4\}$ and
$i_1<i_2<i_3<i_4$. For each $d'\in\{0,1,2,3,4\}$, we will use
Lemma \ref{lem-dtn-jpj-run} to estimate the number of pairs
$(j_{i}, j'_{i'})$ such that
$d_{\text{H}}\big(x_{[n]\backslash\{j_{i}\}},
x'_{[n]\backslash\{j'_{i'}\}}\big)=d'$. For $j_{i}\leq j'_{i'}$,
we have the following five cases.
\begin{enumerate}
 \item[1.] $|\{i_1,i_2,i_3,i_4\}\backslash[j_{i}, j'_{i'}]|=0$.
 Then it must be $j_{i}\leq i_1<i_4\leq j'_{i'}$.
 For each $d'\in\{0,1,2,3,4\}$, by 4) of
 Lemma \ref{lem-dtn-jpj-run}, there is at most $d'+1$
 pairs $(j_{i}, j'_{i'})$ such that
 $d_{\text{H}}\big(x_{[n]\backslash\{j_{i}\}},
 x'_{[n]\backslash\{j'_{i'}\}}\big)=d'$.
 \item[2.] $|\{i_1,i_2,i_3,i_4\}\backslash[j_{i}, j'_{i'}]|=1$.
 Then it must be $j_{i}\leq i_1<i_{3}\leq j'_{i'}<i_4$ or
 $i_1<j_{i}\leq i_2<i_{4}\leq j'_{i'}$. By 4) of Lemma
 \ref{lem-dtn-jpj-run}, for each case, there is no
 pair $(j_{i}, j'_{i'})$ such that
 $d_{\text{H}}\big(x_{[n]\backslash\{j_{i}\}},
 x'_{[n]\backslash\{j'_{i'}\}}\big)=0$. Moreover, for each case and each $d'\in\{1,2,3,4\}$,
 there is at most $d'$ pairs $(j_{i}, j'_{i'})$ such that
 $d_{\text{H}}\big(x_{[n]\backslash\{j_{i}\}},
 x'_{[n]\backslash\{j'_{i'}\}}\big)=d'$.
 \item[3.] $|\{i_1,i_2,i_3,i_4\}\backslash[j_{i}, j'_{i'}]|=2$.
 Then it must be one of the following cases: $j_{i}\leq i_1<i_{2}\leq
 j'_{i'}<i_{3}$;
 $i_1<j_{i}\leq i_2<i_{3}\leq j'_{i'}<i_4$; or
 $i_2<j_{i}\leq i_3<i_{4}\leq j'_{i'}$. By 4) of Lemma
 \ref{lem-dtn-jpj-run}, for each case and each $d'\in\{0,1\}$, there is no
 pair $(j_{i}, j'_{i'})$ such that
 $d_{\text{H}}\big(x_{[n]\backslash\{j_{i}\}},
 x'_{[n]\backslash\{j'_{i'}\}}\big)=d'$. Moreover, for each case and each $d'\in\{2,3,4\}$,
 there is at most $d'-1$ pairs $(j_{i}, j'_{i'})$ such that
 $d_{\text{H}}\big(x_{[n]\backslash\{j_{i}\}},
 x'_{[n]\backslash\{j'_{i'}\}}\big)=d'$.
 \item[4.] $|\{i_1,i_2,i_3,i_4\}\backslash[j_{i}, j'_{i'}]|=3$.
 Then it must be one of the following cases: $j_{i}\leq i_1\leq
 j'_{i'}<i_{2}$;
 $i_1<j_{i}\leq i_2\leq j'_{i'}<i_{3}$;
 $i_2<j_{i}\leq i_3\leq j'_{i'}<i_{4}$; or
 $i_3<j_{i}\leq i_4\leq j'_{i'}$. By 4) of Lemma
 \ref{lem-dtn-jpj-run}, for each case and each $d'\in\{0,1,2\}$, there is no
 pair $(j_{i}, j'_{i'})$ such that
 $d_{\text{H}}\big(x_{[n]\backslash\{j_{i}\}},
 x'_{[n]\backslash\{j'_{i'}\}}\big)=d'$. Moreover, for each case and each $d'\in\{3,4\}$,
 there is at most $d'-2$ pairs $(j_{i}, j'_{i'})$ such that
 $d_{\text{H}}\big(x_{[n]\backslash\{j_{i}\}},
 x'_{[n]\backslash\{j'_{i'}\}}\big)=d'$.
 \item[5.] $|\{i_1,i_2,i_3,i_4\}\backslash[j_{i}, j'_{i'}]|=4$.
 Then it must be one of the following cases: $j_{i}\leq
 j'_{i'}<i_{1}$; $i_1<j_{i}\leq j'_{i'}<i_{2}$;
 $i_2<j_{i}\leq j'_{i'}<i_{3}$; $i_3<j_{i}\leq j'_{i'}<i_{4}$; or
 $i_4<j_{i}\leq j'_{i'}$. By 4) of Lemma
 \ref{lem-dtn-jpj-run}, for each case and each $d'\in\{0,1,2,3\}$, there is no
 pair $(j_{i}, j'_{i'})$ such that
 $d_{\text{H}}\big(x_{[n]\backslash\{j_{i}\}},
 x'_{[n]\backslash\{j'_{i'}\}}\big)=d'$. Moreover,
 there is at most one pair $(j_{i}, j'_{i'})$ such that
 $d_{\text{H}}\big(x_{[n]\backslash\{j_{i}\}},
 x'_{[n]\backslash\{j'_{i'}\}}\big)=4$.
\end{enumerate}
For $j_{i}>j'_{i'}$ and $d'\in\{0,1,2,3\}$, by 5) of Lemma
\ref{lem-dtn-jpj-run}, we can obtain the same number of pairs
$(j_{i}, j'_{i'})$ such that
$d_{\text{H}}\big(x_{[n]\backslash\{j_{i}\}},
x'_{[n]\backslash\{j'_{i'}\}}\big)=d'$. Moreover, by 1)$-$3) of
Lemma \ref{lem-dtn-jpj-run}, there is no pair $(j_{i}, j'_{i'})$
such that $j_{i}>j'_{i'}$ and
$d_{\text{H}}\big(x_{[n]\backslash\{j_{i}\}},
x'_{[n]\backslash\{j'_{i'}\}}\big)=4$. For reference, we list in
Table 4 the number of pairs $(j_{\ell},j'_{\ell'})$ such that
$d_{\text{H}}(x_{[n]\backslash\{j_\ell\}},
x'_{[n]\backslash\{j'_{\ell'}\}}\big)=d'$ for each
$d'\in\{0,1,2,3,4\}$ and each case. Thus, by Lemma
\ref{lem-sub-int-size} and \eqref{xxp-case2}, we can obtain
\begin{align}\label{eq4-pf-main-thm}
\left|B^{\text{DS}}_{1,2}(\bm x,\bm
x')\right|&\leq\sum_{i=1}^m\sum_{i'=1}^{m'}\left|E^{j_i}_{j'_{i'}}\right|\nonumber\\
&\leq
2\xi^{q,n-1}_{0,2}+8\xi^{q,n-1}_{1,2}+20\xi^{q,n-1}_{2,2}+40\xi^{q,n-1}_{3,2}+(n+56)\xi^{q,n-1}_{4,2}\nonumber\\
&=2\left[\frac{(q-1)^2}{2}n^2-\frac{(3q-5)(q-1)}{2}n+(q^2-3q+3)\right]+8[q(q-1)n-2q^2+3q]\nonumber\\
&~~~+20[2(q-1)n+q^2-6q+6]+40(6q-6)+(n+56)(6)\nonumber\\
&=(q-1)^2n^2+(5q^2+40q-39)n+O_q(1)\nonumber\\
&\leq (q^2-1)n^2-(3q^2+5q-5)n+O_q(1).
\end{align}

\begin{table}[htbp]
\begin{center}
\small
\renewcommand\arraystretch{1.4}
\begin{tabular}{|p{5.6cm}|p{1.3cm}|p{1.3cm}|p{1.3cm}|p{1.3cm}|p{1.3cm}|p{1.3cm}|}
\hline  ~ & $d'=0$ & $d'=1$ & $d'=2$ & $d'=3$ & $d'=4$   \\
\hline $j_{i}\leq j'_{i'}$ and $\big|\{i_1,i_2,i_3,i_4\}\backslash[j_{i}, j'_{i'}]\big|=0$ & $~~~~1$ & $2$ & $3$ & $4$ & $5$   \\
\hline $j_{i}\leq j'_{i'}$ and $\big|\{i_1,i_2,i_3,i_4\}\backslash[j_{i}, j'_{i'}]\big|=1$ & $~~~~0$ & $2$ & $4$ & $6$ & $8$   \\
\hline $j_{i}\leq j'_{i'}$ and $\big|\{i_1,i_2,i_3,i_4\}\backslash[j_{i}, j'_{i'}]\big|=2$ & $~~~~0$ & $0$ & $3$ & $6$ & $9$   \\
\hline $j_{i}\leq j'_{i'}$ and $\big|\{i_1,i_2,i_3,i_4\}\backslash[j_{i}, j'_{i'}]\big|=3$ & $~~~~0$ & $0$ & $0$ & $4$ & $8$   \\
\hline $j_{i}\leq j'_{i'}$ and $\big|\{i_1,i_2,i_3,i_4\}\backslash[j_{i}, j'_{i'}]\big|=4$ & $~~~~0$ & $0$ & $0$ & $0$ & $n-4$   \\
\hline $j_{i}>j'_{i'}$ and $\big|\{i_1,i_2,i_3,i_4\}\backslash[j'_{i'}, j_{i}]\big|=0$     & $~~~~1$ & $2$ & $3$ & $4$ & $5$   \\
\hline $j_{i}>j'_{i'}$ and $\big|\{i_1,i_2,i_3,i_4\}\backslash[j'_{i'}, j_{i}]\big|=1$     & $~~~~0$ & $2$ & $4$ & $6$ & $8$   \\
\hline $j_{i}>j'_{i'}$ and $\big|\{i_1,i_2,i_3,i_4\}\backslash[j'_{i'}, j_{i}]\big|=2$     & $~~~~0$ & $0$ & $3$ & $6$ & $9$   \\
\hline $j_{i}>j'_{i'}$ and $\big|\{i_1,i_2,i_3,i_4\}\backslash[j'_{i'}, j_{i}]\big|=3$     & $~~~~0$ & $0$ & $0$ & $4$ & $8$   \\
\hline $j_{i}>j'_{i'}$ and $\big|\{i_1,i_2,i_3,i_4\}\backslash[j'_{i'}, j_{i}]\big|=4$     & $~~~~0$ & $0$ & $0$ & $0$ & $0$   \\
\hline
\end{tabular}\\
\end{center}
\vspace{0.0mm} Table 4. The upper bound of the number of pairs
$(j_{\ell},j'_{\ell'})$ such that
$d_{\text{H}}(x_{[n]\backslash\{j_\ell\}},
x'_{[n]\backslash\{j'_{\ell'}\}}\big)=d'$ for each
$d'\in\{0,1,2,3,4\}$ and each case when $d=4$.
\end{table}

\subsection{Intersection of error balls of sequences with Hamming
distance $d\geq 5$}

Suppose $\text{supp}(\bm x-\bm x')=\{i_1,i_2,\cdots, i_d\}$ and
$i_1<i_2<\cdots<i_d$. For each $d'\in\{0,1,2,3,4\}$, we will use
Lemma \ref{lem-dtn-jpj-run} to estimate the number of pairs
$(j_{i}, j'_{i'})$ such that
$d_{\text{H}}\big(x_{[n]\backslash\{j_{i}\}},
x'_{[n]\backslash\{j'_{i'}\}}\big)=d'$. For $j_{i}\leq j'_{i'}$,
we have the following five cases.
\begin{enumerate}
 \item[1.] $|\{i_1,i_2,\cdots,i_d\}\backslash[j_{i}, j'_{i'}]|=0$.
 Then it must be $j_{i}\leq i_1<i_d\leq j'_{i'}$.
 For each $d'\in\{0,1,2,3,4\}$,
 by 4) of Lemma \ref{lem-dtn-jpj-run},
 there is at most $d'+1$ pair $(j_{i}, j'_{i'})$ such that
 $d_{\text{H}}\big(x_{[n]\backslash\{j_{i}\}},
 x'_{[n]\backslash\{j'_{i'}\}}\big)=d'$.
 \item[2.] $|\{i_1,i_2,\cdots,i_d\}\backslash[j_{i}, j'_{i'}]|=1$.
 Then it must be $j_{i}\leq i_1<i_{d-1}\leq j'_{i'}<i_d$ or
 $i_1<j_{i}\leq i_2<i_{d}\leq j'_{i'}$.
 By 4) of Lemma \ref{lem-dtn-jpj-run}, for each case, there is no
 pair $(j_{i}, j'_{i'})$ such that
 $d_{\text{H}}\big(x_{[n]\backslash\{j_{i}\}},
 x'_{[n]\backslash\{j'_{i'}\}}\big)=0$. Moreover, for each case
 and each $d'\in\{1,2,3,4\}$, there is at most $d'$
 pairs $(j_{i}, j'_{i'})$ such that
 $d_{\text{H}}\big(x_{[n]\backslash\{j_{i}\}},
 x'_{[n]\backslash\{j'_{i'}\}}\big)=d'$.
 \item[3.] $|\{i_1,i_2,\cdots,i_d\}\backslash[j_{i}, j'_{i'}]|=2$.
 Then it must be one of the following cases: $j_{i}\leq i_1<i_{d-2}\leq
 j'_{i'}<i_{d-1}$;
 $i_1<j_{i}\leq i_2<i_{d-1}\leq j'_{i'}<i_d$; or
 $i_2<j_{i}\leq i_3<i_{d}\leq j'_{i'}$. By 4) of Lemma
 \ref{lem-dtn-jpj-run}, for each case and each $d'\in\{0,1\}$, there is no
 pair $(j_{i}, j'_{i'})$ such that
 $d_{\text{H}}\big(x_{[n]\backslash\{j_{i}\}},
 x'_{[n]\backslash\{j'_{i'}\}}\big)=d'$. Moreover, for each case
 and each $d'\in\{2,3,4\}$, there is at most $d'-1$
 pairs $(j_{i}, j'_{i'})$ such that
 $d_{\text{H}}\big(x_{[n]\backslash\{j_{i}\}},
 x'_{[n]\backslash\{j'_{i'}\}}\big)=d'$.
 \item[4.] $|\{i_1,i_2,\cdots,i_d\}\backslash[j_{i}, j'_{i'}]|=3$.
 Then it must be one of the following cases: $j_{i}\leq i_1<i_{d-3}\leq
 j'_{i'}<i_{d-2}$;
 $i_1<j_{i}\leq i_2<i_{d-2}\leq j'_{i'}<i_{d-1}$;
 $i_2<j_{i}\leq i_3<i_{d-1}\leq j'_{i'}<i_{d}$; or
 $i_3<j_{i}\leq i_4<i_{d}\leq j'_{i'}$. By 4) of Lemma
 \ref{lem-dtn-jpj-run}, for each case and each $d'\in\{0,1,2\}$, there is no
 pair $(j_{i}, j'_{i'})$ such that
 $d_{\text{H}}\big(x_{[n]\backslash\{j_{i}\}},
 x'_{[n]\backslash\{j'_{i'}\}}\big)=d'$. Moreover, for each case
 and each $d'\in\{3,4\}$, there is at most $d'-2$
 pairs $(j_{i}, j'_{i'})$ such that $d_{\text{H}}\big(x_{[n]\backslash\{j_{i}\}},
 x'_{[n]\backslash\{j'_{i'}\}}\big)=d'$.
 \item[5.] $|\{i_1,i_2,\cdots,i_d\}\backslash[j_{i}, j'_{i'}]|=4$.
 Then it must be one of the following cases: $j_{i}\leq i_1\leq i_{d-4}\leq
 j'_{i'}<i_{d-3}$;
 $i_1<j_{i}\leq i_2\leq i_{d-3}\leq j'_{i'}<i_{d-2}$;
 $i_2<j_{i}\leq i_3\leq i_{d-2}\leq j'_{i'}<i_{d-1}$;
 $i_3<j_{i}\leq i_4\leq i_{d-1}\leq j'_{i'}<i_{d}$; or
 $i_4<j_{i}\leq i_5\leq i_{d}\leq j'_{i'}$. By 4) of Lemma
 \ref{lem-dtn-jpj-run}, for each case and each $d'\in\{0,1,2,3\}$, there is no
 pair $(j_{i}, j'_{i'})$ such that
 $d_{\text{H}}\big(x_{[n]\backslash\{j_{i}\}},
 x'_{[n]\backslash\{j'_{i'}\}}\big)=d'$. Moreover, for each case,
 there is at most one pair $(j_{i}, j'_{i'})$ such that
 $d_{\text{H}}\big(x_{[n]\backslash\{j_{i}\}},
 x'_{[n]\backslash\{j'_{i'}\}}\big)=4$.
\end{enumerate}
For $j_{i}>j'_{i'}$, by 5) of Lemma \ref{lem-dtn-jpj-run}, we can
obtain the same number of pairs $(j_{i}, j'_{i'})$ such that
$d_{\text{H}}\big(x_{[n]\backslash\{j_{i}\}},
 x'_{[n]\backslash\{j'_{i'}\}}\big)=d'$ for each $d'\in\{0,1,2,3,4\}$.
For reference, we list in Table 5 the number of pairs
$(j_{\ell},j'_{\ell'})$ such that
$d_{\text{H}}(x_{[n]\backslash\{j_\ell\}},
x'_{[n]\backslash\{j'_{\ell'}\}}\big)=d'$ for each
$d'\in\{0,1,2,3,4\}$ and each case. Thus, by Lemma
\ref{lem-sub-int-size} and \eqref{xxp-case2}, we can obtain
\begin{align}\label{eq5-pf-main-thm}
\left|B^{\text{DS}}_{1,2}(\bm x,\bm
x')\right|&\leq\sum_{i=1}^m\sum_{i'=1}^{m'}\left|E^{j_i}_{j'_{i'}}\right|\nonumber\\
&\leq
2\xi^{q,n-1}_{0,2}+8\xi^{q,n-1}_{1,2}+20\xi^{q,n-1}_{2,2}+40\xi^{q,n-1}_{3,2}+70\xi^{q,n-1}_{4,2}\nonumber\\
&=2\left[\frac{(q-1)^2}{2}n^2-\frac{(3q-5)(q-1)}{2}n+(q^2-3q+3)\right]+8[q(q-1)n-2q^2+3q]\nonumber\\
&~~~+20[2(q-1)n+q^2-6q+6]+40(6q-6)+70(6)\nonumber\\
&=(q-1)^2n^2+(5q^2+40q-45)n+O_q(1)\nonumber\\
&\leq (q^2-1)n^2-(3q^2+5q-5)n+O_q(1).
\end{align}

In our proof, we assume that $n$ is large enough such that the
last inequalities in
\eqref{eq3-pf-main-thm}$-$\eqref{eq5-pf-main-thm} hold.

\begin{table}[htbp]
\begin{center}
\small
\renewcommand\arraystretch{1.4}
\begin{tabular}{|p{5.7cm}|p{1.3cm}|p{1.3cm}|p{1.3cm}|p{1.3cm}|p{1.3cm}|p{1.3cm}|}
\hline  ~ & $d'=0$ & $d'=1$ & $d'=2$ & $d'=3$ & $d'=4$   \\
\hline $j_{i}\leq j'_{i'}$ and $\big|\{i_1,i_2,\cdots,i_d\}\backslash[j_{i}, j'_{i'}]\big|=0$ & $~~~~1$ & $2$ & $3$ & $4$ & $5$   \\
\hline $j_{i}\leq j'_{i'}$ and $\big|\{i_1,i_2,\cdots,i_d\}\backslash[j_{i}, j'_{i'}]\big|=1$ & $~~~~0$ & $2$ & $4$ & $6$ & $8$   \\
\hline $j_{i}\leq j'_{i'}$ and $\big|\{i_1,i_2,\cdots,i_d\}\backslash[j_{i}, j'_{i'}]\big|=2$ & $~~~~0$ & $0$ & $3$ & $6$ & $9$   \\
\hline $j_{i}\leq j'_{i'}$ and $\big|\{i_1,i_2,\cdots,i_d\}\backslash[j_{i}, j'_{i'}]\big|=3$ & $~~~~0$ & $0$ & $0$ & $4$ & $8$   \\
\hline $j_{i}\leq j'_{i'}$ and $\big|\{i_1,i_2,\cdots,i_d\}\backslash[j_{i}, j'_{i'}]\big|=4$ & $~~~~0$ & $0$ & $0$ & $0$ & $5$   \\
\hline $j_{i}>j'_{i'}$ and $\big|\{i_1,i_2,\cdots,i_d\}\backslash[j'_{i'}, j_{i}]\big|=0$     & $~~~~1$ & $2$ & $3$ & $4$ & $5$   \\
\hline $j_{i}>j'_{i'}$ and $\big|\{i_1,i_2,\cdots,i_d\}\backslash[j'_{i'}, j_{i}]\big|=1$     & $~~~~0$ & $2$ & $4$ & $6$ & $8$   \\
\hline $j_{i}>j'_{i'}$ and $\big|\{i_1,i_2,\cdots,i_d\}\backslash[j'_{i'}, j_{i}]\big|=2$     & $~~~~0$ & $0$ & $3$ & $6$ & $9$   \\
\hline $j_{i}>j'_{i'}$ and $\big|\{i_1,i_2,\cdots,i_d\}\backslash[j'_{i'}, j_{i}]\big|=3$     & $~~~~0$ & $0$ & $0$ & $4$ & $8$   \\
\hline $j_{i}>j'_{i'}$ and $\big|\{i_1,i_2,\cdots,i_d\}\backslash[j'_{i'}, j_{i}]\big|=4$     & $~~~~0$ & $0$ & $0$ & $0$ & $5$   \\
\hline
\end{tabular}\\
\end{center}
\vspace{0.0mm} Table 5. The upper bound of the number of pairs
$(j_{\ell},j'_{\ell'})$ such that
$d_{\text{H}}(x_{[n]\backslash\{j_\ell\}},
x'_{[n]\backslash\{j'_{\ell'}\}}\big)=d'$ for each
$d'\in\{0,1,2,3,4\}$ and each case when $d=4$.
\end{table}




\appendices

\section{Proof of Claims 1$-$5}

We now prove Claims 1$-$5.

\subsection{Proof of Claim 1}

Since $x_{[n]\backslash\{i_1\}}=x'_{[n]\backslash\{i_2\}}$, we
have
$$E^{i_1}_{i_2}=B^{\text{S}}_{2}\big(x_{[n]\backslash\{i_1\}},
x'_{[n]\backslash\{i_2\}}\big)
=B^{\text{S}}_{2}\big(x_{[n]\backslash\{i_1\}}\big)
=B^{\text{S}}_{2}\big(x'_{[n]\backslash\{i_2\}}\big).$$ On the
other hand, for any $j'\in[n]$,
$E^{i_1}_{j'}=B^{\text{S}}_{2}\big(x_{[n]\backslash\{i_1\}},
x'_{[n]\backslash\{j'\}}\big)\subseteq
B^{\text{S}}_{2}\big(x_{[n]\backslash\{i_1\}}\big)$. So, we have
$$E^{i_1}=\bigcup_{j'\in[n]}E^{i_1}_{j'}=E^{i_1}_{i_2}
=B^{\text{S}}_{2}\big(x_{[n]\backslash\{i_1\}}\big).$$ Similarly,
since $x_{[n]\backslash\{i_2\}} =x'_{[n]\backslash\{i_1\}}$, we
can obtain $E^{i_2}=\bigcup_{j'\in[n]}E^{i_2}_{j'}=E^{i_2}_{i_1}
=B^{\text{S}}_{2}\big(x_{[n]\backslash\{i_2\}}\big)
=B^{\text{S}}_{2}\big(x'_{[n]\backslash\{i_1\}}\big)$. Moreover,
since $i_2=i_1+1$, it is easy to see that
$d_{\text{H}}\big(x_{[n]\backslash\{i_1\}},
x_{[n]\backslash\{i_2\}}\big)=1$, so by Lemma
\ref{lem-sub-int-size}, $|E^{i_1}\cap
E^{i_2}|=\big|B^{\text{S}}_{2}\big(x_{[n]\backslash\{i_1\}},
x_{[n]\backslash\{i_2\}}\big)\big|=\xi^{q,n-1}_{1,2}$. Thus, we
can obtain $$|E^{i_1}\cup
E^{i_2}|=|E^{i_1}|+|E^{i_2}|-|E^{i_1}\cap
E^{i_2}|=2\xi^{q,n-1}_{0,2}-\xi^{q,n-1}_{1,2} =(q-1)^2n^2
-(4q-5)(q-1)n+4q^2-9q+6,$$ where the last equality comes from
\eqref{xxx-SBN} and \eqref{xxp-case2}.

\subsection{Proof of Claim 2}

To prove Claim 2, it suffices to prove the following two claims.

\textbf{Claim 6}: Let $\{j_{i}: i\in[m]\}$, $\{j'_{i'}:
i'\in[m']\}$, $\tau$ and $\tau'$ be defined as in Remark
\ref{rem-form-run}. Then
\begin{enumerate}
 \item[1)] For each $i\in[1,\tau-5]$, we have $E^{j_i}\subseteq
 \bigcup\nolimits_{\ell=\max\{1,i-2\}}^{i+2}
 E^{j_\ell}_{j_\ell}.$
 \item[2)] $E^{j_{\tau-4}}\subseteq\Big(\bigcup\nolimits_{\ell=\tau-6}^{\tau-2}
 E^{j_\ell}_{j_\ell}\Big)\cup E^{i_{1}}$.
 \item[3)] For $i\in\{\tau-3,\tau-2\}$,
 $E^{j_i}\subseteq\Big(\bigcup\nolimits_{\ell=i-2}^{\tau-1}
 E^{j_\ell}_{j_\ell}\Big)\cup(E^{i_{1}}\cup E^{i_{2}})\cup A_i$,
 where $A_i\subseteq\Sigma_q^{n-1}$ and $|A_i|=O_q(1)$.
 \item[4)] $E^{j_{\tau-1}}
 \subseteq\Big(\bigcup\nolimits_{\ell=\tau-3}^{\tau-1}
 E^{j_\ell}_{j_\ell}\Big)\cup(E^{i_{1}}\cup E^{i_{2}})\cup
 E^{j_{\tau-1}}_{j'_{\tau'+2}}\cup A_{\tau-1}$, where
 $A_{\tau-1}\subseteq\Sigma_q^{n-1}$ and $|A_{\tau-1}|=O_q(1)$.
\end{enumerate}

\textbf{Claim 6$'$}: Let $\{j_{i}: i\in[m]\}$, $\{j'_{i'}:
i'\in[m']\}$, $\tau$ and $\tau'$ be defined as in Remark
\ref{rem-form-run}. Then
\begin{enumerate}
 \item[1)] For each $i\in\cup[\tau+6,m]$, we have $E^{j_i}\subseteq
 \bigcup\nolimits_{\ell=i-2}^{\min\{m,i+2\}}
 E^{j_\ell}_{j_\ell}.$
 \item[2)] $E^{j_{\tau+5}}\subseteq\Big(\bigcup\nolimits_{\ell=\tau+3}^{\tau+7}
 E^{j_\ell}_{j_\ell}\Big)\cup E^{i_{2}}$.
 \item[3)] For $i\in\{\tau+3,\tau+4\}$,
 $E^{j_i}\subseteq\Big(\bigcup\nolimits_{\ell=\tau+2}^{i+2}
 E^{j_\ell}_{j_\ell}\Big)\cup(E^{i_{1}}\cup E^{i_{2}})\cup A_i$,
 where $A_i\subseteq\Sigma_q^{n-1}$ and $|A_i|=O_q(1)$.
 \item[4)] $E^{j_{\tau+2}}
 \subseteq\Big(\bigcup\nolimits_{\ell=\tau+2}^{\tau+4}
 E^{j_\ell}_{j_\ell}\Big)\cup(E^{i_{1}}\cup E^{i_{2}})\cup
 E^{j_{\tau+2}}_{j'_{\tau'-1}}\cup A_{\tau+2}$, where
 $A_{\tau+2}\subseteq\Sigma_q^{n-1}$ and $|A_{\tau+2}|=O_q(1)$.
\end{enumerate}

We only prove Claim 6 because the proof of Claim 6$'$ is similar.

\begin{proof}[Proof of Claim 6]
Note that for each $i\in[1,\tau-1]$, by Remark \ref{rem-form-run},
we have
$E^{j_i}=\big(\bigcup\nolimits_{\ell=1}^{\tau-1}E^{j_i}_{j_\ell}\big)
\cup E^{j_i}_{j_{\tau-1}+1}\cup E^{j_i}_{i_1}\cup
E^{j_i}_{i_2}\cup\big(\bigcup\nolimits_{\ell=\tau'+2}^{m'}E^{j_i}_{j'_\ell}\big)$.
The key is to simplify this expression.

1) Let $i\in[1,\tau-5]$. For each $j'\in\{j_\ell:
\ell\in[1,i-3]\cup[i+3,\tau-1]\}\cup\{j_{\tau-1}+1,i_1,i_2\}
\cup\{j'_{\ell}:\ell\in[\tau'+2,m']\}$, from \eqref{form-x-xp-1}
and by Remark \ref{rem-dst-subseq}, we can find that
$d_{\text{H}}\big(x_{[n]\backslash\{j_{i}\}},
x'_{[n]\backslash\{j'\}}\big)>4$, and by Lemma
\ref{lem-sub-int-size}, we can obtain
$E^{j_i}_{j'}=B^{\text{S}}_{2}\big(x_{[n]\backslash\{j_i\}},
x'_{[n]\backslash\{j'\}}\big)=\emptyset$. Therefore, we have
$E^{j_i}=\bigcup\nolimits_{\ell=\max\{1,i-2\}}^{i+2}
E^{j_i}_{j_\ell}$. Moreover, by Lemma \ref{lem-d2-smpl}, we have
$\bigcup\nolimits_{\ell=\max\{1,i-2\}}^{i+2}
E^{j_i}_{j_\ell}\subseteq\bigcup\nolimits_{\ell=\max\{1,i-2\}}^{i+2}
E^{j_\ell}_{j_\ell}$. Thus, we can obtain
$E^{j_i}\subseteq\bigcup\nolimits_{\ell=\max\{1,i-2\}}^{i+2}
E^{j_\ell}_{j_\ell}$.

2) By 1) of Remark \ref{rem-dst-subseq} and by Table 1, we can
obtain $d_{\text{H}}\big(x_{[n]\backslash\{j_{\tau-4}\}},
x'_{[n]\backslash\{j'\}}\big)>4$ for any $j'\in\{j_{\ell}:
\ell<\tau-6\}\cup\{j_{\tau-1},j_{\tau-1}+1,i_1\}\cup\{j'_{\ell}:
\ell\geq\tau'+2\}$, so by Lemma \ref{lem-sub-int-size},
$E^{j_{\tau-4}}_{j'}
=B^{\text{S}}_{2}\big(x_{[n]\backslash\{j_{\tau-4}\}},
x'_{[n]\backslash\{j'\}}\big)=\emptyset$. Therefore, we have
$E^{j_{\tau-4}}=\big(\bigcup\nolimits_{\ell=\tau-6}^{\tau-2}
E^{j_{\tau-4}}_{j_\ell}\big)\cup E^{j_{\tau-4}}_{i_2}.$ By Lemma
\ref{lem-d2-smpl}, $\bigcup\nolimits_{\ell=\tau-6}^{\tau-2}
E^{j_{\tau-4}}_{j_\ell}\subseteq\bigcup\nolimits_{\ell=\tau-6}^{\tau-2}
E^{j_\ell}_{j_\ell}$. Moreover, by Claim 1, we have
$E^{j_{\tau-4}}_{i_2}=B^{\text{S}}_{2}\big(x_{[n]\backslash\{j_{\tau-4}\}},
x'_{[n]\backslash\{i_{2}\}}\big)\subseteq
B^{\text{S}}_{2}\big(x'_{[n]\backslash\{i_{2}\}}\big)=E^{i_1}$.
Thus, we can obtain
$E^{j_{\tau-4}}\subseteq\big(\bigcup\nolimits_{\ell=\tau-6}^{\tau-2}
E^{j_\ell}_{j_\ell}\big)\cup E^{i_{1}}$.

3) Consider $i=\tau-2$. By 1) of Remark \ref{rem-dst-subseq} and
from Table 1, we can find
$d_{\text{H}}\big(x_{[n]\backslash\{j_{\tau-4}\}},
x'_{[n]\backslash\{j'\}}\big)>4$ for any $j'\in\{j_{\ell}:
\ell<\tau-4\}\cup\{j'_{\ell}: \ell\geq\tau'+4\}$, so by Lemma
\ref{lem-sub-int-size}, $E^{j_{\tau-4}}_{j'}
=B^{\text{S}}_{2}\big(x_{[n]\backslash\{j_{\tau-2}\}},
x'_{[n]\backslash\{j'\}}\big)=\emptyset$. Therefore, we can obtain
$E^{j_{\tau-2}}=\big(\bigcup\nolimits_{\ell=\tau-4}^{\tau-1}
E^{j_{\tau-2}}_{j_\ell}\big)\cup E^{j_{\tau-2}}_{j_{\tau-1}+1}\cup
E^{j_{\tau-2}}_{i_1}\cup E^{j_{\tau-2}}_{i_2}\cup
E^{j_{\tau-2}}_{j'_{\tau'+2}}\cup E^{j_{\tau-2}}_{j'_{\tau'+3}}.$
We note that:
\begin{itemize} \item By Lemma \ref{lem-d2-smpl},
$\bigcup\nolimits_{\ell=\tau-4}^{\tau-1}
E^{j_{\tau-2}}_{j_\ell}\subseteq\bigcup\nolimits_{\ell=\tau-4}^{\tau-1}
E^{j_{\ell}}_{j_\ell}$. \item By Claim 1,
$E^{j_{\tau-2}}_{i_1}=B^{\text{S}}_{2}\big(x_{[n]\backslash\{j_{\tau-2}\}},
x'_{[n]\backslash\{i_{1}\}}\big)\subseteq
B^{\text{S}}_{2}\big(x'_{[n]\backslash\{i_{1}\}}\big)=E^{i_2}$.
Similarly, $E^{j_{\tau-2}}_{i_2}\subseteq
B^{\text{S}}_{2}\big(x'_{[n]\backslash\{i_{2}\}}\big)=E^{i_1}$.
\item For $j'\in\{j_{\tau-1}+1, j'_{\tau'+2}, j'_{\tau'+3}\}$,
from Table 1, we have
$d_{\text{H}}\big(x_{[n]\backslash\{j_{\tau-2}\}},
x'_{[n]\backslash\{j'\}}\big)>2$, so by Lemma
\ref{lem-sub-int-size}, $\big|E^{j_{\tau-2}}_{j'}\big|=O_q(1)$.
\end{itemize}
Thus, we can obtain
$E^{j_{\tau-2}}\subseteq\left(\bigcup\nolimits_{\ell=\tau-4}^{\tau-1}
E^{j_\ell}_{j_\ell}\right)\cup(E^{i_{1}}\cup E^{i_{2}})\cup
A_{\tau-2}$, where $A_{\tau-2}=E^{j_{\tau-2}}_{j_{\tau-1}+1}\cup
E^{j_{\tau-2}}_{j'_{\tau'+2}}\cup E^{j_{\tau-2}}_{j'_{\tau'+3}}$
and $|A_{\tau-2}|=O_q(1)$. For $i=\tau-3$, similarly by 1) of
Remark \ref{rem-dst-subseq} and Table 1, we can obtain
$E^{j_{\tau-3}}=\big(\bigcup\nolimits_{\ell=\tau-5}^{\tau-1}
E^{j_{\tau-3}}_{j_\ell}\big)\cup E^{j_{\tau-3}}_{i_1}\cup
E^{j_{\tau-3}}_{i_2}\cup E^{j_{\tau-3}}_{j'_{\tau'+2}}.$ Further,
by Lemmas \ref{lem-sub-int-size} and \ref{lem-d2-smpl}, we can
obtain
$E^{j_{\tau-3}}\subseteq\left(\bigcup\nolimits_{\ell=\tau-5}^{\tau-1}
E^{j_\ell}_{j_\ell}\right)\cup(E^{i_{1}}\cup E^{i_{2}})\cup
A_{\tau-3}$, where $A_{\tau-3}=E^{j_{\tau-3}}_{j'_{\tau'+2}}$ and
$|A_{\tau-3}|=O_q(1)$.

4) The proof is similar to 2). From 1) of Remark
\ref{rem-dst-subseq} and Table 1, we can obtain
$d_{\text{H}}\big(x_{[n]\backslash\{j_{i}\}},
x'_{[n]\backslash\{j_\ell\}}\big)>4$ for any
$j'\in\{j_\ell:\ell<\tau-3\}\cup\{j'_{\ell}: \ell>\tau'+4\}$, so
by Lemma \ref{lem-sub-int-size}, we have
$E^{j_{\tau-1}}=\Big(\bigcup\nolimits_{\ell=\tau-3}^{\tau-1}
E^{j_{\tau-1}}_{j_\ell}\Big)\cup E^{j_{\tau-1}}_{j_{\tau-1}+1}\cup
E^{j_{\tau-1}}_{i_1}\cup E^{j_{\tau-1}}_{i_2}\cup
E^{j_{\tau-1}}_{j'_{\tau'+2}}\cup
E^{j_{\tau-1}}_{j'_{\tau'+3}}\cup E^{j_{\tau-1}}_{j'_{\tau'+4}}$.
From Table 1, $A_{\tau-1}\triangleq
E^{j_{\tau-1}}_{j_{\tau-1}+1}\cup
E^{j_{\tau-1}}_{j'_{\tau'+3}}\cup E^{j_{\tau-1}}_{j'_{\tau'+4}}$
has size $O_q(1)$. Moreover, $E^{j_{\tau-1}}_{i_1}\cup
E^{j_{\tau-1}}_{i_2}\subseteq E^{i_1}\cup E^{i_2}$ and
$\bigcup\nolimits_{\ell=\tau-3}^{\tau-1}
E^{j_{\tau-1}}_{j_\ell}\subseteq\bigcup\nolimits_{\ell=\tau-3}^{\tau-1}
E^{j_\ell}_{j_\ell}$. Thus, $E^{j_{\tau-1}}
\subseteq\Big(\bigcup\nolimits_{\ell=\tau-3}^{\tau-1}
E^{j_\ell}_{j_\ell}\Big)\cup(E^{i_{1}}\cup E^{i_{2}})\cup
E^{j_{\tau-1}}_{j'_{\tau'+2}}\cup A_{\tau-1}$ and
$|A_{\tau-1}|=O_q(1)$.
\end{proof}

By Claim 6, we can obtain $\bigcup_{i=1}^{\tau-1}E^{j_i}\subseteq
\Big(\bigcup\nolimits_{\ell=1}^{\tau-1}
E^{j_\ell}_{j_\ell}\Big)\cup(E^{i_{1}}\cup E^{i_{2}})\cup
E^{j_{\tau-1}}_{j'_{\tau'+2}}\cup A_{L}$ where
$A_{L}=A_{\tau-3}\cup A_{\tau-2}\cup A_{\tau-1}$ and
$|A_{L}|=O_q(1)$. By Claim 6$'$, we can obtain
$\bigcup_{i=\tau+2}^mE^{j_i}\subseteq
\Big(\bigcup\nolimits_{\ell=\tau+2}^{m}
E^{j_\ell}_{j_\ell}\Big)\cup(E^{i_{1}}\cup E^{i_{2}})\cup
E^{j_{\tau+2}}_{j'_{\tau'-1}}\cup A_{R}$, where
$A_{R}\subseteq\Sigma_q^{n-1}$ and $|A_{R}|=O_q(1)$. Thus, by
Remark \ref{rem-form-run}, we have
\begin{align*}
B^{\text{DS}}_{1,2}(\bm x,\bm
x')&=\bigg(\bigcup_{i=1}^{\tau-1}E^{j_i}\bigg) \cup E^{i_1}\cup
E^{i_2}\cup\bigg(\bigcup_{i=\tau+2}^{m}E^{j_i}\bigg)\\
&=\Big(\bigcup\nolimits_{\ell=1}^{\tau-1}
E^{j_\ell}_{j_\ell}\Big)\cup\Big(\bigcup\nolimits_{\ell=\tau+2}^{m}
E^{j_\ell}_{j_\ell}\Big)\cup(E^{i_{1}}\cup E^{i_{2}})\cup
E^{j_{\tau-1}}_{j'_{\tau'+2}}\cup
E^{j_{\tau+2}}_{j'_{\tau'-1}}\cup A\end{align*} where $A=A_L\cup
A_R$ and $|A|=O_q(1)$. This proves Claim 2.

\subsection{Proof of Claim 3}

First, for each
$(\ell,\ell')\in\big([1,\tau-1]\times[\tau+2,m]\big)
\big\backslash\{(\tau-2,\tau+2),(\tau-1,\tau+2),(\tau-1,\tau+3)\}$,
by \eqref{eq-seqx-form} and Observation 2, we can find
$d_{\text{H}}\big(x_{[n]\backslash\{j_{\ell}\}},
x_{[n]\backslash\{j_{\ell'}\}}\big)\geq 5$, so we have
$$\big|E^{j_{\ell}}_{j_{\ell}}\cap E^{j_{\ell'}}_{j_{\ell'}}\big|
=\big|B^{\text{S}}_{2}\big(x_{[n]\backslash\{j_{\ell}\}},
x'_{[n]\backslash\{j_{\ell}\}}\big)\cap
B^{\text{S}}_{2}\big(x_{[n]\backslash\{j_{\ell'}\}},
x_{[n]\backslash\{j_{\ell'}\}}\big)\big|\leq
\big|B^{\text{S}}_{2}\big(x_{[n]\backslash\{j_{\ell}\}},
x_{[n]\backslash\{j_{\ell'}\}}\big)\big|=0,$$ where the last
equality comes from Lemma \ref{lem-sub-int-size}. Therefore, we
can obtain
$$\left(\bigcup_{\ell\in[1,\tau-1]}E^{j_\ell}_{j_\ell}\right)
\cap\left(\bigcup_{\ell\in[\tau+2,m]}E^{j_\ell}_{j_\ell}\right)
=\big(E^{j_{\tau-2}}_{j_{\tau-2}}\cap
E^{j_{\tau+2}}_{j_{\tau+2}}\big)\cup
\big(E^{j_{\tau-1}}_{j_{\tau-1}}\cap
E^{j_{\tau+2}}_{j_{\tau+2}}\big)\cup
\big(E^{j_{\tau-1}}_{j_{\tau-1}}\cap
E^{j_{\tau+3}}_{j_{\tau+3}}\big).$$ Moreover, for each
$(\ell,\ell')\in\{(\tau-2,\tau+2)$, $(\tau-1,\tau+2)$,
$(\tau-1,\tau+3)\}$, by \eqref{eq-seqx-form} and Observation 2, we
can find $d_{\text{H}}\big(x_{[n]\backslash\{j_{\ell}\}},
x_{[n]\backslash\{j_{\ell'}\}}\big)\in\{3,4\}$, so we have
$\big|E^{j_{\ell}}_{j_{\ell}}\cap
E^{j_{\ell'}}_{j_{\ell'}}\big|=\big|B^{\text{S}}_{2}\big(x_{[n]\backslash\{j_{\ell}\}},
x'_{[n]\backslash\{j_{\ell}\}}\big)\cap
B^{\text{S}}_{2}\big(x_{[n]\backslash\{j_{\ell'}\}},
x_{[n]\backslash\{j_{\ell'}\}}\big)\big|\leq
\big|B^{\text{S}}_{2}\big(x_{[n]\backslash\{j_{\ell}\}},
x_{[n]\backslash\{j_{\ell'}\}}\big)\big|=O_q(1)$, where the last
equality comes from Lemma \ref{lem-sub-int-size} and
\eqref{xxp-case2}. Thus, we have
$$\left|\left(\bigcup_{\ell\in[1,\tau-1]}E^{j_\ell}_{j_\ell})
\cap(\bigcup_{\ell\in[\tau+2,m]}E^{j_\ell}_{j_\ell}\right)\right|
=\Big|\big(E^{j_{\tau-2}}_{j_{\tau-2}}\cap
E^{j_{\tau+2}}_{j_{\tau+2}}\big)\cup
\big(E^{j_{\tau-1}}_{j_{\tau-1}}\cap
E^{j_{\tau+2}}_{j_{\tau+2}}\big)\cup
\big(E^{j_{\tau-1}}_{j_{\tau-1}}\cap
E^{j_{\tau+3}}_{j_{\tau+3}}\big)\Big|=O_q(1),$$ which completes
the proof.

\subsection{Proof of Claim 4}

Before prove Claim 4, we first prove two claims.

\textbf{Claim 7}: Let $\{j_i: i\in[m]\}$, $\{j'_{i'}: i'\in[m]\}$,
$\tau$ and $\tau'$ be defined as in Remark \ref{rem-form-run}. We
have
\begin{itemize}
 \item[1)] $E^{j_\ell}_{j_{\ell}}\cap\big(\bigcup_{\ell'=\ell+3}^{\tau-1}
 E^{j_{\ell'}}_{j_{\ell'}}\big)=\emptyset$ for each
 $\ell\in[\tau-4]$.
 \item[2)] $\big|E^{j_\ell}_{j_{\ell}}\backslash\big(E^{j_{\ell+1}}_{j_{\ell+1}}
 \cup E^{j_{\ell+2}}_{j_{\ell+2}}\big) \big|=2(q-1)n+q^2-8q+4$ for
 each $\ell\in[\tau-3]$.
 \item[3)] $\big|E^{j_\ell}_{j_{\ell}}\backslash
 E^{j_{\ell+1}}_{j_{\ell+1}} \big|=2(q-1)n+q^2-8q+6$ for each
 $\ell\in[\tau-2]$.
\end{itemize}
\begin{proof}[Proof of Claim 7]
By Remark \ref{rem-both-i1-i2}, we can write
\begin{align}\label{eq-seqxxp-form}\bm
x&=x_{[1,i_1-1]}~a~b~x_{[i_2+1,n]}\nonumber\\
\bm x'&=x_{[1,i_1-1]}~b~a~x_{[i_2+1,n]}
\end{align}
where $a=x_{i_1}$ and $b=x_{i_2}$. For each $\ell\in[1,\tau-1]$,
we can obtain
\begin{align*}x_{[n]\backslash\{j_{\ell}\}}
&=[x_{[1,j_{\ell}-1]}\!~x_{[j_{\ell}+1,i_1-1]}]\!~a\!~b\!~
[x_{[i_2+1,n]}]\\
x'_{[n]\backslash\{j_{\ell}\}}&=[x_{[1,j_{\ell}-1]}\!~x_{[j_{\ell}+1,i_1-1]}]\!~b\!~a\!~
[x_{[i_2+1,n]}]
\end{align*} so by Lemma \ref{lem-sub-int-size} and
\eqref{xxp-case2},
$\big|E^{j_\ell}_{j_{\ell}}\big|=\xi^{q,n-1}_{2,2}=2(q-1)n+q^2-6q+6$.
Moreover, for each $\bm v_\ell\in E^{j_{\ell}}_{j_{\ell}}$, by 2)
of Lemma \ref{lem-Cap-SE}, $\bm v_\ell$ is of the form $\bm
v_\ell=[x_{[1,j_{\ell}-1]}\!~x_{[j_{\ell}+1,i_1-1]}]\!~v\!~v'\!~
[x_{[i_2+1,n]}]$ and satisfies the following conditions (i) or
(ii):
\begin{enumerate} \item[(i)]
$[x_{[1,j_{\ell}-1]}\!~x_{[j_{\ell}+1,i_1-1]}]\!~[x_{[i_2+1,n]}]=
x_{[1,j_{\ell}-1]}\!~x_{[j_{\ell}+1,i_1-1]}\!~x_{[i_2+1,n]}$ and
$vv'\in\Sigma_q^2$; \item[(ii)]
$[x_{[1,j_{\ell}-1]}\!~x_{[j_{\ell}+1,i_1-1]}]\!~[x_{[i_2+1,n]}]$
is obtained from
$x_{[1,j_{\ell}-1]}\!~x_{[j_{\ell}+1,i_1-1]}\!~x_{[i_2+1,n]}$ by
exactly one substitution and $vv'\in\{aa,bb\}$.\end{enumerate}

1) For $\ell,\ell'\in[1,m]$ such that $\ell'\geq\ell+3$, by the
definition of $\{j_i: i\in[m]\}~($see Remark \ref{rem-form-run}$)$
and by Remark \ref{rem-spc-Obv2}, we can find that the Hamming
distance between $x_{[1,j_{\ell}-1]}\!~x_{[j_{\ell}+1,i_1-1]}$ and
$x_{[1,j_{\ell'}-1]}\!~x_{[j_{\ell'}+1,i_1-1]}$ is $\ell'-\ell\geq
3$, so we always have $\bm v_\ell\neq \bm v_{\ell'}$, which
deduces that
$E^{j_\ell}_{j_{\ell}}\cap\big(\bigcup_{\ell'=\ell+3}^{m}
E^{j_{\ell'}}_{j_{\ell'}}\big)=\emptyset$.

2) We consider
$E^{j_\ell}_{j_{\ell}}\backslash\big(E^{j_{\ell+1}}_{j_{\ell+1}}
\cup E^{j_{\ell+2}}_{j_{\ell+2}}\big)$. According to
\eqref{form-x-xp-1} and 2) Of Lemma \ref{lem-Cap-SE}, we can write
each $\bm v_\ell\in E^{j_{\ell}}_{j_{\ell}}$, $\bm v_{\ell+1}\in
E^{j_{\ell+1}}_{j_{\ell+1}}$ and $\bm v_{\ell+2}\in
E^{j_{\ell+2}}_{j_{\ell+2}}$ as the following form that satisfy
the conditions similar to the above (i) or (ii).
\begin{align}
\bm v_{\ell} &=[x_{[1,j_{\ell}-1]}~~\bar{x}_{\ell+1}~
\bar{x}_{\ell+1}^{k_{\ell+1}-1}~~\bar{x}_{\ell+2}~
\!~x_{[j_{\ell+2}+1,i_1-1]}]~~v~\!~~v'~\!
~~[x_{[i_2+1,n]}] \label{eq2-jljl}\\
\bm v_{\ell+1}&=[x_{[1,j_{\ell}-1]}~~\!~~\bar{x}_{\ell}~
\!~~\bar{x}_{\ell+1}^{k_{\ell+1}-1}\!~~\bar{x}_{\ell+2}~
\!~x_{[j_{\ell+2}+1,i_1-1]}]~\!~~u~~u'~~\!~
[x_{[i_2+1,n]}]\label{eq4-jljl}\\
\bm v_{\ell+2}&=[x_{[1,j_{\ell}-1]}~~\!~~\bar{x}_{\ell}~
\!~~\bar{x}_{\ell+1}^{k_{\ell+1}-1}\!~~\bar{x}_{\ell+1}~
\!~x_{[j_{\ell+2}+1,i_1-1]}]~\!~\!~w~\!~w'~~\!~ [x_{[i_2+1,n]}].
\label{eq6-jljl}\end{align} From
\eqref{eq2-jljl}$-$\eqref{eq6-jljl}, we can see that $\bm
v_\ell\in E^{j_{\ell+1}}_{j_{\ell+1}}\cup
E^{j_{\ell+2}}_{j_{\ell+2}}$ if and only if it satisfies the
conditions: (a) $\bm v_\ell=x_{[1,j_{\ell}-1]}\!~a'\!~
\bar{x}_{\ell+1}^{k_{\ell+1}-1}\!~\bar{x}_{\ell+2}\!~
x_{[j_{\ell+2}+1,i_1-1]}\!~v\!~v'\!~x_{[i_2+1,n]}$ such that
$vv'\in\{aa,bb\}$ and $a'\in\Sigma_q$; or (b) $\bm
v_\ell=x_{[1,j_{\ell}-1]}\!~\bar{x}_{\ell+1}\!~
\bar{x}_{\ell+1}^{k_{\ell+1}-1}\!~\bar{x}_{\ell+1}\!~
x_{[j_{\ell+2}+1,i_1-1]}\!~v\!~v'\!~x_{[i_2+1,n]}$ and
$vv'\in\{aa,bb\}$. Clearly, the number of $\bm v_\ell$ satisfying
condition (a) or (b) is $2(q+1)$, so
$\big|E^{j_\ell}_{j_{\ell}}\backslash\big(E^{j_{\ell+1}}_{j_{\ell+1}}
\cup E^{j_{\ell+2}}_{j_{\ell+2}}\big)\big|
=\big|E^{j_\ell}_{j_{\ell}}\big|-2(q+1)=2(q-1)n+q^2-6q+6-2(q+1)
=2(q-1)n+q^2-8q+4$.

3) We consider $E^{j_\ell}_{j_{\ell}}\backslash
E^{j_{\ell+1}}_{j_{\ell+1}}$. From \eqref{eq2-jljl} and
\eqref{eq4-jljl}, we can see that $\bm v_\ell\in
E^{j_{\ell+1}}_{j_{\ell+1}}$ if and only if it satisfies the above
condition (a) and the number of such $\bm v_\ell$ is $2q$, so
$\big|E^{j_\ell}_{j_{\ell}}\backslash
E^{j_{\ell+1}}_{j_{\ell+1}}\big|
=\big|E^{j_\ell}_{j_{\ell}}\big|-2q=2(q-1)n+q^2-8q+6$.
\end{proof}

For $\ell\in[\tau+2,m]$, we have the following Claim 7$'$, which
can be proved by the same discussions as that of Claim 7.

\textbf{Claim 7$'$}: Let $\{j_i: i\in[m]\}$, $\{j'_{i'}:
i'\in[m]\}$, $\tau$ and $\tau'$ be defined as in Remark
\ref{rem-form-run}. Then we have
\begin{itemize}
 \item[1)] $E^{j_\ell}_{j_{\ell}}\cap\big(\bigcup_{\ell'=\tau+2}^{\ell-3}
 E^{j_{\ell'}}_{j_{\ell'}}\big)=\emptyset$ for each $\ell\in[\tau+5,m]$.
 \item[2)] $\big|E^{j_\ell}_{j_{\ell}}\backslash\big(E^{j_{\ell-1}}_{j_{\ell-1}}
 \cup E^{j_{\ell-2}}_{j_{\ell-2}}\big) \big|=2(q-1)n+q^2-8q+4$ for
 each $\ell\in[\tau+4,m]$.
 \item[3)] $\big|E^{j_\ell}_{j_{\ell}}\backslash
 E^{j_{\ell-1}}_{j_{\ell-1}}\big|=2(q-1)n+q^2-8q+6$ for each
 $\ell\in[\tau+3,m]$.
\end{itemize}

\textbf{Claim 8}: Let $\{j_i: i\in[m]\}$, $\{j'_{i'}: i'\in[m]\}$,
$\tau$ and $\tau'$ be defined as in Remark \ref{rem-form-run}. If
$\sigma_1>0$, we have
\begin{enumerate}
 \item[1.1)] $\big|E^{j_{\tau-1}}_{j_{\tau-1}}
 \cap(E^{i_{1}}\cup E^{i_{2}})\big|=(q-1)n+O_q(1)$.
 \item[1.2)] For $\ell\in\{\tau-2,\tau-3\}$,
 $\big|E^{j_\ell}_{j_\ell}\cap(E^{i_{1}}\cup
 E^{i_{2}})\big|=O_q(1)$.
 \item[1.3)] For $\ell\in[1,\tau-4]$,
 $\big|E^{j_\ell}_{j_\ell}\cap(E^{i_{1}}\cup E^{i_{2}})\big|=0$.
\end{enumerate}
If $\sigma_1=0$, we have
\begin{enumerate}
 \item[2.1)] $E^{j_{\tau-1}}_{j_{\tau-1}}
 \subseteq(E^{i_{1}}\cup E^{i_{2}})$.
 \item[2.2)] For $\ell\in\{\tau-2,\tau-3,\tau-4\}$,
 $\big|E^{j_\ell}_{j_\ell}\cap(E^{i_{1}}\cup
 E^{i_{2}})\big|=O_q(1)$.
 \item[2.3)] For $\ell\in[1,\tau-5]$,
 $\big|E^{j_\ell}_{j_\ell}\cap(E^{i_{1}}\cup E^{i_{2}})\big|=0$.
\end{enumerate}
\begin{proof}[Proof of Claim 8]
First, suppose $\sigma_1>0$. We prove 1.1)$-$1.3).

1.1) Since $\sigma_1>0$, from \eqref{form-x-xp-1}, we can obtain
the following \eqref{eq1-jm-i1i2-sgn0}, \eqref{eq2-jm-i1i2-sgn0},
\eqref{eq6-jm-i1i2-sgn0} and \eqref{eq7-jm-i1i2-sgn0}. Moreover,
for any $\bm v\in
E^{j_{\tau-1}}_{j_{\tau-1}}=B^{\text{S}}_{2}\big(x_{[n]\backslash\{j_{\tau-1}\}},
x'_{[n]\backslash\{j_{\tau-1}\}}\big)$, by 2) of Lemma
\ref{lem-Cap-SE}, we can obtain \eqref{eq4-jm-i1i2-sgn0} such
that: (i) $[\cdots~\bar{x}_{\tau-1}^{k_{\tau-1}-1}
\!~a\!~a^{\sigma_1-1}][a^{\sigma'_2}\!~
~\tilde{x}_{\tau'+2}^{k'_{\tau'+2}}~\cdots]=\cdots~\bar{x}_{\tau-1}^{k_{\tau-1}-1}
\!~a\!~a^{\sigma_1-1}\!~a^{\sigma'_2}\!~\tilde{x}_{\tau'+2}^{k'_{\tau'+2}}
~\cdots$ and $vv'\in\Sigma_q^2$; or (ii)
$[\cdots~\bar{x}_{\tau-1}^{k_{\tau-1}-1}
\!~a\!~a^{\sigma_1-1}][a^{\sigma'_2}\!~
~\tilde{x}_{\tau'+2}^{k'_{\tau'+2}}~\cdots]$ is obtained from
$\cdots~\bar{x}_{\tau-1}^{k_{\tau-1}-1}
\!~a\!~a^{\sigma_1-1}\!~a^{\sigma'_2}\!~\tilde{x}_{\tau'+2}^{k'_{\tau'+2}}
~\cdots$ by exactly one substitution and $vv'\in\{aa, bb\}$.
\begin{align}
x_{[n]\backslash\{j_{\tau-1}\}}
&=\!~\cdots~\bar{x}_{\tau-1}^{k_{\tau-1}-1}~\!~\!~~~a~\!~~\!~~a^{\sigma_1-1}~~
~a~~b~~~a^{\sigma'_2}\!~
~\tilde{x}_{\tau'+2}^{k'_{\tau'+2}}~\cdots\label{eq1-jm-i1i2-sgn0}\\
x'_{[n]\backslash\{j_{\tau-1}\}}
&=\!~\cdots~\bar{x}_{\tau-1}^{k_{\tau-1}-1}~\!~\!~~~a~\!~~\!~~a^{\sigma_1-1}~~~b~
~\!~a\!~~~a^{\sigma'_2}\!~
~\tilde{x}_{\tau'+2}^{k'_{\tau'+2}}~\cdots\label{eq2-jm-i1i2-sgn0}\\
\bm v~~ &=[\cdots~\bar{x}_{\tau-1}^{k_{\tau-1}-1}
\!~~\!~~~a\!~~~\!~~a^{\sigma_1-1}]~\!~~v~\!~\!~v'\!~~[a^{\sigma'_2}\!~
~\tilde{x}_{\tau'+2}^{k'_{\tau'+2}}~\cdots]\label{eq4-jm-i1i2-sgn0}\\
x_{[n]\backslash\{i_1\}}
&=\!~\cdots~\bar{x}_{\tau-1}^{k_{\tau-1}-1}~\!~\bar{x}_{\tau-1}~~a^{\sigma_1-1}
~~\!~~a~\!\!~~b~\!~~a^{\sigma'_2}\!~
~\tilde{x}_{\tau'+2}^{k'_{\tau'+2}}~\cdots\label{eq6-jm-i1i2-sgn0}\\
x_{[n]\backslash\{i_2\}}
&=\!~\cdots~\bar{x}_{\tau-1}^{k_{\tau-1}-1}\!~~\bar{x}_{\tau-1}~~a^{\sigma_1-1}
~~~\!~a~~a~\!~~a^{\sigma'_2}\!~
~\tilde{x}_{\tau'+2}^{k'_{\tau'+2}}~\cdots\label{eq7-jm-i1i2-sgn0}\\
\nonumber\end{align} Recall that by Claim 1,
$E^{i_1}=B^{\text{S}}_{2}\big(x_{[n]\backslash\{i_1\}}\big)$ and
$E^{i_2}=B^{\text{S}}_{2}\big(x_{[n]\backslash\{i_2\}}\big)$. We
can see that: a) if $\bm v$ satisfies (i), then $\bm v\in
E^{i_1}\cup E^{i_2}$ if and only if $v=a$ or $v'\in\{a,b\}$, so
the number of such $\bm v$ is $3q-2$; b) if $\bm v$ satisfies (ii)
and $vv'=aa$, then $d_{\text{H}}\big(\bm v,
x_{[n]\backslash\{i_{2}\}}\big)=1$, so $\bm v\in
E^{i_2}=B^{\text{S}}_{2}\big(x_{[n]\backslash\{i_2\}}\big)$ and
the number of such $\bm v$ is $(n-3)(q-1)$; c) if $\bm v$
satisfies (iii) and $vv'=bb$, then $\bm v\in E^{i_1}\cup E^{i_2}$
if and only if $\bm
v=\cdots~\bar{x}_{\tau-1}^{k_{\tau-1}-1}\!~a'\!~a^{\sigma_1-1}
\!~b~b~a^{\sigma'_2}\!~\tilde{x}_{\tau'+2}^{k'_{\tau'+2}}~\cdots$
for some $a'\in\Sigma_q\backslash\{a\}$, so the number of such
$\bm v$ is $q-1$. Hence,
$\big|E^{j_{\tau-1}}_{j_{\tau-1}}\cap(E^{i_{1}}\cup
E^{i_{2}})\big|=(3q-2)+(n-3)(q-1)+(q-1)=(q-1)n+O_q(1)$.

1.2) Since $\sigma_1>0$, from \eqref{form-x-xp-1} and by
Observation 2, we can find that:
$d_{\text{H}}(x_{[n]\backslash\{j_{\tau-2}\}},
x'_{[n]\backslash\{i_1\}})=3$ and
$d_{\text{H}}(x_{[n]\backslash\{j_{\tau-2}\}},
x'_{[n]\backslash\{i_2\}})=3$, so
$\big|E^{j_{\tau-2}}_{j_{\tau-2}}\cap(E^{i_{1}}\cup
E^{i_{2}})\big|\leq \big|E^{j_{\tau-2}}_{j_{\tau-2}}\cap
E^{i_{1}}\big|+\big|E^{j_{\tau-2}}_{j_{\tau-2}}\cap
E^{i_{2}}\big|=O_q(1)$, where the last equality comes from Lemma
\ref{lem-sub-int-size} and \eqref{xxp-case2}. Similarly, we can
find $d_{\text{H}}(x_{[n]\backslash\{j_{\tau-3}\}},
x'_{[n]\backslash\{i_1\}})=4$ and
$d_{\text{H}}(x_{[n]\backslash\{j_{\tau-3}\}},
x'_{[n]\backslash\{i_2\}})=4$, so
$\big|E^{j_{\tau-3}}_{j_{\tau-3}}\cap(E^{i_{1}}\cup
E^{i_{2}})\big|\leq \big|E^{j_{\tau-3}}_{j_{\tau-3}}\cap
E^{i_{1}}\big|+\big|E^{j_{\tau-3}}_{j_{\tau-3}}\cap
E^{i_{2}}\big|=O_q(1)$.

1.3) For $\ell\in[1,\tau-4]$, we find
$d_{\text{H}}(x_{[n]\backslash\{j_\ell\}},
x'_{[n]\backslash\{i_1\}})\geq 5$ and
$d_{\text{H}}(x_{[n]\backslash\{j_\ell\}},
x'_{[n]\backslash\{i_2\}})\geq 5$, so
$\big|E^{j_\ell}_{j_\ell}\cap(E^{i_{1}}\cup E^{i_{2}})\big|\leq
\big|E^{j_{\ell}}_{j_{\ell}}\cap
E^{i_{1}}\big|+\big|E^{j_{\ell}}_{j_{\ell}}\cap E^{i_{2}}\big|=0$.

Next, suppose $\sigma_1=0$. We prove 2.1)$-$2.3). Since
$\sigma_1=0$, then \eqref{eq-seqx-form} becomes
\begin{align}\label{eq-seqx-form-sg0}\bm
x=\bar{x}_1^{k_1}~\cdots~\bar{x}_{\tau-1}^{k_{\tau-1}}~a~b~b^{\sigma_2}
~\bar{x}_{\tau+2}^{k_{\tau+2}}~\cdots~\bar{x}_{m}^{k_{m}}\end{align}
and \eqref{form-x-xp-1} becomes
\begin{align}\label{form-x-xp-1-sg0}
\bm x &=\bar{x}_1^{k_1}~\cdots~\bar{x}_{\tau-1}^{k_{\tau-1}}
~a~b~a^{\sigma'_2}~\tilde{x}_{\tau'+2}^{k'_{\tau'+2}}~\cdots
~\tilde{x}_{m'}^{k'_{m'}}\nonumber\\
\bm x'&=\bar{x}_1^{k_1}~\cdots~\bar{x}_{\tau-1}^{k_{\tau-1}}
~b~a~a^{\sigma'_2}~\tilde{x}_{\tau'+2}^{k'_{\tau'+2}}~\cdots
~\tilde{x}_{m'}^{k'_{m'}}
\end{align}
Note that $\bar{x}_{\tau-1}\neq a$ because
$\bar{x}_{\tau-1}^{k_{\tau-1}}$ and $a$ are two runs of $\bm x$.
However, it is not necessary that $\bar{x}_{\tau-1}\neq b$.

2.1) From \eqref{form-x-xp-1-sg0}, we can obtain the following
\eqref{eq1-jm-i1i2-sg0}, \eqref{eq2-jm-i1i2-sg0},
\eqref{eq6-jm-i1i2-sg0} and \eqref{eq7-jm-i1i2-sg0}. For any $\bm
v\in
E^{j_{\tau-1}}_{j_{\tau-1}}=B^{\text{S}}_{2}\big(x_{[n]\backslash\{j_{\tau-1}\}},
x'_{[n]\backslash\{j_{\tau-1}\}}\big)$, by 2) of Lemma
\ref{lem-Cap-SE}, $\bm v$ can be written as the form
\eqref{eq4-jm-i1i2-sg0} such that: (i)
$[\cdots~\bar{x}_{\tau-1}^{k_{\tau-1}-1}][a^{\sigma'_2}\!~
~\tilde{x}_{\tau'+2}^{k'_{\tau'+2}}~\cdots]=
\cdots~\bar{x}_{\tau-1}^{k_{\tau-1}-1}~a^{\sigma'_2}\!~\tilde{x}_{\tau'+2}^{k'_{\tau'+2}}
~\cdots$ and $vv'\in\Sigma_q^2$; or (ii)
$[\cdots~\bar{x}_{\tau-1}^{k_{\tau-1}-1}][a^{\sigma'_2}\!~
~\tilde{x}_{\tau'+2}^{k'_{\tau'+2}}~\cdots]$ is obtained from
$\cdots~\bar{x}_{\tau-1}^{k_{\tau-1}-1}~a^{\sigma'_2}\!~\tilde{x}_{\tau'+2}^{k'_{\tau'+2}}
~\cdots$ by exactly one substitution and $vv'\in\{aa,bb\}$.
\begin{align}
x_{[n]\backslash\{j_{\tau-1}\}}
&=\!~\cdots~\bar{x}_{\tau-1}^{k_{\tau-1}-1}\!~~~\!~~~a~\!~~~~\!~b~~~
a^{\sigma'_2}\!~~
\tilde{x}_{\tau'+2}^{k'_{\tau'+2}}~\cdots\label{eq1-jm-i1i2-sg0}\\
x'_{[n]\backslash\{j_{\tau-1}\}}
&=\!~\cdots~\bar{x}_{\tau-1}^{k_{\tau-1}-1}\!~~~\!~~~b~~~
\!~~\!~a~~~a^{\sigma'_2}~\!~
\tilde{x}_{\tau'+2}^{k'_{\tau'+2}}~\cdots\label{eq2-jm-i1i2-sg0}\\
\bm v~~
&=[\cdots~\bar{x}_{\tau-1}^{k_{\tau-1}-1}]\!~~~~~v~\!~~~~v'\!~~\!~
[a^{\sigma'_2}\!~~
\tilde{x}_{\tau'+2}^{k'_{\tau'+2}}~\cdots]\label{eq4-jm-i1i2-sg0}\\
x_{[n]\backslash\{i_1\}}
&=\!~\cdots~\bar{x}_{\tau-1}^{k_{\tau-1}-1}\!~~~\bar{x}_{\tau-1}~~~b~~~
a^{\sigma'_2}\!~~
\tilde{x}_{\tau'+2}^{k'_{\tau'+2}}~\cdots\label{eq6-jm-i1i2-sg0}\\
x_{[n]\backslash\{i_2\}}
&=\!~\cdots~\bar{x}_{\tau-1}^{k_{\tau-1}-1}\!~~~\bar{x}_{\tau-1}~~~a~~~
a^{\sigma'_2}\!~~
\tilde{x}_{\tau'+2}^{k'_{\tau'+2}}\!~\cdots\label{eq7-jm-i1i2-sg0}\\
\nonumber
\end{align}
It is easy to see that: a) if $\bm v$ satisfies (i) or $\bm v$
satisfies (ii) and $vv'=aa$, then $d_{\text{H}}\big(\bm v,
x_{[n]\backslash\{i_{2}\}}\big)\leq 2$, so $\bm v\in
E^{i_2}=B^{\text{S}}_{2}\big(x_{[n]\backslash\{i_2\}}\big)$; b) if
$\bm v$ satisfies (ii) and $vv'=bb$, then $d_{\text{H}}\big(\bm v,
x_{[n]\backslash\{i_{1}\}}\big)\leq 2$, so $\bm v\in
E^{i_1}=B^{\text{S}}_{2}\big(x_{[n]\backslash\{i_1\}}\big)$.
Hence, we have $E^{j_{\tau-1}}_{j_{\tau-1}}\subseteq
(E^{i_{1}}\cup E^{i_{2}})$.

2.2) From \eqref{eq-seqx-form-sg0} and by Observation 2, we can
find that: $d_{\text{H}}(x_{[n]\backslash\{j_{\tau-3}\}},
x_{[n]\backslash\{i_1\}})=3$ and
$d_{\text{H}}(x_{[n]\backslash\{j_{\tau-3}\}},
x_{[n]\backslash\{i_2\}})=4$, so
$\big|E^{j_{\tau-3}}_{j_{\tau-3}}\cap(E^{i_{1}}\cup
E^{i_{2}})\big|\leq \big|E^{j_{\tau-3}}_{j_{\tau-3}}\cap
E^{i_{1}}\big|+\big|E^{j_{\tau-3}}_{j_{\tau-3}}\cap
E^{i_{2}}\big|=O_q(1)$, where the last equality comes from Lemma
\ref{lem-sub-int-size} and \eqref{xxp-case2}. Similarly, we can
find $d_{\text{H}}(x_{[n]\backslash\{j_{\tau-4}\}},
x_{[n]\backslash\{i_1\}})=4$ and
$d_{\text{H}}(x_{[n]\backslash\{j_{\tau-4}\}},
x_{[n]\backslash\{i_2\}})=5$, so
$\big|E^{j_{\tau-4}}_{j_{\tau-4}}\cap(E^{i_{1}}\cup
E^{i_{2}})\big|\leq \big|E^{j_{\tau-4}}_{j_{\tau-4}}\cap
E^{i_{1}}\big|+\big|E^{j_{\tau-4}}_{j_{\tau-4}}\cap
E^{i_{2}}\big|=O_q(1)$. Finally, we prove that
$\big|E^{j_{\tau-2}}_{j_{\tau-2}}\cap(E^{i_{1}}\cup
E^{i_{2}})\big|=O_q(1)$. The proof is similar to 2.1). From
\eqref{form-x-xp-1-sg0}, we can obtain the following
\eqref{eq1-jtau-i1i2-sg0}$-$\eqref{eq7-jtau-i1i2-sg0}, where $\bm
v\in
E^{j_{\tau-2}}_{j_{\tau-2}}=B^{\text{S}}_{2}\big(x_{[n]\backslash\{j_{\tau-2}\}},
x'_{[n]\backslash\{j_{\tau-2}\}}\big)$, satisfying: (i)
$[\cdots~\bar{x}_{\tau-2}^{k_{\tau-2}-1}~\bar{x}_{\tau-1}~\bar{x}_{\tau-1}^{k_{\tau-1}-1}]
[a^{\sigma'_2}~ \tilde{x}_{\tau'+2}^{k'_{\tau'+2}}~\cdots]=
\cdots~\bar{x}_{\tau-2}^{k_{\tau-2}-1}~\bar{x}_{\tau-1}~\bar{x}_{\tau-1}^{k_{\tau-1}-1}~
a^{\sigma'_2}~ \tilde{x}_{\tau'+2}^{k'_{\tau'+2}}~\cdots$ and
$vv'\in\Sigma_q^2$; or (ii)
$[\cdots~\bar{x}_{\tau-2}^{k_{\tau-2}-1}~\bar{x}_{\tau-1}~\bar{x}_{\tau-1}^{k_{\tau-1}-1}]
[a^{\sigma'_2}~ \tilde{x}_{\tau'+2}^{k'_{\tau'+2}}~\cdots]$ is
obtained from
$\cdots~\bar{x}_{\tau-2}^{k_{\tau-2}-1}~\bar{x}_{\tau-1}~\bar{x}_{\tau-1}^{k_{\tau-1}-1}~
a^{\sigma'_2}~ \tilde{x}_{\tau'+2}^{k'_{\tau'+2}}~\cdots$ by
exactly one substitution and $vv'\in\{aa,bb\}$.
\begin{align}
x_{[n]\backslash\{j_{\tau-2}\}}
&=\!~\cdots~\bar{x}_{\tau-2}^{k_{\tau-2}-1}~\bar{x}_{\tau-1}~\bar{x}_{\tau-1}^{k_{\tau-1}-1}\!~~~\!~~~a~\!~~~~\!~b~~~
a^{\sigma'_2}\!~~
\tilde{x}_{\tau'+2}^{k'_{\tau'+2}}~\cdots\label{eq1-jtau-i1i2-sg0}\\
x'_{[n]\backslash\{j_{\tau-2}\}}
&=\!~\cdots~\bar{x}_{\tau-2}^{k_{\tau-2}-1}~\bar{x}_{\tau-1}~\bar{x}_{\tau-1}^{k_{\tau-1}-1}\!~~~\!~~~b~~~
\!~~\!~a~~~a^{\sigma'_2}~\!~
\tilde{x}_{\tau'+2}^{k'_{\tau'+2}}~\cdots\label{eq2-jtau-i1i2-sg0}\\
\bm v~~
&=[\cdots~\bar{x}_{\tau-2}^{k_{\tau-2}-1}~\bar{x}_{\tau-1}~\bar{x}_{\tau-1}^{k_{\tau-1}-1}]\!~~~~~v~\!~~~~v'~~
[a^{\sigma'_2}\!~~
\tilde{x}_{\tau'+2}^{k'_{\tau'+2}}~\cdots]\label{eq4-jtau-i1i2-sg0}\\
x_{[n]\backslash\{i_1\}}
&=\!~\cdots~\bar{x}_{\tau-2}^{k_{\tau-2}-1}~\bar{x}_{\tau-2}~\bar{x}_{\tau-1}^{k_{\tau-1}-1}\!~~~\bar{x}_{\tau-1}~~~b~~~
a^{\sigma'_2}\!~~
\tilde{x}_{\tau'+2}^{k'_{\tau'+2}}~\cdots\label{eq6-jtau-i1i2-sg0}\\
x_{[n]\backslash\{i_2\}}
&=\!~\cdots~\bar{x}_{\tau-2}^{k_{\tau-2}-1}~\bar{x}_{\tau-2}~\bar{x}_{\tau-1}^{k_{\tau-1}-1}\!~~~\bar{x}_{\tau-1}~~~a~~~
a^{\sigma'_2}\!~~
\tilde{x}_{\tau'+2}^{k'_{\tau'+2}}\!~\cdots\label{eq7-jtau-i1i2-sg0}\\
\nonumber
\end{align}
Recall that by Claim 1,
$E^{i_1}=B^{\text{S}}_{2}\big(x_{[n]\backslash\{i_1\}}\big)$ and
$E^{i_2}=B^{\text{S}}_{2}\big(x_{[n]\backslash\{i_2\}}\big)$. We
can easily find that: a) if $\bm v$ satisfies (i), then $\bm v\in
E^{i_1}\cup E^{i_2}$ if and only if $v=\bar{x}_{\tau-1}$ or
$v'\in\{a,b\}$ and the number of such $\bm v$ is $3q-2$; b) if
$\bm v$ satisfies (ii) and $vv'=aa$, then $\bm v\in E^{i_1}\cup
E^{i_2}$ if and only if $\bm
v=\cdots~\bar{x}_{\tau-2}^{k_{\tau-2}-1}~a'~\bar{x}_{\tau-1}^{k_{\tau-1}-1}~a~a~
a^{\sigma'_2}~\tilde{x}_{\tau'+2}^{k'_{\tau'+2}}~\cdots$ for some
$a'\in\Sigma_q\backslash\{\bar{x}_{\tau-1}\}$ and the number of
such $\bm v$ is $q-1$; c) if $\bm v$ satisfies (ii) and $vv'=bb$,
then $\bm v\in E^{i_1}\cup E^{i_2}$ if and only if $\bm
v=\cdots~\bar{x}_{\tau-2}^{k_{\tau-2}-1}~a'~\bar{x}_{\tau-1}^{k_{\tau-1}-1}~b~b~
a^{\sigma'_2}~ \tilde{x}_{\tau'+2}^{k'_{\tau'+2}}~\cdots$ for some
$a'\in\Sigma_q\backslash\{a\}$, so the number of such $\bm v$ is
$q-1$. Hence, $\big|E^{j_{\tau-2}}_{j_{\tau-2}}\cap(E^{i_{1}}\cup
E^{i_{2}})\big|=(3q-2)+2(q-1)=O_q(1)$.

2.3) For $\ell\in[1,\tau-5]$, we find
$d_{\text{H}}(x_{[n]\backslash\{j_\ell\}},
x_{[n]\backslash\{i_1\}})\geq 5$ and
$d_{\text{H}}(x_{[n]\backslash\{j_\ell\}},
x_{[n]\backslash\{i_2\}})\geq 5$, so
$\big|E^{j_\ell}_{j_\ell}\cap(E^{i_{1}}\cup E^{i_{2}})\big|\leq
\big|E^{j_{\ell}}_{j_{\ell}}\cap
E^{i_{1}}\big|+\big|E^{j_{\ell}}_{j_{\ell}}\cap E^{i_{2}}\big|=0$.
\end{proof}

For $\ell\in[\tau+2,m]$, we have the following Claim 8$'$, which
can be proved by the same discussions as that of Claim 8.

\textbf{Claim 8$'$}: Let $\{j_i: i\in[m]\}$, $\{j'_{i'}:
i'\in[m]\}$, $\tau$ and $\tau'$ be defined as in Remark
\ref{rem-form-run}. If $\sigma_2>0$, we have
\begin{enumerate}
 \item[1.1)] $\big|E^{j_{\tau+2}}_{j_{\tau+2}}
 \cap(E^{i_{1}}\cup E^{i_{2}})\big|=(q-1)n+O_q(1)$.
 \item[1.2)] For $\ell\in\{\tau+3,\tau+4\}$,
 $\big|E^{j_\ell}_{j_\ell}\cap(E^{i_{1}}\cup
 E^{i_{2}})\big|=O_q(1)$.
 \item[1.3)] For $\ell\in[\tau+5,m]$,
 $\big|E^{j_\ell}_{j_\ell}\cap(E^{i_{1}}\cup E^{i_{2}})\big|=0$.
\end{enumerate}
If $\sigma_2=0$, we have
\begin{enumerate}
 \item[2.1)] $E^{j_{\tau+2}}_{j_{\tau+2}}
 \subseteq(E^{i_{1}}\cup E^{i_{2}})$.
 \item[2.2)] For $\ell\in\{\tau+3,\tau+4,\tau+5\}$,
 $\big|E^{j_\ell}_{j_\ell}\cap(E^{i_{1}}\cup
 E^{i_{2}})\big|=O_q(1)$.
 \item[2.3)] For $\ell\in[\tau+6,m]$,
 $\big|E^{j_\ell}_{j_\ell}\cap(E^{i_{1}}\cup E^{i_{2}})\big|=0$.
\end{enumerate}

Now, we prove Claim 4. For each $\ell\in[\tau-1]$, since
$|\text{supp}(\bm x-\bm x')|=2$, by Observation 2, we have
$d_{\text{H}}(x_{[n]\backslash\{j_{\ell}\}},
x_{[n]\backslash\{j_{\ell}\}})=2$, so by Lemma
\ref{lem-sub-int-size} and \eqref{xxp-case2},
$\big|E^{j_\ell}_{j_\ell}\big|=2(q-1)n+q^2-6q+6$. Therefore, by
Claim 7, we can obtain
\begin{align}\label{eq-Lcup-size}
\left|\bigcup_{\ell\in[1,\tau-1]}E^{j_\ell}_{j_\ell}\right|
&=\sum_{\ell=1}^{\tau-1}\left|E^{j_\ell}_{j_{\ell}}\Big\backslash
\left(\bigcup_{\ell'=\ell+1}^{\tau-1}
E^{j_{\ell'}}_{j_{\ell'}}\right)\right|\nonumber\\
&=\sum_{\ell=1}^{\tau-3}\bigg|E^{j_\ell}_{j_{\ell}}\Big\backslash
\left(E^{j_{\ell+1}}_{j_{\ell+1}}\cup
E^{j_{\ell+2}}_{j_{\ell+2}}\right)\bigg|
+\big|E^{j_{\tau-2}}_{j_{\tau-2}}\big\backslash
E^{j_{\tau-1}}_{j_{\tau-1}}\big|+\big|E^{j_{\tau-1}}_{j_{\tau-1}}\big|\nonumber\\
&=(\tau-3)\big(2(q-1)n+q^2-8q+4\big)+\big(2(q-1)n+q^2-8q+6\big)
+\big(2(q-1)n+q^2-6q+6\big)\nonumber\\
&=2(q-1)(\tau-1)n+(q^2-8q+4)\tau+O_q(1).
\end{align}
Similarly, by Claim 7$'$, we can obtain
\begin{align}\label{eq-Rcup-size}
\left|\bigcup_{\ell\in[\tau+2,m]}E^{j_\ell}_{j_\ell}\right|
&=\sum_{\ell=\tau+2}^{m}\left|E^{j_\ell}_{j_{\ell}}\Big\backslash
\left(\bigcup_{\ell'=\tau+2}^{\ell-1}
E^{j_{\ell'}}_{j_{\ell'}}\right)\right|\nonumber\\
&=\sum_{\ell=\tau+4}^{m}\bigg|E^{j_\ell}_{j_{\ell}}\Big\backslash
\left(E^{j_{\ell-1}}_{j_{\ell-1}}\cup
E^{j_{\ell-2}}_{j_{\ell-2}}\right)\bigg|
+\big|E^{j_{\tau+3}}_{j_{\tau+3}}\big\backslash
E^{j_{\tau+2}}_{j_{\tau+2}}\big|+\big|E^{j_{\tau+2}}_{j_{\tau+2}}\big|\nonumber\\
&=(m-\tau-3)\big(2(q-1)n+q^2-8q+4\big)+\big(2(q-1)n+q^2-8q+6\big)
+\big(2(q-1)n+q^2-6q+6\big)\nonumber\\
&=2(q-1)(m-\tau-1)n+(q^2-8q+4)(m-\tau)+O_q(1).
\end{align}
Thus, we have the following arguments.
\begin{itemize}
 \item[a)] If $\tau>1$ and $\sigma_1>0$, then by 1.1)$-$1.3) of Claim 8, we can obtain
 $\big|\big(\bigcup_{\ell=1}^{\tau-1}E^{j_\ell}_{j_\ell}\big)\cap
 (E^{i_1}\cup E^{i_2})\big|=(q-1)n+O_q(1)$, so
 $\big|\big(\bigcup_{\ell=1}^{\tau-1}E^{j_\ell}_{j_\ell}\big)\backslash
 (E^{i_1}\cup E^{i_2})\big|=\big|\big(\bigcup_{\ell=1}^{\tau-1}
 E^{j_\ell}_{j_\ell}\big)\big|-(q-1)n+O_q(1)=(q-1)(2\tau-3)n+(q^2-8q+4)\tau+O_q(1)$.
 \item[b)] If $\tau>1$ and $\sigma_1=0$, then by 2.1)$-$2.3) of Claim 8,
 $\big|\big(\bigcup_{\ell=1}^{\tau-1}E^{j_\ell}_{j_\ell}\big)\cap
 (E^{i_1}\cup E^{i_2})\big|=\big|E^{j_{\tau-1}}_{j_{\tau-1}}\big|+O_q(1)=2(q-1)n+O_q(1)$, so
 $\big|\big(\bigcup_{\ell=1}^{\tau-1}E^{j_\ell}_{j_\ell}\big)\backslash
 (E^{i_1}\cup E^{i_2})\big|=\big|\big(\bigcup_{\ell=1}^{\tau-1}
 E^{j_\ell}_{j_\ell}\big)\big|-2(q-1)n+O_q(1)=(q-1)(2\tau-4)n+(q^2-8q+4)\tau+O_q(1)$.
 \item[c)] If $\tau<m-1$ and $\sigma_2>0$, then by 1.1)$-$1.3) of Claim
 8$'$, we can obtain
 $\big|\big(\bigcup_{\ell=\tau+2}^{m}E^{j_\ell}_{j_\ell}\big)\cap
 (E^{i_1}\cup E^{i_2})\big|=(q-1)n+O_q(1)$, so
 $\Big|\left(\bigcup_{\ell=\tau+2}^{m}E^{j_\ell}_{j_\ell}\right)\backslash
 (E^{i_1}\cup E^{i_2})\Big|=(q-1)(2m-2\tau-3)n+(q^2-8q+4)(m-\tau)+O_q(1)$.
 \item[d)] If $\tau<m-1$ and $\sigma_2=0$, then by 2.1)$-$2.3) of Claim
 8$'$, $\big|\big(\bigcup_{\ell=\tau+2}^{m}E^{j_\ell}_{j_\ell}\big)\cap
 (E^{i_1}\cup E^{i_2})\big|=\big|E^{j_{\tau+2}}_{j_{\tau+2}}\big|+O_q(1)=2(q-1)n+O_q(1)$, so
 $\Big|\left(\bigcup_{\ell=\tau+2}^{m}E^{j_\ell}_{j_\ell}\right)\backslash
 (E^{i_1}\cup E^{i_2})\Big|=(q-1)(2m-2\tau-4)n+(q^2-8q+4)(m-\tau)+O_q(1)$.
\end{itemize}
If $\tau>1$, then by the items a) and b), we can obtain
$\big|\big(\bigcup_{\ell=1}^{\tau-1}E^{j_\ell}_{j_\ell}\big)\backslash
(E^{i_1}\cup
E^{i_2})\big|=(q-1)(2\tau-4)n+(q^2-8q+4)\tau+1_{\sigma_1>0}(q-1)n+O_q(1)$,
where $1_{\sigma_1>0}=1$ if $\sigma_1>0$ and $1_{\sigma_1>0}=0$ if
$\sigma_1=0$. If $\tau<m-1$, then by the items c) and d), we can
obtain
$\Big|\left(\bigcup_{\ell=\tau+2}^{m}E^{j_\ell}_{j_\ell}\right)\backslash
(E^{i_1}\cup
E^{i_2})\Big|=(q-1)(2m-2\tau-4)n+(q^2-8q+4)(m-\tau)+1_{\sigma_2>0}(q-1)n+O_q(1)$,
where $1_{\sigma_2>0}=1$ if $\sigma_2>0$ and $1_{\sigma_2>0}=0$ if
$\sigma_2=0$.

We require $\tau>1$ in a) and b) because
$\bigcup_{\ell=1}^{\tau-1}E^{j_\ell}_{j_\ell}=\emptyset$ when
$\tau=1$. Similarly, we require $\tau<m-1$ in c) and d) because
$\bigcup_{\ell=\tau+2}^{m}E^{j_\ell}_{j_\ell}=\emptyset$ when
$\tau=m-1$.

\subsection{Proof of Claim 5}

Before proving Claim 5, we need to prove the following claims.

\textbf{Claim 9}: Let $\{j_i: i\in[m]\}$, $\{j'_{i'}: i'\in[m]\}$,
$\tau$ and $\tau'$ be defined as in Remark \ref{rem-form-run}. If
$\sigma_1>0$, we have
\begin{enumerate}
 \item[1.1)] For $\ell\in\{\tau-1,\tau-2\}$, we have
 $\big|E^{j_\ell}_{j_\ell}\cap E^{j_{\tau-1}}_{j'_{\tau'+2}}\big|=O_q(1)$ and
 $\big|E^{j_\ell}_{j_\ell}\cap
 E^{j_{\tau+2}}_{j'_{\tau'-1}}\big|=O_q(1)$.
 \item[1.2)] For $\ell\in[1,\tau-3]$, we have
 $\big|E^{j_\ell}_{j_\ell}\cap E^{j_{\tau-1}}_{j'_{\tau'+2}}\big|=0$ and
 $\big|E^{j_\ell}_{j_\ell}\cap E^{j_{\tau+2}}_{j'_{\tau'-1}}\big|=0$.
\end{enumerate}
If $\sigma_1=0$, we have
\begin{enumerate}
 \item[2.1)] For $\ell\in\{\tau-2,\tau-3\}$, we have
 $\big|E^{j_\ell}_{j_\ell}\cap E^{j_{\tau-1}}_{j'_{\tau'+2}}\big|=O_q(1)$ and
 $\big|E^{j_\ell}_{j_\ell}\cap E^{j_{\tau+2}}_{j'_{\tau'-1}}\big|=O_q(1)$.
 \item[2.2)] For $\ell\in[1,\tau-4]$, we have
 $\big|E^{j_\ell}_{j_\ell}\cap E^{j_{\tau-1}}_{j'_{\tau'+2}}\big|=0$ and
 $\big|E^{j_\ell}_{j_\ell}\cap E^{j_{\tau+2}}_{j'_{\tau'-1}}\big|=0$.
\end{enumerate}
\begin{proof}[Proof of Claim 9]
First, we suppose $\sigma_1>0$ and prove 1.1) and 1.2).

1.1) Since $\sigma_1>0$, by \eqref{form-x-xp-1} and Remark
\ref{rem-spc-Obv2}, we can find
$d_{\text{H}}\big(x'_{[n]\backslash\{j_{\tau-1}\}},
x'_{[n]\backslash\{j'_{\tau'+2}\}}\big)=4$ and
$d_{\text{H}}\big(x'_{[n]\backslash\{j_{\tau-2}\}},
x'_{[n]\backslash\{j'_{\tau'+2}\}}\big)=5$, so for
$\ell\in\{\tau-1,\tau-2\}$, we have
$\Big|E^{j_{\ell}}_{j_{\ell}}\cap
E^{j_{\tau-1}}_{j'_{\tau'+2}}\Big|
=\big|B^{\text{S}}_{2}\big(x_{[n]\backslash\{j_{\ell}\}},
x'_{[n]\backslash\{j_{\ell}\}}\big)\cap
B^{\text{S}}_{2}\big(x_{[n]\backslash\{j_{\tau-1}\}},
x'_{[n]\backslash\{j'_{\tau'+2}\}}\big)\big|\leq
\big|B^{\text{S}}_{2}\big(x'_{[n]\backslash\{j_{\ell}\}},
x'_{[n]\backslash\{j'_{\tau'+2}\}}\big)\big|=O_q(1)$, where the
last equality comes from Lemma \ref{lem-sub-int-size} and
\eqref{xxp-case2}. Similarly, by \eqref{eq-seqx-form} and Remark
\ref{rem-spc-Obv2}, we can find
$d_{\text{H}}\big(x_{[n]\backslash\{j_{\tau-1}\}},
x_{[n]\backslash\{j_{\tau+2}\}}\big)=3$ and
$d_{\text{H}}\big(x_{[n]\backslash\{j_{\tau-2}\}},
x_{[n]\backslash\{j_{\tau+2}\}}\big)=4$, so for
$\ell\in\{\tau-1,\tau-2\}$, we have
$\big|E^{j_{\ell}}_{j_{\ell}}\cap
E^{j_{\tau+2}}_{j'_{\tau'-1}}\big|
=\big|B^{\text{S}}_{2}\big(x_{[n]\backslash\{j_{\ell}\}},
x'_{[n]\backslash\{j_{\ell}\}}\big)\cap
B^{\text{S}}_{2}\big(x_{[n]\backslash\{j_{\tau+2}\}},
x'_{[n]\backslash\{j'_{\tau'-1}\}}\big)\big|\leq
\big|B^{\text{S}}_{2}\big(x_{[n]\backslash\{j_{\ell}\}},
x_{[n]\backslash\{j_{\tau+2}\}}\big)\big|=O_q(1)$.

1.2) The proof is similar to 1). For $\ell\in[1,\tau-3]$, we can
find $d_{\text{H}}\big(x'_{[n]\backslash\{j_{\ell}\}},
x'_{[n]\backslash\{j'_{\tau'+2}\}}\big)\geq 5~($by
\eqref{form-x-xp-1} and Remark \ref{rem-spc-Obv2}$)$ and
$d_{\text{H}}\big(x_{[n]\backslash\{j_{\ell}\}},
x_{[n]\backslash\{j_{\tau+2}\}}\big)\geq 5~($by
\eqref{eq-seqx-form} and Remark \ref{rem-spc-Obv2}$)$, so by Lemma
\ref{lem-sub-int-size}, $\Big|E^{j_{\ell}}_{j_{\ell}}\cap
E^{j_{\ell}}_{j'_{\tau'+2}}\Big|\leq
\big|B^{\text{S}}_{2}\big(x'_{[n]\backslash\{j_{\ell}\}},
x'_{[n]\backslash\{j'_{\tau'+2}\}}\big)\big|=0$ and
$\big|E^{j_{\ell}}_{j_{\ell}}\cap
E^{j_{\tau+2}}_{j'_{\tau'-1}}\big|\leq
\big|B^{\text{S}}_{2}\big(x_{[n]\backslash\{j_{\ell}\}},
x_{[n]\backslash\{j_{\tau+2}\}}\big)\big|=0$.

Next, we suppose $\sigma_1=0$ and prove 2.1) and 2.2). Since
$\sigma_1=0$, then \eqref{eq-seqx-form} becomes
\begin{align}\label{eq-seqx-form-sg0}\bm
x=\bar{x}_1^{k_1}~\cdots~\bar{x}_{\tau-1}^{k_{\tau-1}}~a~b~b^{\sigma_2}
~\bar{x}_{\tau+2}^{k_{\tau+2}}~\cdots~\bar{x}_{m}^{k_{m}}\end{align}
and \eqref{form-x-xp-1} becomes
\begin{align}\label{form-x-xp-1-sg0}
\bm x &=\bar{x}_1^{k_1}~\cdots~\bar{x}_{\tau-1}^{k_{\tau-1}}
~a~b~a^{\sigma'_2}~\tilde{x}_{\tau'+2}^{k'_{\tau'+2}}~\cdots
~\tilde{x}_{m'}^{k'_{m'}}\nonumber\\
\bm x'&=\bar{x}_1^{k_1}~\cdots~\bar{x}_{\tau-1}^{k_{\tau-1}}
~b~a~a^{\sigma'_2}~\tilde{x}_{\tau'+2}^{k'_{\tau'+2}}~\cdots
~\tilde{x}_{m'}^{k'_{m'}}
\end{align}
Note that $\bar{x}_{\tau-1}\neq a$ because
$\bar{x}_{\tau-1}^{k_{\tau-1}}$ and $a$ are two runs of $\bm x$.
However, it is not necessary that $\bar{x}_{\tau-1}\neq b$. We
prove the four statements separately.

2.1) By \eqref{form-x-xp-1-sg0} and Remark \ref{rem-spc-Obv2}, we
can find $d_{\text{H}}\big(x'_{[n]\backslash\{j_{\tau-1}\}},
x'_{[n]\backslash\{j'_{\tau'+2}\}}\big)\in\{3,4\}$ and
$d_{\text{H}}\big(x'_{[n]\backslash\{j_{\tau-2}\}},
x'_{[n]\backslash\{j'_{\tau'+2}\}}\big)\in\{4,5\}$,\footnote{In
fact, if $\bar{x}_{\tau-1}=b$, then
$d_{\text{H}}\big(x'_{[n]\backslash\{j_{\tau-1}\}},
x'_{[n]\backslash\{j'_{\tau'+2}\}}\big)=3$ and
$d_{\text{H}}\big(x'_{[n]\backslash\{j_{\tau-2}\}},
x'_{[n]\backslash\{j'_{\tau'+2}\}}\big)=4$; if
$\bar{x}_{\tau-1}\neq b$, then
$d_{\text{H}}\big(x'_{[n]\backslash\{j_{\tau-1}\}},
x'_{[n]\backslash\{j'_{\tau'+2}\}}\big)=4$ and
$d_{\text{H}}\big(x'_{[n]\backslash\{j_{\tau-2}\}},
x'_{[n]\backslash\{j'_{\tau'+2}\}}\big)=5$.} so for
$\ell\in\{\tau-2,\tau-3\}$, we have
$\Big|E^{j_{\ell}}_{j_{\ell}}\cap E^{j_{\ell}}_{j'_{\tau'+2}}\Big|
=\big|B^{\text{S}}_{2}\big(x_{[n]\backslash\{j_{\ell}\}},
x'_{[n]\backslash\{j_{\ell}\}}\big)\cap
B^{\text{S}}_{2}\big(x_{[n]\backslash\{j_{\ell}\}},
x'_{[n]\backslash\{j'_{\tau'+2}\}}\big)\big|\leq
\big|B^{\text{S}}_{2}\big(x'_{[n]\backslash\{j_{\ell}\}},
x'_{[n]\backslash\{j'_{\tau'+2}\}}\big)\big|=O_q(1)$, where the
last equality comes from Lemma \ref{lem-sub-int-size} and
\eqref{xxp-case2}. Similarly, by \eqref{eq-seqx-form-sg0} and
Remark \ref{rem-spc-Obv2}, we can find
$d_{\text{H}}\big(x_{[n]\backslash\{j_{\tau-1}\}},
x_{[n]\backslash\{j_{\tau+2}\}}\big)=4$ and
$d_{\text{H}}\big(x_{[n]\backslash\{j_{\tau-2}\}},
x_{[n]\backslash\{j_{\tau+2}\}}\big)=5$, so for
$\ell\in\{\tau-2,\tau-3\}$, we have
$\big|E^{j_{\ell}}_{j_{\ell}}\cap
E^{j_{\tau+2}}_{j'_{\tau'-1}}\big|
=\big|B^{\text{S}}_{2}\big(x_{[n]\backslash\{j_{\ell}\}},
x'_{[n]\backslash\{j_{\ell}\}}\big)\cap
B^{\text{S}}_{2}\big(x_{[n]\backslash\{j_{\tau+2}\}},
x'_{[n]\backslash\{j'_{\tau'-1}\}}\big)\big|\leq
\big|B^{\text{S}}_{2}\big(x_{[n]\backslash\{j_{\ell}\}},
x_{[n]\backslash\{j_{\tau+2}\}}\big)\big|=O_q(1)$.

2.2) The proof is similar to 1). For $\ell\in[1,\tau-4]$, we can
find $d_{\text{H}}\big(x'_{[n]\backslash\{j_{\ell}\}},
x'_{[n]\backslash\{j'_{\tau'+2}\}}\big)\geq 5~($by
\eqref{form-x-xp-1-sg0} and Remark \ref{rem-spc-Obv2}$)$ and
$d_{\text{H}}\big(x_{[n]\backslash\{j_{\ell}\}},
x_{[n]\backslash\{j_{\tau+2}\}}\big)\geq 5~($by
\eqref{eq-seqx-form-sg0} and Remark \ref{rem-spc-Obv2}$)$, so by
Lemma \ref{lem-sub-int-size}, $\Big|E^{j_{\ell}}_{j_{\ell}}\cap
E^{j_{\ell}}_{j'_{\tau'+2}}\Big|\leq
\big|B^{\text{S}}_{2}\big(x'_{[n]\backslash\{j_{\ell}\}},
x'_{[n]\backslash\{j'_{\tau'+2}\}}\big)\big|=0$ and
$\big|E^{j_{\ell}}_{j_{\ell}}\cap
E^{j_{\tau+2}}_{j'_{\tau'-1}}\big|\leq
\big|B^{\text{S}}_{2}\big(x_{[n]\backslash\{j_{\ell}\}},
x_{[n]\backslash\{j_{\tau+2}\}}\big)\big|=0$.
\end{proof}

For $\ell\in[\tau+2,m]$, we have the following Claim 9$'$, which
can be proved by the same discussions as that of Claim 9.

\textbf{Claim 9$'$}: Let $\{j_i: i\in[m]\}$, $\{j'_{i'}:
i'\in[m]\}$, $\tau$ and $\tau'$ be defined as in Remark
\ref{rem-form-run}. If $\sigma_2>0$, we have
\begin{enumerate}
 \item[1.1)] For $\ell\in\{\tau+2,\tau+3\}$, we have
 $\big|E^{j_\ell}_{j_\ell}\cap E^{j_{\tau-1}}_{j'_{\tau'+2}}\big|=O_q(1)$ and
 $\big|E^{j_\ell}_{j_\ell}\cap E^{j_{\tau+2}}_{j'_{\tau'-1}}\big|=O_q(1)$.
 \item[1.2)] For $\ell\in[\tau+4,m]$, we have
 $\big|E^{j_\ell}_{j_\ell}\cap E^{j_{\tau-1}}_{j'_{\tau'+2}}\big|=0$ and
 $\big|E^{j_\ell}_{j_\ell}\cap E^{j_{\tau+2}}_{j'_{\tau'-1}}\big|=0$.
\end{enumerate}
If $\sigma_2=0$, we have
\begin{enumerate}
 \item[2.1)] For $\ell\in\{\tau+3,\tau+4\}$, we have
 $\big|E^{j_\ell}_{j_\ell}\cap E^{j_{\tau-1}}_{j'_{\tau'+2}}\big|=O_q(1)$ and
 $\big|E^{j_\ell}_{j_\ell}\cap E^{j_{\tau+2}}_{j'_{\tau'-1}}\big|=O_q(1)$.
 \item[2.2)] For $\ell\in[\tau+5,m]$, we have
 $\big|E^{j_\ell}_{j_\ell}\cap E^{j_{\tau-1}}_{j'_{\tau'+2}}\big|=0$ and
 $\big|E^{j_\ell}_{j_\ell}\cap E^{j_{\tau+2}}_{j'_{\tau'-1}}\big|=0$.
\end{enumerate}

\textbf{Claim 10}: Let $\{j_i: i\in[m]\}$, $\{j'_{i'}:
i'\in[m]\}$, $\tau$ and $\tau'$ be defined as in Remark
\ref{rem-form-run}. The following statements hold.
\begin{itemize}
 \item[1)] $\Big|E^{j_{\tau-1}}_{j'_{\tau'+2}}\cap(E^{i_1}\cup
 E^{i_2})\Big|=(q-1)n+O_q(1)$ and
 $\Big|E^{j_{\tau+2}}_{j'_{\tau'-1}}\cap(E^{i_1}\cup
 E^{i_2})\Big|=(q-1)n+O_q(1)$.\vspace{3pt}
 \item[2)] $\Big|E^{j_{\tau-1}}_{j'_{\tau'+2}}\cap E^{j_{\tau+2}}_{j'_{\tau'-1}}\Big|
 =O_q(1)$.
\end{itemize}
\begin{proof}[Proof of Claim 10]
We prove the two statements separately.

1) We only prove the result for
$\Big|E^{j_{\tau-1}}_{j'_{\tau'+2}}\cap(E^{i_1}\cup
E^{i_2})\Big|$, because the proof for
$\Big|E^{j_{\tau+2}}_{j'_{\tau'-1}}\cap(E^{i_1}\cup E^{i_2})\Big|$
is similar. By \eqref{form-x-xp-1}, we can obtain the following
\eqref{eq1-jmjm1-i1i2-sgn0}, \eqref{eq2-jmjm1-i1i2-sgn0},
\eqref{eq6-jmjm1-i1i2-sgn0} and \eqref{eq7-jmjm1-i1i2-sgn0}.
Suppose $\bm v\in
E^{j_{\tau-1}}_{j'_{\tau'+2}}=B^{\text{S}}_{2}\big(x_{[n]\backslash\{j_{\tau-1}\}},
x'_{[n]\backslash\{j'_{\tau'+2}\}}\big)$. Then by 2) of Lemma
\ref{lem-Cap-SE}, we can obtain \eqref{eq4-jmjm1-i1i2-sgn0} such
that: (i) $[\cdots~\bar{x}_{\tau-1}^{k_{\tau-1}-1}][a^{\sigma_1}~
b~a^{\sigma'_2}][\tilde{x}_{\tau'+2}^{k'_{\tau'+2}-1}\cdots]
=\cdots~\bar{x}_{\tau-1}^{k_{\tau-1}-1}~a^{\sigma_1}~
b~a^{\sigma'_2}~\tilde{x}_{\tau'+2}^{k'_{\tau'+2}-1}\cdots$ and
$v\!~v'\in\Sigma_q^2$; or (ii)
$[\cdots~\bar{x}_{\tau-1}^{k_{\tau-1}-1}][a^{\sigma_1}~
b~a^{\sigma'_2}][\tilde{x}_{\tau'+2}^{k'_{\tau'+2}-1}\cdots]$ is
obtained from
$\cdots~\bar{x}_{\tau-1}^{k_{\tau-1}-1}~a^{\sigma_1}~
b~a^{\sigma'_2}~\tilde{x}_{\tau'+2}^{k'_{\tau'+2}-1}\cdots$ by
exactly one substitution and $vv'\in\{aa,
\bar{x}_{\tau-1}\bar{x}_{\tau'+2}\}$.
\begin{align}
x_{[n]\backslash\{j_{\tau-1}\}}
&=\!~\cdots~\bar{x}_{\tau-1}^{k_{\tau-1}-1}~~~~~a~~~~a^{\sigma_1}\!~\!~
b~a^{\sigma'_2}\!~~\!~\tilde{x}_{\tau'+2}
~\!~~\tilde{x}_{\tau'+2}^{k'_{\tau'+2}-1}~\cdots\label{eq1-jmjm1-i1i2-sgn0}\\
x'_{[n]\backslash\{j'_{\tau'+2}\}}
&=\!~\cdots~\bar{x}_{\tau-1}^{k_{\tau-1}-1}~~\bar{x}_{\tau-1}\!~~~a^{\sigma_1}
\!~\!~b~a^{\sigma'_2}~~~~~a~~~~
~\tilde{x}_{\tau'+2}^{k'_{\tau'+2}-1}~\cdots\label{eq2-jmjm1-i1i2-sgn0}\\
\bm v~~
&=[\cdots~\bar{x}_{\tau-1}^{k_{\tau-1}-1}]~\!~~~~v~\!~\!~~[a^{\sigma_1}~
b~a^{\sigma'_2}]~~~~v'~\!~~
~[\tilde{x}_{\tau'+2}^{k'_{\tau'+2}-1}~\cdots]\label{eq4-jmjm1-i1i2-sgn0}\\
x_{[n]\backslash\{i_1\}}
&=\!~\cdots~\bar{x}_{\tau-1}^{k_{\tau-1}-1}~~\bar{x}_{\tau-1}~\!~~a^{\sigma_1}
\!~\!~b~a^{\sigma'_2}~\!~\!~ \tilde{x}_{\tau'+2}~
~\!~\tilde{x}_{\tau'+2}^{k'_{\tau'+2}-1}~\cdots\label{eq6-jmjm1-i1i2-sgn0}\\
x_{[n]\backslash\{i_2\}}
&=\!~\cdots~\bar{x}_{\tau-1}^{k_{\tau-1}-1}~~\bar{x}_{\tau-1}\!~~~a^{\sigma_1}
\!~\!~a~a^{\sigma'_2}\!~~\tilde{x}_{\tau'+2}
~\!~~\tilde{x}_{\tau'+2}^{k'_{\tau'+2}-1}~\cdots\label{eq7-jmjm1-i1i2-sgn0}
\end{align}
Note that by \eqref{form-x-xp-1}, we have $a\neq \bar{x}_{\tau-1}$
and $a\neq \tilde{x}_{\tau'+2}$. Then we can see that: a) if $\bm
v$ satisfies condition (i), then $d_{\text{H}}\big(\bm v,
x_{[n]\backslash\{i_1\}}\big)\leq 2$ and so $\bm v\in E^{i_1}$; b)
if $\bm v$ satisfies (ii) and $vv'=aa$, then by comparing
\eqref{eq4-jmjm1-i1i2-sgn0} with \eqref{eq6-jmjm1-i1i2-sgn0} and
\eqref{eq7-jmjm1-i1i2-sgn0}, we have $\bm v\in(E^{i_1}\cup
E^{i_2})$ if and only if $[\cdots\bar{x}_{m}^{k_m-1}][a^{\sigma_1}
ba^{\sigma'_2}][\tilde{x}_{j'_{m_1+1}}^{k'_{m_1+1}-1}\cdots]
=\cdots\bar{x}_{m}^{k_m-1}a^{\sigma_1}
a\!~a^{\sigma'_2}\tilde{x}_{j'_{m_1+1}}^{k'_{m_1+1}-1}\cdots$; c)
if $\bm v$ satisfies (ii) and $vv'=
\bar{x}_{\tau-1}\bar{x}_{\tau'+2}$, then $d_{\text{H}}\big(\bm v,
x_{[n]\backslash\{i_2\}}\big)=1$, so $\bm v\in E^{i_2}$ and the
number of such $\bm v$ is $(n-3)(q-1)$. Hence, we have
$\big|E^{j_m}_{j'_{m_1+1}}\cap(E^{i_1}\cup
E^{i_2})\big|=q^2+1+(n-3)(q-1)=(q-1)n+O_q(1)$.

2) Since $\bm
x=\bar{x}_1^{k_1}\!~\cdots\!~\bar{x}_{\tau-1}^{k_{\tau-1}}\!~a^{\sigma_1}
\!~a\!~b\!~b^{\sigma_2}
\!~\bar{x}_{\tau+2}^{k_{\tau+2}}\!~\cdots\!~\bar{x}_{m}^{k_{m}}~($see
Remark \ref{rem-form-run}$)$, then by Observation 2, we have
$d_{\text{H}}\big(x_{[n]\backslash\{j_{\tau-1}\}},
x_{[n]\backslash\{j_{\tau+2}\}}\big)=3$, and so by Lemma
\ref{lem-sub-int-size},
$\big|B^{\text{S}}_{2}\big(x_{[n]\backslash\{j_{\tau-1}\}},
x_{[n]\backslash\{j_{\tau+2}\}}\big)\big|=O_q(1)$. Moreover, we
have $E^{j_{\tau-1}}_{j'_{\tau'+2}}\cap
E^{j_{\tau+2}}_{j'_{\tau'-1}}=B^{\text{S}}_{2}\big(x_{[n]\backslash\{j_{\tau-1}\}},
x'_{[n]\backslash\{j'_{\tau'+2}\}}\big)\cap
B^{\text{S}}_{2}\big(x_{[n]\backslash\{j_{\tau+2}\}},
x'_{[n]\backslash\{j'_{\tau'-1}\}}\big)\subseteq
B^{\text{S}}_{2}\big(x_{[n]\backslash\{j_{\tau-1}\}},
x_{[n]\backslash\{j_{\tau+2}\}}\big)$. Hence,
$\Big|E^{j_{\tau-1}}_{j'_{\tau'+2}}\cap
E^{j_{\tau+2}}_{j'_{\tau'-1}}\Big|\leq
\big|B^{\text{S}}_{2}\big(x_{[n]\backslash\{j_{\tau-1}\}},
x_{[n]\backslash\{j_{\tau+2}\}}\big)\big|=O_q(1)$.
\end{proof}

Now, we can prove Claim 5. By \eqref{eq1-jmjm1-i1i2-sgn0} and
\eqref{eq2-jmjm1-i1i2-sgn0}, we can find
$d_{\text{H}}\big(x_{[n]\backslash\{j_{\tau-1}\}},
x'_{[n]\backslash\{j'_{\tau'+2}\}}\big)=2$, so by Lemma
\ref{lem-sub-int-size} and \eqref{xxp-case2},
$$\Big|E^{j_{\tau-1}}_{j'_{\tau'+2}}\Big|=\xi^{q,n-1}_{2,2}=
2(q-1)n+O_q(1).$$ By Claim 9 and Claim 9$'$, we can obtain
$\Big|E^{j_{\tau-1}}_{j'_{\tau'+2}}\cap\Big(\bigcup_{\ell=1}^{\tau-1}E^{j_\ell}_{j_\ell}\Big)
\Big|=O_q(1)$ and
$\Big|E^{j_{\tau-1}}_{j'_{\tau'+2}}\cap\Big(\bigcup_{\ell=\tau+2}^{m}E^{j_\ell}_{j_\ell}\Big)
\Big|=O_q(1)$. Moreover, by 1) of Claim 10, we have
$\Big|E^{j_{\tau-1}}_{j'_{\tau'+2}}\cap(E^{i_1}\cup
E^{i_2})\Big|=(q-1)n+O_q(1)$. So, we can obtain
$$\Big|E^{j_{\tau-1}}_{j'_{\tau'+2}}\cap\Omega
\Big|=(q-1)n+O_q(1)$$ where
$\Omega=\Big(\bigcup_{\ell=1}^{\tau-1}E^{j_\ell}_{j_\ell}\Big)\cup
\Big(\bigcup_{\ell=\tau+2}^{m}E^{j_\ell}_{j_\ell}\Big)\cup(E^{i_1}\cup
E^{i_2})$ as defined in Claim 5. Therefore, we have
$\Big|E^{j_{\tau-1}}_{j'_{\tau'+2}}\cup\Omega
\Big|=\Big|E^{j_{\tau-1}}_{j'_{\tau'+2}}\Big|+|\Omega|-(q-1)n+O_q(1)=|\Omega|+(q-1)n+O_q(1).$
We require $\tau>1$ and $\tau'<m'-1$ because
$E^{j_{\tau-1}}_{j'_{\tau'+2}}=\emptyset$ when $\tau=1$ or
$\tau'=m'-1$. Similarly, we can obtain
$\Big|E^{j_{\tau+2}}_{j'_{\tau'-1}}\Big|=\xi^{q,n-1}_{2,2}=
2(q-1)n+O_q(1)$ and $$\Big|E^{j_{\tau+2}}_{j'_{\tau'-1}}\cap\Omega
\Big|=(q-1)n+O_q(1).$$ Therefore,
$\Big|E^{j_{\tau+2}}_{j'_{\tau'-1}}\cup\Omega
\Big|=\Big|E^{j_{\tau-1}}_{j'_{\tau'+2}}\Big|+|\Omega|-(q-1)n+O_q(1)=|\Omega|+(q-1)n+O_q(1).$
We require $\tau<m-1$ and $\tau'>1$ because
$E^{j_{\tau+2}}_{j'_{\tau'-1}}=\emptyset$ when $\tau=m-1$ or
$\tau'=1$.

Note that by 2) of Claim 10, we have
$\Big|E^{j_{\tau-1}}_{j'_{\tau'+2}}\cap
E^{j_{\tau+2}}_{j'_{\tau'-1}}\Big|=O_q(1)$, so
$\Big|E^{j_{\tau-1}}_{j'_{\tau'+2}}\cap
E^{j_{\tau+2}}_{j'_{\tau'-1}}\cap\Omega\Big|=O_q(1)$ and we can
obtain $\Big|E^{j_{\tau-1}}_{j'_{\tau'+2}}\cup
E^{j_{\tau+2}}_{j'_{\tau'-1}}\cup\Omega\Big|
=\Big|E^{j_{\tau-1}}_{j'_{\tau'+2}}\Big|+\Big|E^{j_{\tau+2}}_{j'_{\tau'-1}}\Big|+
|\Omega|-\Big|E^{j_{\tau-1}}_{j'_{\tau'+2}}\cap
E^{j_{\tau+2}}_{j'_{\tau'-1}}\Big|-\Big|E^{j_{\tau-1}}_{j'_{\tau'+2}}\cap\Omega\Big|
-\Big|E^{j_{\tau+2}}_{j'_{\tau'-1}}\cap\Omega\Big|+\Big|E^{j_{\tau-1}}_{j'_{\tau'+2}}\cap
E^{j_{\tau+2}}_{j'_{\tau'-1}}\cap\Omega\Big|
=|\Omega|+2(q-1)n+O_q(1),$ which completes the proof.



\end{document}